\documentclass{article}

\usepackage{arxiv}

\usepackage[utf8]{inputenc} 
\usepackage[T1]{fontenc}    
\usepackage{hyperref}       
\usepackage{url}            
\usepackage{booktabs}       
\usepackage{amsfonts}       
\usepackage{nicefrac}       
\usepackage{microtype}      
\usepackage{graphicx}

\usepackage[round]{natbib}

\hypersetup{
  pdftitle={On Solving Cycle Problems with Branch-and-Cut: Extending Shrinking and Exact Subcycle Elimination Separation Algorithms},
  pdfsubject={math.CO, cs.DM},
  pdfauthor={G. Kobeaga, M. Merino, J.A. Lozano},
  pdfkeywords={cycle problem, branch-and-cut, shrinking, exact separation, subcycle elimination, gomory-hu tree}
}

\date{August 3, 2021}

\usepackage{enumerate}
\usepackage{amsopn, amstext, amscd, amsfonts, amssymb, amsthm}
\usepackage{url, hyperref, verbatim}
\usepackage{etex}

\usepackage{booktabs}
\usepackage{afterpage}
\usepackage{pdflscape}
\usepackage{multirow}

\usepackage{amssymb, amsmath}

\usepackage[ruled, lined, boxed, linesnumbered]{algorithm2e}

\usepackage{setspace}

\newtheorem{theorem}{Theorem}[section]
\newtheorem{definition}[theorem]{Definition}

\newtheorem{lemma}[theorem]{Lemma}
\newtheorem{corollary}[theorem]{Corollary}
\newtheorem{rule-thm}[theorem]{Rule}

\newtheorem{conjecture}[theorem]{Conjecture}

\title{On Solving Cycle Problems with Branch-and-Cut: Extending Shrinking and Exact Subcycle Elimination Separation Algorithms}

\author{ \href{https://orcid.org/0000-0002-8669-4482}{\includegraphics[scale=0.06]{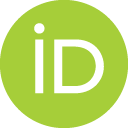}\hspace{1mm}Gorka Kobeaga} \\
  Basque Center for Applied Mathematics BCAM\\
  \texttt{gkobeaga@bcamath.org} \\
  \And
  \href{https://orcid.org/0000-0002-4947-2784}{\includegraphics[scale=0.06]{orcid.png}\hspace{1mm}Mar\'ia Merino} \\
  University of the Basque Country UPV/EHU\\
  \texttt{maria.merino@ehu.eus} \\
  \And
  \href{https://orcid.org/0000-0002-4683-8111}{\includegraphics[scale=0.06]{orcid.png}\hspace{1mm}Jose A. Lozano} \\
  Basque Center for Applied Mathematics BCAM\\
  University of the Basque Country UPV/EHU\\
  \texttt{jlozano@bcamath.org} \\
}

\begin{document}
\maketitle

\begin{abstract}
  In this paper, we extend techniques developed in the context of the Travelling Salesperson Problem for cycle problems. Particularly, we study the shrinking of support graphs and the exact algorithms for subcycle elimination separation problems. The efficient application of the considered techniques has proved to be essential in the Travelling Salesperson Problem when solving large size problems by Branch-and-Cut, and this has been the motivation behind this work. Regarding the shrinking of support graphs, we prove the validity of the Padberg-Rinaldi general shrinking rules and the Crowder-Padberg subcycle-safe shrinking rules. Concerning the subcycle separation problems, we extend two exact separation algorithms, the Dynamic Hong and the Extended Padberg-Gr{\"o}tschel algorithms, which are shown to be superior to the ones used so far in the literature of cycle problems.

  The proposed techniques are empirically tested in 24 subcycle elimination problem instances generated by solving the Orienteering Problem (involving up to 15112 vertices) with Branch-and-Cut. The experiments
  suggest the relevance of the proposed techniques for cycle problems.
  The obtained average speedup for the subcycle separation problems in the Orienteering Problem when the proposed techniques are used together is around 50 times in medium-sized instances and around 250 times in large-sized instances.
\end{abstract}


\begin{section}{Introduction}

  The Travelling Salesperson Problem (TSP) has been the source and the testbed of the most important techniques developed for the exact solution of combinational optimization problems. These techniques have been principally developed in the context of the Branch-and-Cut (B\&C) algorithm, which combines the Branch-and-Bound (B\&B) and the cutting-planes methods, see~\cite{concorde} for an historical overview. Eventually, many of these techniques have been successfully adapted to other related problems. However, there are procedures, such as the support graph shrinking and some separation algorithms, that are strongly dependent on the problem peculiarities. As a consequence, these techniques might not have been adapted yet, or there might still be room for further improvements.

  As TSP is the most well-known cycle problem, we motivate the goals of this paper focusing on this problem. When a B\&B algorithm is used to exactly solve the TSP, which is an Integer Program (IP), the cutting-planes method arises as a natural strategy to handle at least two situations: the exponential number of constraints of the model and the consequences of the linear relaxation of the integer problem. Recall that in a B\&B algorithm the branching decisions are made guided by a sequence of Linear Program (LP). These LPs are principally obtained by relaxing the integrality and fixing the variables according to the preceding branching decisions.

  Within this approach, the cutting-planes method is required due to the fact that, in order to define a TSP model, an exponential number of constraints in terms of the number of vertices in the TSP is needed, see~\cite{Padberg1991}. In order to deal with this situation, the exact algorithm is initialized with a subproblem of the LP, let us call this $LP_0$, that considers a controlled number of constraints. During the algorithm, the excluded constraints are added to $LP_0$ only if they are required, i.e.,~if they are violated by the solution of the $LP_0$.
  The second reason to consider the cutting-planes method is that since the variables in the linear relaxation of the TSP are considered continuous instead of integers, new families of valid inequalities arise (inequalities that are satisfied by all the cycles), also called cuts, that are not linear combinations of the constraints defining the TSP\@. Since the number of branch nodes needed to visit by the algorithm is reduced, the cutting-planes are very valuable to decrease the solving time of a B\&B algorithm.

  Computationally, the most expensive part of the cutting-planes method is to solve the separation problems. Given a solution of the $LP_0$ and an inequality family, the separation problem for the given family consists of finding either the violated inequalities of the family or a certificate that no violated inequality of the family exists.

  The difficulty of efficiently solving the separation problems becomes evident when the number of vertices of the problem increases. It is well known that, in practice, even a polynomial time separation algorithm might turn out to be inefficient for certain families. To mitigate this practical issue, a technique known as shrinking has been exploited in the TSP, see~\cite{crowder1980, Padberg1990b, Grotschel1991}. Shrinking consists of safely simplifying, i.e.,\@ without losing all the violated inequalities of the family, the support graph generated by the solution of the $LP_0$. This way, considering that, generally, the separation is harder than the shrinking, the cost of finding the violated inequalities is reduced because the separation is performed in a graph involving a lower number of vertices and edges.
  In Figure~\ref{fig:scheme}, a flowchart of a generic B\&C algorithm and the separation algorithm with and without the shrinking.
  \begin{figure}[htb!]
    \centering
    \includegraphics[width=\columnwidth]{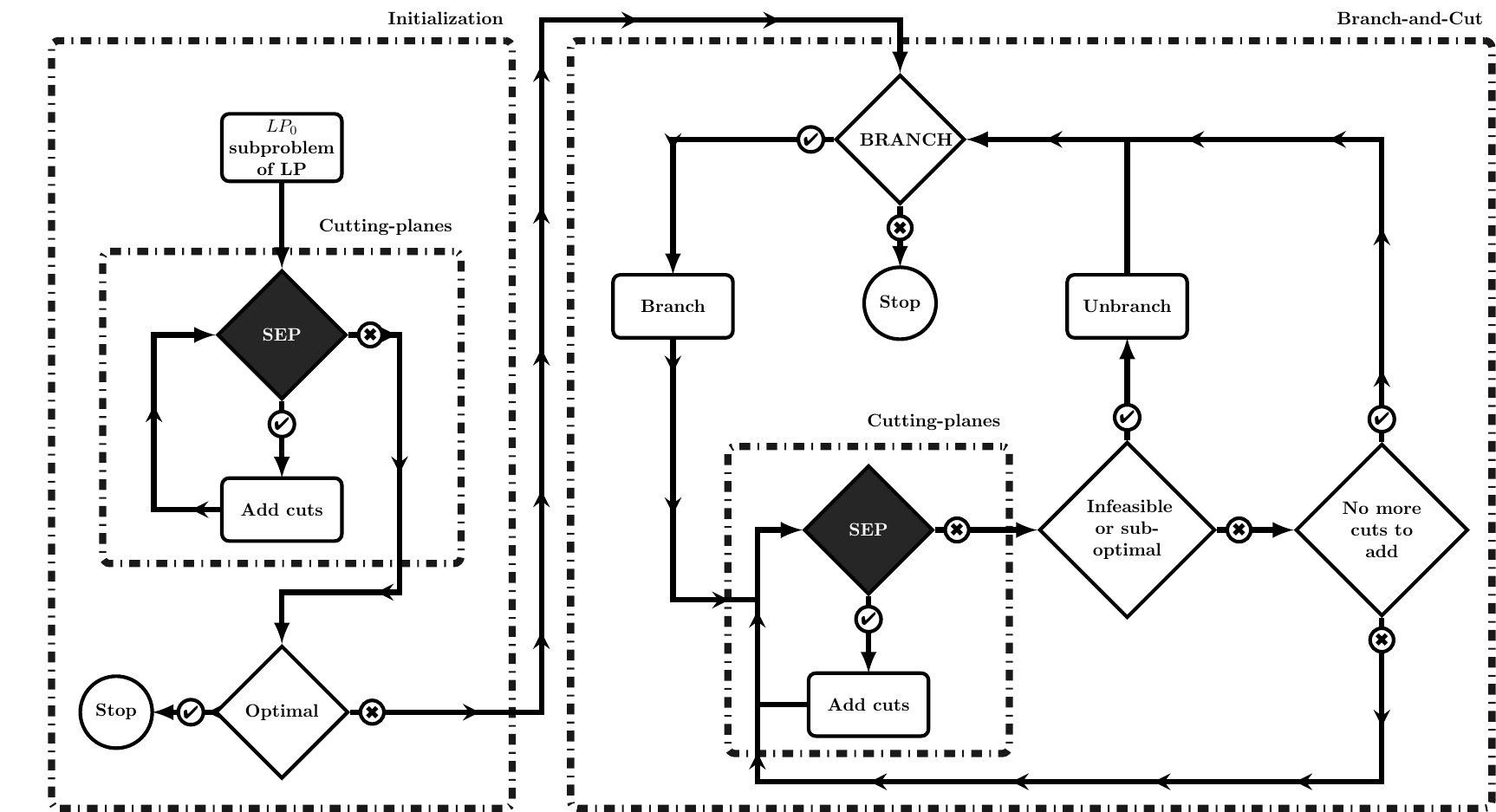}\\\vspace*{2mm}
    \includegraphics[width=0.6\columnwidth]{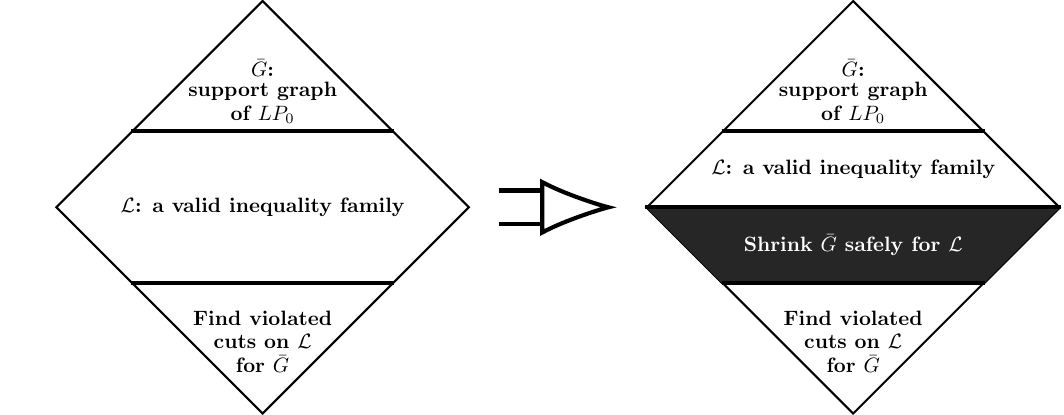}
    \caption{In the top, a flowchart of a generic Branch-and-Cut algorithm. BRANCH is an oracle which returns an unevaluated node in the branching tree. At each action box of the flowchart the subproblem $LP_0$ is updated and solved. In the bottom, the detailed separation algorithm (SEP) without and with shrinking.}\label{fig:scheme}
  \end{figure}

  In the last few decades, many optimization problems have proliferated whose solution is required to be a cycle, but not necessarily Hamiltonian as in the TSP\@. This is the case for some extensions of the TSP itself, as can be seen in the extensive collection about TSP variants of~\cite{Gutin07}. For instance, the weighted girth problem, consists of finding the minimum cost cycle in a weighted graph, see~\cite{Coullard1989} and~\cite{bauer1997}. Cycles are also the solutions of the Generalized TSP (GTSP) where the vertices are labeled in clusters and at least one vertex of each cluster is required to be visited, but not all the vertices, see~\cite{GTSP1995}. Other routing problems, which are recently gaining popularity because of their wide range of applications, are the TSP with profits, see~\cite{feillet05} and~\cite{speranza2014}. These problems are the Profitable Tour Problem (PTP), the Orienteering Problem (OP), the Prize Collecting TSP (PCTSP), and their variations. From the TSP with profits, the OP, which consists of finding the cycle that maximizes the collected vertex profits subject to a cycle length constraint, is the one which has been most extensively studied. For a recent book on applications and variants of the OP see~\cite{Vansteenwegen2019}.

  This work has three main aims: first, to generalize the shrinking rules (global and subcycle specific) proposed in the literature of the TSP to the case of cycle problems; second, to extend in an effective manner the subcycle exact separation algorithms for cycle problems; and third, to show experimentally the relevance of the proposed shrinking rules and separation algorithms. On the one hand, 6 different shrinking rules for cycle problems are presented in this work, of which three are safe for all the valid inequalities and three are specifically safe for subcycle elimination constraints. On the other hand, we extend two exact separation algorithms proposed in~\cite{Padberg1985} and~\cite{Padberg1990b}. We empirically show the contribution of the shrinking and separation strategies in the time reduction and in the generation of violated subcycle elimination constraints. For the experiments, we have used 24 instances of the subcycle separation problem generated in the solution of OP by B\&C with up to 15112 number of vertices. The results show that the speedup of using the combination of the proposed shrinking and separation techniques is around 50 times in medium-sized instances and 200 times in large-sized instances.

  This paper is structured as follows. In Section~\ref{sec:cp}, we introduce the cycle polytope and other related polytopes used in this work. In Section~\ref{sec:srk}, we study the safe shrinking rules for the cycle polytope. Section~\ref{sec:gsec} includes rules that are particularly safe for Subcycle Elimination Constraints (SEC). In Section~\ref{sec:sep-alg}, two exact separation algorithms of SECs for cycle problems are presented. Finally, in Section~\ref{sec:exp}, we discuss the computational experiments for the different separation algorithms for SECs. Appendices~\ref{appendix:pseudo},~\ref{appendix:detailed} and~\ref{appendix:figures} of this work are available, which contain pseudocodes of the shrinking and separation algorithms as well as detailed computational results and figures to illustrate the shrinking techniques. Additionally, we have released the source code of the implementations used for the computational experiments.

\end{section}

\begin{section}{The Cycle Polytope}\label{sec:cp}

  Let $G=(V, E)$ be an undirected graph with no loops. Let us define the following sets:
  \begin{subequations}
    \begin{align}
      (Q:W) & : = \{[u, v] \in E:  u \in Q, v \in W\} & Q, W \subset V \\
      \delta(Q) & : = (Q:V-Q) & Q \subset V \\
      E(Q) & : = (Q:Q) & Q \subset V  \\
      V(T) & : = \{v \in V:  T \cap (v : V) \neq \emptyset\} & T \subset E \\
      N(Q) & : = V(\delta(Q))-Q & Q \subset V
    \end{align}
  \end{subequations}
  \noindent where  $(Q:W)$ are the edges connecting $Q$ and $W$, $\delta(Q)$ is the set of edges in the coboundary of $Q$ also known as the star-set of $Q$, $E(Q)$ is the set of edges between the vertices of $Q$, $V(T)$ is the set of vertices incident with an edge set $T$, and $N(Q)$ are the neighbour vertices set of $Q$. For simplicity, we sometimes denote $\{e\}$ and $\{v\}$ by $e$ and $v$, respectively, e.g., $\delta(v)$ and $V(e)$.

  We denote by $\mathbb{R}^{V}$ and $\mathbb{R}^{E}$ the space of real vectors whose components are indexed by elements of $V$ and $E$, respectively. With every subset $T \subset E$ we associate a vector ${(y, x)}^T=(y^T, x^T)$ called the characteristic vector of $T$, defined as follows:
  \begin{equation}
    y^T_v : = \left\{\begin{aligned}
        1 & \quad \; \mbox{if $v \in V(T)$} \\
        0 & \quad \; \mbox{otherwise}
    \end{aligned} \right.
    \quad
    x^T_{e} : = \left\{\begin{aligned}
        1 & \quad \;\mbox{if $e \in T$ } \\
        0 & \quad \; \mbox{otherwise}
    \end{aligned} \right.
  \end{equation}
  When $y_v^T = 1$, i.e. $v \in V(T)$, we say that the vertex $v$ is visited by the edge set $T$.

  We denote by $\mathcal{C}_G$ the set of (simple) cycles of the graph $G$. We assume that every cycle $\tau \in \mathcal{C}_G$ is represented as a subset of edges. Then, the cycle polytope $P^G_C$ of the graph $G$ is the convex hull of the characteristic vectors of all the cycles of the graph:
  \begin{equation}
    P^G_C : = conv\{{(y, x)}^{\tau} \in \mathbb{R}^{V \times E}: \tau \in \mathcal{C}_G\}
  \end{equation}

  By definition, a vector $(y, x)$ belongs to $P^G_C$ if it is a convex combination of cycles of $\mathcal{C}_G$, i.e., $(y, x)\in P^G_C$ if and only if there exists a set of real numbers ${\{\lambda_{\tau}\}}_{\tau \in \mathcal{C}_G}$ such that
  \begin{equation}
    (y, x) = \sum_{\tau \in \mathcal{C}_G} \lambda_{\tau} {(y, x)}^\tau
  \end{equation}
  $\lambda_{\tau} \geq 0$ for every $\tau \in \mathcal{C}_G$ and $\sum_{\tau \in \mathcal{C}_G} \lambda_\tau = 1$.

  Similarly, we denote by $\mathcal{T}_G$ the set of tours, i.e., Hamiltonian cycles, of the graph $G$, and by $P^G_{TSP}$ the TSP polytope of the graph $G$. The $P^G_{TSP}$ is the convex hull of the characteristic vectors of all the tours of the graph:
  \begin{equation}
    P^G_{TSP} : = conv\{{(y, x)}^{\tau} \in \mathbb{R}^{V \times E}: \tau \in \mathcal{T}_G\}
  \end{equation}
  Note that, $y = 1$ is satisfied by every $(y, x) \in P^G_{TSP}$. Since, the tours form a subset of cycles of $G$, we have that:
  \begin{equation}
    P^G_{TSP} \subset P^G_C
  \end{equation}

  In order to use Linear Programming based techniques such as the B\&C algorithm, the polytope $P^G_C$ must be characterized by means of a system of linear constraints. A complete characterization of the integer points of $P^G_C$ using only edge variables was given in~\cite{bauer1997}. In this work, since we find it more convenient to formulate the shrinking rules of Section~\ref{sec:srk} and Section~\ref{sec:gsec}, we consider an equivalent one which uses the vertex and edge variables for the characterization. For $(y, x) \in \mathbb{R}^{V \times E} $, $S\subset V$ and $T\subset E$, we define $y(S) = \sum_{v \in S} y_v$ and $x(T) = \sum_{e \in T} x_e$. Let us consider the following constraints:
  \begin{subequations}
    \begin{align}
      x(\delta(v)) - 2y_v & =0\label{cp:deg}, & v \in V \\
      y_v - x_e & \geq 0, & v \in V, \; e \in \delta(v)\label{cp:logical} \\
      x(\delta(Q)) - 2y_v-2y_w & \geq -2,\label{cp:gsec} & v \in Q \subset V, 3 \leq |Q| \leq |V|-3,\; w \in V-Q \\
      x(E) & \geq 3, & \label{cp:egt3} \\
      1 \geq y_v & \geq 0,\label{cp:vbnd} & v\in V\\
      x_e & \geq 0,\label{cp:ebnd} &e\in E \\
      x_e & \in \mathbb{Z}\label{cp:eint} &  e \in E
    \end{align}\label{cp:all}
  \end{subequations}
  The degree equations~\eqref{cp:deg} together with the logical constraints~\eqref{cp:logical} and the integrality constraints~\eqref{cp:eint} ensure that the visited vertices have exactly two incident edges and the unvisited vertices none. The Subcycle Elimination Constraints (SEC)~\eqref{cp:gsec} ensure that only one connected cycle exists. Throughout the paper, we use the notation $\langle Q, v, w\rangle$ to refer to the SEC defined by the set $Q$ and the vertices $v\in Q$ and $w\notin Q$. In the literature, the SECs have also been called Generalized Subtour Elimination Constraints (GSEC). The inequality~\eqref{cp:egt3} imposes the property that the undirected cycles contain at least 3 edges. The conditions~\eqref{cp:vbnd},~\eqref{cp:ebnd} and~\eqref{cp:eint} impose that all the variables are 0{-}1. Note that the integrality of the $y_v$ variables is ensured by~\eqref{cp:deg},~\eqref{cp:logical} and~\eqref{cp:eint}, and the condition $x_e \leq 1$ is ensured by~\eqref{cp:logical} and~\eqref{cp:vbnd}. Considering the constraints in~\eqref{cp:all}, the cycle polytope of a graph $G=(V, E)$ can be expressed as follows:
  \begin{equation}
    P^G_{{C}}  =  conv \{(y, x) \in \mathbb{R}^{V \times E}: (y, x) \textnormal{~satisfies~\eqref{cp:deg},~\eqref{cp:logical},~\eqref{cp:gsec},~\eqref{cp:egt3},~\eqref{cp:vbnd},~\eqref{cp:ebnd},~\eqref{cp:eint}}\}
  \end{equation}

  In some problems, for instance OP and PCTSP, a feasible solution must visit a depot vertex, i.e., $y_d = 1$ for a vertex $d \in V$. In such cases, the family of SECs~\eqref{cp:gsec} that define the cycle polytope can be substituted with the following subfamily:
  \begin{equation}
    x(\delta(Q)) - 2y_v \geq 0,\label{cp:gsec2}  \qquad  v \in Q \subset V, 3 \leq |Q| \leq |V|-3,\; d \notin Q
  \end{equation}
  where each constraint can be represented as $\langle Q, v\rangle$. In a B\&C algorithm, where all the constraints of the model are not considered in the $LP_0$, the only advantage by using this constraint family is that we simplify a vertex in the SEC representation. However, it has one important disadvantage, in the family~\eqref{cp:gsec2} we might need to consider an SEC with $|Q|>|V|/2$, while in the family~\eqref{cp:gsec} it can be considered always a SEC such that $|Q| \leq |V|/2$. Therefore, we always consider the family~\eqref{cp:gsec} regardless of whether it is given a depot or not in the cycle problem.

  When a B\&C is used to solve a cycle problem, the integrality constraints~\eqref{cp:eint} of the $P^G_C$ are relaxed in order to first seek a solution that satisfies the rest of the constraints. Contrary to this strategy,~\cite{pferschy2017} have recently considered again relaxing the SEC constraints in the TSP, to first solve the resulting problem to integer optimality with MILP{-}solvers and then introduce the SECs if required. Despite the improvement of the new MILP{-}solvers, this approach is still inferior compared to the opposite strategy. As a consequence of the continuous relaxation, a solution $(y,x)$ that satisfies the rest of the constraints of~\eqref{cp:all} might still not belong to $P^G_{C}$. In these cases, instead of directly resorting to the branching phase to tighten the integrality gap, we could check if additional (not dominated by those in~\eqref{cp:all}) and facet-defining valid inequalities for the $P^G_C$ are violated. The strength of considering additional valid inequalities was shown in the 1970s in the study of the TSP~\cite{Grotschel1979}. In~\cite{bauer1997} an extension of the clique trees inequality family (originally defined for the TSP) was given, which includes the so-called comb inequalities, for cycle problems. The shrinking rules proposed in Section~\ref{sec:srk} are safe for all the valid inequalities for $P^G_C$.

  A polytope that it is closely related to $P^G_{C}$ is the so-called lower cycle polytope, see~\cite{bauer1997}:
  \begin{equation}
    L^G_C=conv \{P^G_C, (0,0)\}
  \end{equation}
  where $(0,0)\in \mathbb{R}^{V \times E}$ is the vector that represents that no vertex and edges of the graph are visited. It is easy to see, that for every graph $G$, so that it contains at least one cycle, there exist an infinity number of vectors $(y,x) \in L_C$ such that $x(E) < 3$. Hence, the polytope $P^G_C$ is a proper subspace of $L^G_C$ for every graph $G$ that contains at least one cycle. It is crucial to consider the polytope $L^G_C$ to obtain the shrinking results in Section~\ref{sec:srk}.

  In a B\&C algorithm, it is reasonable to solve the separation problems of the valid inequality families following an order determined by their complexity. This order defines a hierarchy of the inequality families and their closure polytopes. We refer to the closure polytope of an inequality family as the polytope that satisfies all the inequalities of the given family and its preceding families in this hierarchy.

  Without considering the variable bounds~\eqref{cp:vbnd}-\eqref{cp:ebnd} and the inequality~\eqref{cp:egt3}, the simplest inequalities are the degree equations~\eqref{cp:deg} and the logical constraints~\eqref{cp:logical}. These have, respectively, linear and quadratic exact algorithms in terms of the number of the vertices of $G$ and generally are always included in the $LP_0$. The closure polytope of the inequalities~\eqref{cp:deg} and~\eqref{cp:logical} (the inequality~\eqref{cp:egt3} is excluded to favour the convexity) turns out to be the undirected Assignment Polytope (with loops), $P^G_A$, which is defined as:
  \begin{equation}
    P^G_A := \{(y,x) \in \mathbb{R}^{V \times E}: (y,x) \textnormal{~satisfies~\eqref{cp:deg},~\eqref{cp:logical},~\eqref{cp:vbnd},~\eqref{cp:ebnd}}\}\label{poly:cp}
  \end{equation}

  Next in the hierarchy comes the SEC family. A straightforward exact separation algorithm for the SECs has $O(|V|^4)$ time complexity (see Section~\ref{sec:sep-strat} for further discussion) and its closure polytope is defined as:
  \begin{equation}
    P^G_{SEC} := \{(y,x) \in P^G_{A}: (y,x) \textnormal{~satisfies~\eqref{cp:gsec}} \}\label{poly:sec}
  \end{equation}
  Considering the relationship $P^G_{C} \subset P^G_{SEC} \subset P^G_{A}$, the underlying purpose of this paper is to effectively determine if a given solution $(y,x) \in P^G_A$ of a $LP_0$ belongs to $P^G_{SEC}$, or in case that it does not belong, to provide the violated inequalities.

  Throughout the paper, we make use of the following well-known identity repeatedly. Given a graph $G$, a subset $S \subset V$ and a vector $x \in \mathbb{R}^{ E}$, the identity
  \begin{equation}
    x(\delta(S)) = \sum_{v \in S} x(\delta(v)) - 2 x(E(S))
  \end{equation}
  is always satisfied. In addition, if the vector $(y, x) \in \mathbb{R}^{V\times E}$ satisfies the degree constraints~\eqref{cp:deg}, then the equations
  \begin{equation}
    x(\delta(S)) = 2 y(S) - 2 x(E(S)) \quad S \subset V \label{cp:eq}
  \end{equation}
  are satisfied by the vector $(y, x)$. Particularly, the identity~\eqref{cp:eq} is satisfied by every vector in $P^G_{TSP}$, $P^G_C$, $P^G_{SEC}$ and $P^G_A$.

\end{section}

\begin{section}{Shrinking for the Cycle Polytope}\label{sec:srk}

  In this section, we present three shrinking rules that are safe for the $P^G_C$, i.e., rules that preserve the existence of violated cycle inequalities in the shrunk graph. In essence, we have generalized for every (simple) cycle problem the results obtained by~\cite{Padberg1990b} for Hamiltonian cycle problems. In the following lines, we formalize the concept of safe shrink for $P^G_{C}$ and we prove the lemmas and the theorem in which shrinking rules for cycle problems are based on. In addition, we show that the three shrinking rules can be consecutively applied for the $P^G_C$.

  Let us introduce the following notation. Given a graph $G=(V,E)$, the vector $(y,x) \in \mathbb{R}^{V \times E}$ and a subset $S \subset V$, we denote by $ G[S]= (V[S],E [S])$ the graph obtained by shrinking the set $S$ into a single vertex $s\notin V$, where the resulting set of vertices and edges are as follows:
  \begin{subequations}
    \begin{align}
      V[S]        & = (V-S) \cup \{s\} \\
      E[S]        & = E(V-S)  \cup \{[s, v]: v \in {V-S}, x(S:v) > 0\}
    \end{align}
  \end{subequations}
  and by $(y[S],x[S]) \in \mathbb{R}^{V[S] \times E[S]}$ we denote the vector with components
  \begin{subequations}
    \begin{align}
      x[S]([u,v])  & = x_{[u,v]}  & \forall [u,v] & \in E \cap E[S] \\
      x[S]([s,v])  & = x(S:v)     & \forall v & \in V -S \\
      y[S](v)      & = y_v         & \forall v & \in V \cap V[S] \\
      y[S](s)      & = x(\delta(S))/2
    \end{align}
  \end{subequations}

  Let $Q \subset V $ be a subset of vertices, we denote with $Q[S]$ the subset derived by shrinking $S$

  \begin{equation}
    Q[S]= \left\{\begin{aligned}
        & (Q - S) \cup \{s\} & \quad \; \mbox{if $S \cap Q \neq \emptyset $} \\
        & Q & \quad \; \mbox{otherwise}
    \end{aligned} \right.
    \label{def:srkset}
  \end{equation}
  which has the following associated values:
  \begin{subequations}
    \begin{align}
      y[S](Q[S]) & = \left\{\begin{aligned}
         & y(Q) - y(Q \cap S) + \frac{x(\delta(S))}{2}  &  \mbox{if $S \cap Q \neq \emptyset $} \\
         &   y(Q) & \qquad \quad \; \mbox{otherwise}
      \end{aligned} \right. \\[10pt]
          x[S](\delta(Q[S])) & = \left\{\begin{aligned}
          & x(\delta(S \cup Q))&  \hspace{30mm} \mbox{if $S \cap Q \neq \emptyset $} \\
          & x(\delta(Q)) & \mbox{otherwise}
          \end{aligned}\label{def:srkdelta} \right. \\[10pt]
              x[S](E(Q[S])) & = \quad x(E(Q)) - x(E(Q \cap S)) \label{def:srkstar}
            \end{align}
            \end{subequations}

            Based on the definition given in~\cite{Padberg1990b} for safe shrinking for the $P^G_{{TSP}}$, an analogue definition can be formulated for safe shrinking for the $P^G_{C}$.

            \begin{definition}
              Given a vector $(y, x) \notin P^G_{C}$, a set $S\subset V$ is safe to shrink if $(y[S],x[S]) \notin P^{G[S]}_{C}$.
            \end{definition}

            Note that the definition does not assume a one-to-one correspondence between the violated inequalities of $(y, x)$ and $(y[S], x[S])$ (e.g. different violated cuts for $(y, x)$ might overlap to the same violated cut for $(y[S], x[S])$).
            When a set $S$ is safe to shrink for a given $(y,x)$, it is also said that $S$ is shrinkable for $(y,x)$.

            The definition of shrinkable set does not provide a practical tool for finding them. Hence, the first goal is to give a set of rules of shrinking for $P^G_{C}$, which are obtained in Theorem~\ref{corol:srk}.
            The strategy used in~\cite{Padberg1990b} to obtain the shrinking rules for tours cannot be applied directly for simple cycles, because it relies on the fact that the tours visit every vertex in the graph. So, first we need to obtain the following lemma.

            \begin{lemma}\label{lemma:srk2}
              Let $(y, x) \in L^G_{C}$ be a vector. Suppose that $\{Q, \{u\}, \{v\}\}$ is a partition of $V$ such that $x_{[u,v]}=x(u:Q)=x(v:Q)>0$.
              Then any cycle $\tau$ of $\mathcal{C}^G$ that has a positive coefficient in the convex combination of $(y,x)$, $\lambda_{\tau}>0$, fulfills one of the following cases:
              \begin{enumerate}[(i)]
                \item $V(\tau) \subset Q$
                \item $|\tau \cap (u:Q)| = |\tau \cap (v:Q)| = |\tau \cap [u, v]| = 1$
              \end{enumerate}
            \end{lemma}
            \begin{proof}
              Let $\mathcal{C}_{uv}$ denote the subset of cycles in $\mathcal{C}$ that visits the edge $[u, v]$ and has a positive value, $\lambda_{\tau} >0$. Note that since $(y, x)\in  L^G_{C}$, then $x_{[u,v]} \leq y_v$ and $x_{[u,v]} \leq y_u$. So, in order to satisfy the degree equations, every cycle $\tau$ in $\mathcal{C}_{uv}$ must contain at least an edge in $(u:Q)$ and $(v:Q)$. Moreover, since $\tau$ is a simple cycle, every $\tau \in \mathcal{C}_{uv}$ crosses exactly once $(u : Q)$ and $(v : Q)$. Now, let us see that if $\tau$ does not belong to $\mathcal{C}_{uv}$ and $\lambda_{\tau}>0$, then $\tau$ is contained in $Q$. Consider the following inequality:
              \begin{subequations}
                \begin{align}
                  x_{[u,v]} = \sum_{\substack{\zeta \in \mathcal{C}_{uv}}} \lambda_{\zeta} x^{\zeta}_{[u,v]} = \sum_{\substack{\zeta \in \mathcal{C}_{uv}}} \lambda_{\zeta} = \sum_{\zeta \in \mathcal{C}_{uv}} \sum_{e \in (u:Q)} \lambda_{\zeta} x^{\zeta}_e & \leq  \\
                  \sum_{\zeta \in \mathcal{C}_{uv}} \sum_{e \in (u:Q)} \lambda_{\zeta} x^{\zeta}_e + \sum_{\zeta \notin \mathcal{C}_{uv}} \sum_{e \in (u:Q)} \lambda_{\zeta} x^{\zeta}_e & = x(u:Q)
                \end{align}
              \end{subequations}
              Since $x_{[u,v]} = x(u : Q)$, we have that $x^{\tau}_e = 0$ for every $e \in (u:Q)$. Similarly, we obtain that $x^{\tau}_e = 0$ for every $e \in (v:Q)$. Therefore, $\tau$ is contained in $Q$.
              
            \end{proof}

            \begin{figure}[htb!]
              \centering
              \includegraphics[width=0.4\columnwidth]{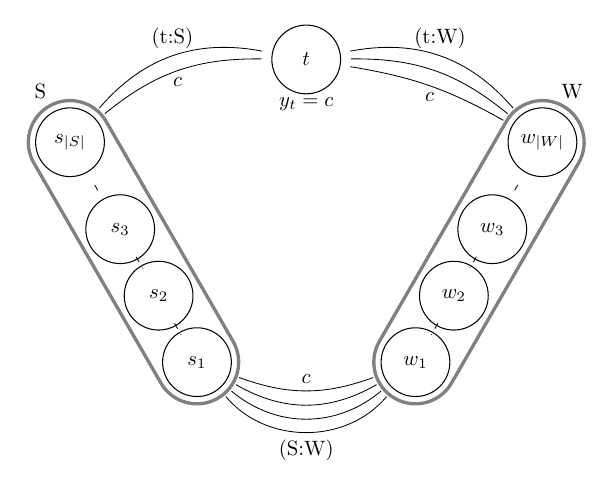}
              \caption{Illustration of the scenario in Lemma~\ref{thm:srk}.}\label{fig:situation}
            \end{figure}

            The next result generalizes the main theorem of shrinking in~\cite{Padberg1990b}. The principal idea is to use a constant, $c$, to extend the rules of the original paper (where $\forall v \in V$ satisfies $y_v=1$) for vertices that have fractional value. We also need an additional hypothesis about the vector $(y[W], x[W])$ obtained by shrinking the subset $W$, the ``complement'' of $S$, which is not required for the TSP because it is trivially satisfied by Hamiltonian cycles.

            \begin{lemma}\label{thm:srk}
              Given a vector $(y, x) \notin P^G_{C}$, let $\{S, W, \{t\}\}$ be a partition of $V$ with $2 \leq |S|$ and $c$ be a constant where $0<c\leq1$  such that:
              \begin{enumerate}[(i)]
                \item $y_v = c$ $\forall v \in S \cup \{t\}$\label{hip:srk-1}
                \item $x(E(S)) = c \cdot (|S|-1)$\label{hip:srk-2}
                \item $x(t:S) = c$\label{hip:srk-3}
                \item $(y[W], x[W]) \in L^{G[W]}_{C}$\label{hip:srk-4}
                \item No cycle in the convex combination of $(y[W], x[W])$ is contained in $S$\label{hip:srk-5}
              \end{enumerate}
              Then it is safe to shrink $S$ for $(y, x)$.
            \end{lemma}

            \begin{proof}
              Based on the hypotheses~\ref{hip:srk-1}),~\ref{hip:srk-2}) and~\ref{hip:srk-3}) of the lemma and the identity~\eqref{cp:eq} we obtain that $x(S:W)=c$ and $x(t:W)=c$, as illustrated in Figure~\ref{fig:situation}.

              Suppose for contradiction that $S$ is not shrinkable, so $(y[S],x[S]) \in P^{G[S]}_{C}$. Since $x_{[s,t]}=x(s:W)=x(t:W)$, based on Lemma~\ref{lemma:srk2}, the vector $(y[S], x[S])$ can be written as:
              \begin{equation}
                (y[S], x[S]) = \sum_{\substack{\zeta \in \mathcal{W}_s}} \alpha_{\zeta}  {(y,x)}^{\zeta} + \sum_{\substack{\zeta \in \mathcal{W}_0}} \alpha^0_{\zeta}  {(y,x)}^{\zeta}
              \end{equation}
              \noindent where $\mathcal{W}_s$ is the set of cycles visiting the shrunk vertex $s$ having $\alpha_{\zeta} > 0$ and $\mathcal{W}_0$ is the set of cycles contained in $W$ having $\alpha_{\zeta}^0>0$. Note that $\mathcal{W}_0$ might be an empty set. The coefficients satisfy  ${\sum_{\substack{\zeta \in \mathcal{W}_s}} \alpha_{\zeta} + \sum_{\zeta \in \mathcal{W}_0} \alpha^0_{\zeta}=1}$.

              By hypothesis the vector $(y[W], x[W])$ belongs to $ L^{G[W]}_{C}$, so $(y[W], x[W])$ can be written as a convex combination of cycles of $\mathcal{C}_{G[W]}$ and the vector $(0, 0)$. Because of the Lemma~\ref{lemma:srk2}
              and by the hypothesis~\ref{hip:srk-5}) the vector $(y[W], x[W])$ can be written as:
              \begin{equation}
                (y[W], x[W]) = \sum_{\substack{\eta \in \mathcal{S}_w}} \beta_{\eta}  {(y,x)}^{\eta} + \beta_{(0,0)} (0,0)
              \end{equation}
              \noindent where $\mathcal{S}_w$ is the set of cycles visiting $w$ (the vertex to which $W$ is contracted to) having $\beta_{\eta}>0$, $\beta_{(0,0)}\geq 0$ and $\sum_{\substack{\eta \in \mathcal{S}_w}}\beta_{\eta}  + \beta_{(0,0)} = 1$.

              Now, considering $x(t:s)=x(t:w)=c$ we have that:
              \begin{equation}
                c = \sum_{\substack{\zeta \in \mathcal{W}_s}} \alpha_{\zeta}  = \sum_{\substack{\eta \in \mathcal{S}_w}} \beta_{\eta}
              \end{equation}
              and from the fact that the coefficients sum up to one, we have that:
              \begin{equation}
                1 - c  = \sum_{\substack{\eta \in \mathcal{W}_0}} \alpha^0_{\eta} = \beta_{(0,0)}
              \end{equation}

              To prove the lemma we follow the ``patch-and-weight'' strategy used in~\cite{Padberg1990b} for the $P^G_{TSP}$ whose goal is to reconstruct the cycles and coefficients of the convex combination of the vector $(y, x)$. According to the vertices in $W$, we can partition $\mathcal{W}_s$ into $|W|$ pairwise disjoint subsets (some of them which be empty). For $j\in \{1,\ldots,|W|\}$ let us call $\mathcal{W}^j_s$ the subset of cycles in $\mathcal{W}_s$ containing the edge $[s, w_j]$, and denote by $\zeta^j_1, \ldots, \zeta^j_{k_j}$ the cycles of $\mathcal{W}^j_s$ and by $\beta^j_1, \ldots, \beta^j_{k_j}$ their coefficients in the convex combination. In the same way, we can partition $\mathcal{S}_w$ into $|S|$ subsets calling $\mathcal{S}^i_w$ the subset of cycles in $\mathcal{S}_w$ containing the edge $[s_i, w]$. We denote by $\eta^i_1, \ldots, \eta^i_{h_i}$ the cycles of $\mathcal{S}^i_w$ and by $\alpha^i_1, \ldots, \alpha^i_{h_i}$ their coefficients in the convex combination.

              The cycles of the convex combination of $(y,x)$ are constructed in two steps. In the first step, $|\mathcal{S}_w|$ copies of each cycle in $\mathcal{W}_s$ are created. With this goal, for each $j \in \{1, \ldots, |W|\}$ and for each $l \in \{1, \ldots, k_j\}$, create $|S|$ copies of the cycle $\zeta^j_l$, and denote them by $\{\tau^{ij}_l\}$ for $i \in \{1,\ldots,|S|\}$. Then, for each $j \in \{1, \ldots, |W|\}$, for each $l \in \{1, \ldots, k_j\}$ and for each $i \in \{1, \ldots, |S|\}$ create $h_i$ copies of $\tau^{ij}_{l}$, and denote them by $\{\tau^{ij}_{ml}\}$ for $m\in \{1,\ldots,h_i\}$. At this point we have $|\mathcal{W}_s| \cdot |\mathcal{S}_w|$ cycles that belong to $G[S]$. In the second step, these cycles of $G[S]$ are extended to cycles of $G$. To that end, consider each cycle $\tau^{ij}_{ml}$ and remove the edges $[t,s]$ and $[s, w_j]$ and join the resulting path with the path in $G[W]$ obtained from the cycle $\eta^i_m$ by removing the edges $[w, t]$ and $[s_i, w]$, and add the edge $[s_i, w_j]$ to obtain the extension of $\tau^{ij}_{ml}$ to $G$.

              The coefficients of the constructed $\tau^{ij}_{ml}$ cycles are defined in the following way:

              \begin{equation}
                \lambda^{ij}_{ml} =
                \frac{x_{[s_j, w_i]} \cdot \alpha^j_l \cdot \beta^i_m}{\sum_{r=1}^{k_j} \alpha^j_r \cdot \sum_{r=1}^{h_i} \beta^i_r}
              \end{equation}
              \noindent where $i \in \{1, \ldots, |S|\}$, $j \in \{1, \ldots, |W|\}$,  $m \in \{1, \ldots, h_i\}$ and $l \in \{1, \ldots, k_j\}$. It can be verified that the coefficients defined this way sum $c$ in total:
              \begin{equation}
                \sum_{i,j,m,l}\lambda^{ij}_{ml} =
                \sum_{i,j} x_{[s_j, w_i]} \sum_{m,l} \frac{\alpha^j_l \cdot \beta^i_m}{\sum_{r=1}^{k_j} \alpha^j_r \cdot \sum_{r=1}^{h_i} \beta^i_r} = \sum_{i,j} x_{[s_j, w_i]} = x(S:W) = c
              \end{equation}

              Then the vector $(y, x)$ can be obtained as a convex combination of the cycles in $\mathcal{W}_0$ and $\{\tau^{ij}_{ml}\}$ with coefficients $\{\alpha^0_\zeta\}$ and $\{\lambda^{ij}_{ml}\}$, respectively. We conclude $(y, x) \in P^G_{C}$ which is a contradiction.
              
            \end{proof}

            The lemma gives a sufficient condition for a set to be shrinkable, but still it is not practical. The next theorem gives three practical scenarios to make use of Lemma~\ref{thm:srk}. Beforehand, let us obtain a useful result for $L^G_{C}$. Consider the undirected version of the Assignment Polytope (without loops) $P^1_A$ defined as:
            \begin{equation}
              P^1_{A}  :=  \{(y,x)\in \mathbb{R}^{V\times E}: (y,x) \textnormal{~satisfies~\eqref{cp:deg},~\eqref{cp:logical},~\eqref{cp:ebnd}, $y=1$}\}
            \end{equation}
            \noindent
            It is a well-known result of the literature that $P^G_{TSP}=P^1_A$ for $3 \leq |V| \leq 5$ (see~\cite{Grotschel1979}). This relationship is the key to obtaining the shrinking rules for the $P^G_{TSP}$ in~\cite{Padberg1990b}. So, we would like to obtain a similar result for $L^G_C$ and $P^G_A$.
            However, $L^G_C \neq P_A$ when $4 \leq |V|$, as shown in the counterexample of Figure~\ref{fig:assig}. The vector defined in the figure belongs to $P^G_A$, but it does not belong to $L^G_C$, because it cannot  be expressed as a convex combination of cycles.

            \begin{figure}[htb!]
              \begin{center}
                \includegraphics[width=0.25\columnwidth]{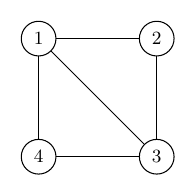}
                \caption{An example of a solution that belongs to $P^G_A$ but not to $L^G_C$ when $|V| = 4$ (it can be easily extended for $|V| \geq 4$ by means of subdivisions). All the edges in the figure have value $\frac{1}{2}$. The values of the vertices satisfy the degree equations.}\label{fig:assig}
              \end{center}
            \end{figure}
            Nevertheless, we have the following lemma which is enough to prove Theorem~\ref{corol:srk}.
            \begin{lemma}\label{thm:assgim5}
              Let $G=(V,E)$ be a graph and $c$ be a constant such that $3 \leq |V| \leq 5$ and $0 < c \leq 1$. If $(y, x) \in P_A$ such that $y_v=c$ for all $v \in V$, then $ (y,x) \in L_{C}$.
            \end{lemma}

            \begin{proof}
              It is straightforward that if $(y,x) \in P_A$ such that $y_v=c$ for all $v \in V$, then $\frac{1}{c}(y,x) \in P^1_A$. By the classical result in~\cite{Grotschel1979}, since $3 \leq |V| \leq 5$, the equality $P^1_A = P_{TSP}$ is satisfied. Since  $P^G_{TSP}$ is contained in $L^G_{C}$, the vector $\frac{1}{c} (y,x)$ belongs to $L^G_{C}$. Then, since both $(0, 0)$ and $\frac{1}{c} (y,x)$ belong to $L^G_{C}$, which is convex, and $0 \leq c \leq 1$ we have that $(y,x) \in L_{C}$.
              
            \end{proof}

            \begin{lemma}\label{lema:notcontinS}
              Given a graph $G$ such that $|V| = 5$, a vector $(y, x) \in L^G_{{C}}$ and $0 \leq c \leq 1$, suppose that $\lambda_{(0,0)} = 1-c$. Let $\{S, \{t\}, \{w\}\}$ be a partition of $V$ such that $x_{[t,w]} = x(t:S) = x(w:S) = c$, then every cycle $\tau$ in $\mathcal{C}^G$ such that $\lambda_{\tau} > 0$ is not contained in $S$.
            \end{lemma}
            \begin{proof}
              Since $\{S, \{t\}, \{w\}\}$ is a partition of $V$, we have that $|S| = 3$ and $|V-S| = 2$. Hence, every cycle in $\mathcal{C}^G$ has vertices in $S$. According to the number of visited vertices of $S$, we can partition $\mathcal{C}^G$ into $3$ subsets $\{\mathcal{C}_1, \mathcal{C}_2 ,\mathcal{C}_3\}$. Furthermore, the set $\mathcal{C}_3$ can be partitioned into two subsets, $\mathcal{C}^{in}_3$ and  $\mathcal{C}^{out}_3$, determined by whether the cycles are fully contained in $S$ or not. Since $(y, x)$ belongs to $L^G_{{C}}$, there is a convex combination of cycles of $\mathcal{C}^G$ whose coefficients satisfy
              \begin{equation}
                \sum_{\tau \in \mathcal{C}_1} \lambda^1_{\tau} + \sum_{\tau \in \mathcal{C}_2} \lambda^2_{\tau}+ \sum_{\tau \in \mathcal{C}^{out}_3} \lambda^{3o}_{\tau} + \sum_{\tau \in \mathcal{C}^{in}_3}  \lambda^{3i}_{\tau}  + \lambda_{(0,0)}=1\label{eq:notcontinS1}
              \end{equation}
              Since the cycles in $\mathcal{C}_1$, $\mathcal{C}_2$ and  $\mathcal{C}^{out}_3$ have edges in $(t:S)$ and  $(w:S)$, by the Lemma~\ref{lemma:srk2}, each cycle has exactly one edge in the mentioned edge sets. Now, consider the hypothesis that $x(t:S)=c$ (or $x(w:S)=c$), so the coefficients also satisfy the following identity:
              \begin{equation}
                \sum_{\tau \in \mathcal{C}_1} \lambda^1_{\tau} + \sum_{\tau \in \mathcal{C}_2} \lambda^2_{\tau} + \sum_{\tau \in \mathcal{C}^{out}_3} \lambda^{3o}_{\tau} =c\label{eq:notcontinS2}
              \end{equation}
              By hypothesis, we have that $\lambda_{(0, 0)}=1-c$ and by~\eqref{eq:notcontinS1} and~\eqref{eq:notcontinS2}, we obtain that $\lambda^{3i}_{\tau}=0$ for all $\tau \in \mathcal{C}^{in}_3$, which means that every cycle in $\mathcal{C}^G$ contained in $S$ has null coefficient.
              
            \end{proof}

            \begin{theorem}[Rules C1, C2 and C3]\label{corol:srk}
              Given a vector $(y,x) \notin P^G_{C}$, let $S \subset V$ with $2 \leq |S| \leq 3$, $t \in V-S$ and $0 < c \leq 1$ be such that:
              \begin{enumerate}[(i)]
                \item $y_v = c$ $\forall v \in S \cup \{t\}$
                \item $x(E(S)) = c \cdot (|S|-1)$
                \item $x(t:S) = c$
              \end{enumerate}
              Then it is safe to shrink $S$ for $(y, x)$.
            \end{theorem}

            \begin{proof}
              Let $W = V-(S \cup \{t\})$ be a subset of $V$. If the hypotheses are satisfied, note that $W$ is non-empty. Since $2 \leq |S| \leq 3$, we have that $4 \leq |V[W]| \leq 5$. Notice that, $y_v=c$ for all the vertices of $V[W]$ and $(y[W], x[W]) \in P^{G[W]}_A$. Under these hypotheses, by Lemma~\ref{thm:assgim5}, the vector $(y[W], x[W])$ belongs to $L^{G[W]}_{C}$. When $|S|=2$, it does not exist any cycle contained in $S$. When $|S|=3$, as a consequence of Lemma~\ref{lema:notcontinS}, we have that it does not exist a cycle in the convex combination of $(y[W], x[W])$ contained in $S$. Therefore, the hypotheses of Lemma~\ref{thm:srk} are satisfied and $S$ is shrinkable.
              
            \end{proof}

            From Theorem~\ref{corol:srk}, three shrinking rules can be derived, which are summarized in Figure~\ref{fig:srk}: the rules C1 and C2 correspond to the case $|S|=2$ and the rule C3 to $|S|=3$.

            \begin{figure}[htb!]
              \centering
              \includegraphics[width=\columnwidth]{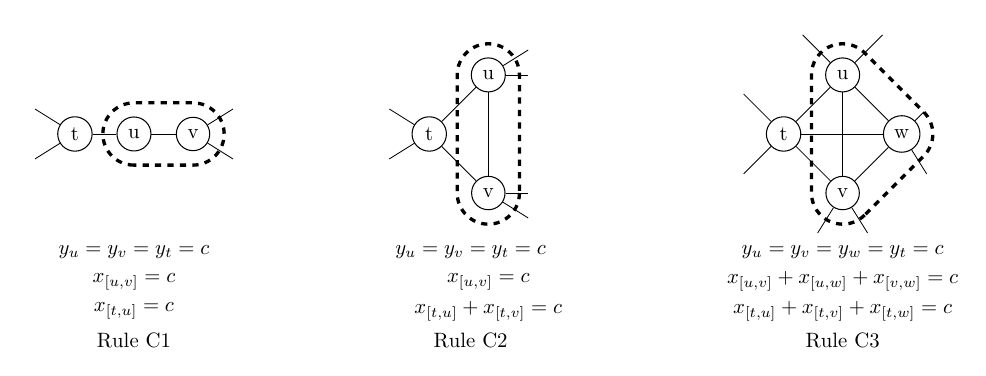}
              \caption{Illustration of the three shrinking rules derived from the Theorem~\ref{corol:srk}}\label{fig:srk}
            \end{figure}

            It is easy to see that rule C2 dominates the rule C1, in fact it is just a particular case of it. The reason to split them, is that the cost of checking C1 is lower than the cost of C2. By contrast, rule C3 is not dominated by the rules C1 and C2. In Figure~\ref{fig:sec1}, an example is given of a vector  $(y,x) \in P_A$ in which rule C3 can be applied but not C1 and C2. For instance, if we consider $S=\{1,2,3\}$, $W=\{4,5,6\}$ and $t=7$, then $S$ is shrinkable by rule C3. Since the vertices and edges have different values, there is no shrinkable set that can be identified by rule C1 or C2.
            \begin{figure}[htb!]
              \centering
              \hspace*{0.2\columnwidth}
              \includegraphics[width=0.7\columnwidth]{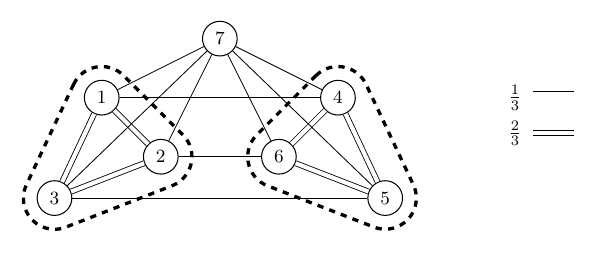}
              \caption{Example of a pair $G$ and $(y,x)\in P_A$ where rule C3 can be applied but not rules C1 nor C2. The values of the edges are the ones detailed in the legend and all the vertices have value 1.}\label{fig:sec1}
            \end{figure}

            A useful property of the rules derived from Theorem~\ref{corol:srk} is that the value of the vertices is inherited in the shrunk graphs.

            \begin{lemma}\label{lemma:yc}
              Under the hypotheses of Theorem~\ref{corol:srk}, $y[S](v[S]) = y_v$ for all $v\in V$.
            \end{lemma}
            \begin{proof}
              For every $v\in V-S$, we have $y[S](v[S]) = y_v$ by definition. Since $2y_s=x(\delta(S))=2y_v$ for $v\in S$ we obtain the result of the lemma.
              
            \end{proof}

            In the preprocess of separation algorithms, it is desirable to perform multiple consecutive safe shrinkings. For that aim, we need to analyse what happens with the hypotheses of Theorem~\ref{corol:srk} after the contraction of a shrinkable set. More precisely, we need to see when the shrunk vector belongs to $P^G_A$.

            \begin{lemma}\label{lemma:srkPa}
              Let $S$ be a shrinkable set for $(y, x) \in P^G_A$ obtained from Theorem~\ref{corol:srk} using the $\{S, W, \{t\}\}$ partition. Then, $(y[S], x[S])$ satisfies the degree equations and the logical constraints associated with every edge in $E(W) \cup (t:V)$. In addition, we have either
              \begin{enumerate}[i)]
                \item $(y[S], x[S]) \in P^{G[S]}_A$, or
                \item $\exists w \in W$ such that $y_w < y_s$ and $y_w < x_{[w,s]} \leq y_s$\label{lemma:srkPa-2}
              \end{enumerate}
            \end{lemma}
            \begin{proof}
              From the definition of the shrunk vector, it is clear that $(y[S], x[S])$ satisfies the degree equations. Since $v\in S$ satisfies $y_v \leq 1$, $y_s=y_v$ also satisfies $y_s \leq 1$. Moreover, $x_{[t, s]} = y_s = y_t$.
              If $x(w:S) \leq y_w$ for all $w \in W$ then $(y[S], x[S])$ satisfies the logical constraints and $(y[S], x[S]) \in P^{G[S]}_A$. If the previous is not true, there exists a vertex $w\in W$ such that $x(w:S)> y_w$ and $y_w < y_s$ (because by hypothesis $(y,x) \in P^{G}_A$). Therefore, the logical constraint $x_{[w,s]}\leq y_w$ is violated for $(y[S], x[S])$ by a vertex $w \in W$ such that $y_w < y_s$.
              
            \end{proof}

            There are two scenarios where the shrunk vector always belongs to $P^G_A$. First, when all the vertices of $V$ have the same $y$ value, as is the case when $(y, x) \in P^G_{TSP}$, and secondly, when only rule C1 is applied. The next theorem shows that if $(y, x) \in P^G_A$, it is possible to shrink a subset $S$ obtained by the rules of Theorem~\ref{corol:srk} and continue with further safe shrinkings regardless of whether or not $(y[S], x[S])$ belongs to $P^{G[S]}_A$.

            \begin{theorem}\label{thm:multisrk}
              Given a vector $(y, x) \in P^G_A$, it is safe to consecutively apply the shrinking rules derived from Theorem~\ref{corol:srk}.
            \end{theorem}
            \begin{proof}
              Let $S$ be a subset obtained from Theorem~\ref{corol:srk} such that $(y[S], x[S]) \notin P^{G[S]}_A$. By Lemma~\ref{lemma:srkPa} we know that the only violated logical constraints of $(y[S], x[S])$ consist of edges whose vertices, $s$ and $v\in W$, have different values $y_v<y_s$. Notice that in the proof of Theorem~\ref{corol:srk} the hypothesis that the logical constraints are satisfied is used twice. First in Lemma~\ref{lemma:srk2}, which is applied for vertices having the same value. Secondly in Theorem~\ref{corol:srk}, where it is assumed $(y[N], x[N])\in P^{G[N]}_A$ for a given subset $N$ of $V[S]$. In order to see that this last hypothesis is always satisfied by every shrinkable set candidate, let us suppose that $\{M, N, \{r\}\}$ is a partition of $V[S]$ that satisfies hypotheses~\ref{hip:srk-1}),~\ref{hip:srk-2}) and~\ref{hip:srk-3}) of Theorem~\ref{corol:srk}. Then there are two possible cases: $v\in M\cup\{r\}$ and $s\in N$, or vice versa. The hypothesis $(y[N], x[N])\in P^{G[N]}_A$ is satisfied in both cases, because $x_{[n,v]} \leq y_n=y_u$ for $u\in M\cup\{r\}$.
              
            \end{proof}

            Another interesting scenario occurs when there is at least a vertex $v\in V$ satisfying $y_v=1$, as happens in the context of cycle problems with depot. In all these problems, the case~\ref{lemma:srkPa-2}) of Lemma~\ref{lemma:srkPa} has a special meaning as shown in Theorem~\ref{thm:srkPa2}.

            \begin{lemma}\label{lemma:egt}
              If $(y, x) \in \mathbb{R}^{V \times E}$ satisfies the degree equations~\eqref{cp:deg} and $u,v \in V$ are two vertices such that $x_{[u,v]} > y_u$ then $x(\delta(\{u, v\})) < 2 y_v$.
            \end{lemma}
            \begin{proof}
              As $(y, x)$ satisfies the degree equations:
              \begin{equation}
                2y_u < 2x_{[u,v]} = 2y_u + 2y_v - x(\delta(\{u, v\}))
              \end{equation}
              
            \end{proof}

            \begin{theorem}\label{thm:srkPa2}
              Given a vector $(y, x)\in P^G_A $, let $O=\{v \in V: y_v = 1\}$ be the subset of vertices with value equal to one and $S$ be a shrinkable set for $(y, x)$ obtained from Theorem~\ref{corol:srk} such that $O-S \neq \emptyset$. Then, we have either
              \begin{enumerate}[i)]
                \item $(y[S], x[S]) \in P^{G[S]}_A$, or
                \item $\exists w \in V-S$ such that, for every $u\in S$ and $v \in O-S$, the SEC $\langle S\cup\{w\}, u, v\rangle$ is violated by $(y, x)$.
              \end{enumerate}
            \end{theorem}
            \begin{proof}
              Note that, in the case~\ref{lemma:srkPa-2}) of Lemma~\ref{lemma:srkPa}, the vertex $w \in V-S$ cannot be contained in $O$ because $y_w<1$. Now, as a consequence of Lemma~\ref{lemma:egt} we can rewrite the second case.
              
            \end{proof}

        \end{section}

        \begin{section}{Safe Shrinking Rules for the Subcycle Closure Polytope}\label{sec:gsec}
          Depending on the inequality, more aggressive contractions can be employed as a preprocess of separation algorithms. In the TSP, for the subtour separation problem,~\cite{crowder1980} introduced subtour specific shrinking rules to simplify the support graphs before proceeding with the separation algorithms. With the aim of motivating the concepts in the subcycle-safe shrinking procedure, let us prove the following result.

          \begin{lemma}\label{lemma:predef}
            Given a vector $(y, x) \in P^G_A$ and an edge $e \in E$, let $S = V(e)$ be the subset associated with the edge $e$. If $(y[S], x[S]) \in P^{G[S]}_{SEC}$, then either
            \begin{enumerate}[i)]
              \item  $(y, x) \in P^G_{SEC}$, or
              \item every violated SEC $\langle Q, r, t\rangle$ for $(y, x)$ satisfies $S \cap Q \neq \emptyset$ and $S - Q \neq \emptyset$\label{scenario:two}

            \end{enumerate}
          \end{lemma}
          \begin{proof}
            Let $e = [u, v]$ be the given edge and $\langle Q, r, t \rangle$ be a SEC for $(y, x)$ such that $S \subset Q$ (or $S \subset V-Q$). On the one hand, since $(y, x) \in P^G_A$, we have $y[S](u[S]) \geq y_u$ and $y[S](v[S]) \geq y_v$. On the other hand, $x[S](\delta(Q[S])) = x(\delta(Q))$ by definition. Then the SEC $\langle Q[S], r[S], t[S] \rangle$ for $(y[S], x[S])$, is at least as violated as $\langle Q, r, t \rangle$ for $(y,x)$. So if $(y[S], x[S]) \in P^{G[S]}_{SEC}$, and $(y, x) \notin P^G_{SEC}$, the only violated SECs for $(y, x)$ are associated with subsets that separate $u$ and $v$.
              
          \end{proof}

          Recall that we want to search the violated SECs for a vector $(y, x) \in P^G_A$, which has been obtained from the $LP_0$ subproblem. Let us assume that we have defined a first shrinking rule that contracts edges by avoiding the scenario~\ref{scenario:two}) of Lemma~\ref{lemma:predef}. So if $(y,x) \notin P^G_{SEC}$, as a consequence of the lemma, $(y,x) \notin P^{G[S]}_{SEC}$. In this case, the vector $(y[S], x[S])$ does not belong to the closure of SECs because either there exists violated logical constraints, SECs or both. Let us suppose that we have a second shrinking rule that identifies (and saves) the violated logicals and ``fixes'' them. Repeatedly applying the second rule, we will eventually reach a vector that satisfies the logical constraints. Now, we are in a similar situation to the starting point, so we can try with the first rule again and so on. This is the main idea exploited in the subcycle-safe shrinking process.

          \begin{definition}
            Given a vector $(y, x) \in \mathbb{R}^{V \times E}$ that satisfies the degree equations, a set $S=\{u,v\} \subset V$ is subcycle-safe to shrink if at least one of the following conditions is satisfied:
            \begin{enumerate}[i)]
              \item $(y[S], x[S]) \notin P^{G[S]}_{SEC}$, or
              \item if there exist violated logical constraints for $(y, x)$, these are associated with the edge $[u,v]$
            \end{enumerate}
          \end{definition}
          Note that the second condition does not require the existence of violated logical constraints for $(y,x)$, which enables the subcycle-safe shrinkable set definition for vectors $(y ,x)$ in $P^G_{SEC}$ to be used. Furthermore, this condition means: if we have already found a violated constraint, we should not worry if later the shrinking the vector is projected to the subcycle closure polytope, since we have already achieved the goal of the separation problem.

          In some sense, from Theorem~\ref{srk:crowder} we derive the first shrinking rule of the motivation above and from Theorem~\ref{srk:mincut} the second shrinking rule. The condition that avoids the case~\ref{scenario:two}) of the Lemma~\ref{lemma:logical} is the hypothesis $x_{[u,v]} \geq \max\{y_u, y_v\}$ in the theorems. Actually, the hypothesis that $(y, x) \in P^G_A$ of the first rule can be replaced with the hypothesis that all the logical constraints associated with vertices $u$ and $v$ (excluding the one with $[u,v]$) are satisfied, which is a consequence of the hypothesis $x_{[u,v]} \geq \max\{y_u, y_v\}$. Let us address the next lemma as an intermediate step.

          \begin{lemma}\label{lemma:logical}
            Given a vector $(y, x) \in \mathbb{R}^{V \times E}$ that satisfies the degree equations, let $S=\{u,v\} \subset V$ be a subset such that $ x_{[u,v]} \geq \max\{y_u, y_v\}$. Then, if $(y, x) \notin P^G_A$, at least one of the following conditions is satisfied:
            \begin{enumerate}[i)]
              \item $(y[S], x[S]) \notin P^{G[S]}_A$, or
              \item if there exist violated logical constraints for $(y, x)$, these are associated with the edge $[u,v]$
            \end{enumerate}
          \end{lemma}
          \begin{proof}
            On the one hand, since $x(\{u,v\}:w) \geq x_{[u,w]}$ and $x(\{u,v\}:w) \geq x_{[v,w]}$ for all $w \in V-\{u, v\}$, every violated logical constraint for $(y,x)$ associated with the vertices in $V-\{u, v\}$ can be adapted to violated constraints for $(y[S],x[S])$. On the other hand, since $ x_{[u,v]} \geq \max\{y_u, y_v\}$ and the degree equations are satisfied, we have that $x_{[u,w]} \leq y_u$ and $x_{[v,w]} \leq y_v$ for all $w \subset V - \{u,v\}$. Therefore, if $(y[S], x[S]) \in P^{G[S]}_A$, the only possible violated logical constraints associated with the vertices of $S$ correspond with the edge $[u, v]$.
              
          \end{proof}

          The SEC inequalities~\eqref{cp:gsec} are defined for sets, $Q$, such that $3 \leq |Q| \leq |V| - 3$. However, if $\langle Q, u, v \rangle$ violates for $(y, x)$ the inequality of~\eqref{cp:gsec} but $|Q|=2$ or $|Q|=|V|-2$, then a violated logical constraint can be identified and therefore we also know that $(y, x) \notin P^G_{SEC}$. For instance, if $\langle \{u, w\}, u, v \rangle$ does not satisfy the inequality~\eqref{cp:gsec}, then $y_w < x_{uw}$ is a violated constraint. In the following proofs, the term violated SEC, embracing the cases $|Q[S]|=2$ and $|Q[S]|=|V[S]|-2$, refers to its associated violated logical constraint when required.

          \begin{theorem}[Rule S1]\label{srk:crowder}
            Given a vector $(y,x) \in \mathbb{R}^{V \times E}$ that satisfies the degree equations, let $u,v \in V$ be two vertices such that $x_{[u,v]} = y_u = y_v = c$. If there exists a vertex $w \in V-\{u,v\}$ such that $y_w \geq c$, then it is subcycle-safe to shrink $S = \{u,v\}$. 
          \end{theorem}
          \begin{proof}
            Assume the vector $(y,x)$ belongs to $P^G_{A}$, i.e., only violated SECs exists for $(y, x)$, otherwise the theorem is satisfied by Lemma~\ref{lemma:logical}.
            Let $\langle Q, r, t \rangle $ be a violated SEC for $(y, x)$, and without loss of generality, suppose that $S \cap Q \neq \emptyset$. The goal is to see that for a violated SEC for $(y, x)$, there is a violated SEC for $(y[S], x[S])$.

            First, let us suppose that $S \subset Q$, where $x[S](\delta(Q[S]))= x(\delta(Q))$ is satisfied by definition. The only case that is needed to check is when $r \in S$. Without loss of generality, suppose that $r = v$. By hypothesis $y_u=x_{[u,v]}$, so $2y_v=x(\delta(S))=2y[S](v)$ and $\langle Q[S], y[S](s), y[S](r) \rangle$ define the desired SEC for $(y[S], x[S])$.
            \begin{equation}
              x[S](\delta(Q[S]))= x(\delta(Q)) < 2 y_v + 2 y_t -2 = 2 y[S](s) + 2 y[S](t) -2
            \end{equation}

            Next, let us analyze the case $S\cap Q \neq \emptyset$ and $Q-S \neq \emptyset$. Without loss of generality, suppose that $u \in Q$ and $v,w \in V-Q$. The subcase that requires a special attention is when $r = u$ and $t = v$. Note that, since $(y, x)$ satisfies the degree equations and, also by hypothesis, $y_v = x_{[u,v]}$, we have that $x(v:V-Q) \leq x(v:Q)$, and therefore:
            \begin{subequations}
              \begin{alignat}{2}
                x[S](\delta(Q[S])) & = x(\delta(Q \cup S))  \\
                                   & = x(\delta(Q)) + x(\delta(v)) - 2 x(v:Q)\\
                                   & = x(\delta(Q)) + x(v:V-Q) - x(v:Q) \leq x(\delta(Q)) \\
                                   & < 2 y_r + 2 y_v -2 =  2 y_w + 2 y_t -2  = 2 y[S](r) + 2 y[S](w) -2
              \end{alignat}
            \end{subequations}
            Hence, there also exists a violated SEC (or logical constraint) for $(y[S],x[S])$ and the set $S$ is subcycle-safe to shrink.
              
          \end{proof}

          Clearly, the shrinking rule S1 dominates the rules C1 and C2 of Theorem~\ref{corol:srk}. For every scenario where rules C1 or C2 can be applied, rule S1 is also applicable, since the existence of $w \in V-\{u, v\}$ is determined by the vertex $t \in V-\{u, v\}$ in Theorem~\ref{corol:srk}. Moreover, rule C3 should not be combined with rule S1, since might exist vertices with the same $y$ value whose connecting edge has a greater value in the shrunk graph obtained by S1.

          \begin{theorem}[Rule S2]\label{srk:mincut}
            Given a vector $(y,x) \in \mathbb{R}^{V \times E}$ that satisfies the degree equations, let $u,v \in V$ be two vertices such that $ x_{[u,v]} > \max\{y_u, y_v\}$ then it is subcycle-safe to shrink $S=\{u,v\}$.
          \end{theorem}
          \begin{proof}
            The theorem is a direct consequence of Lemma~\ref{lemma:logical}.
              
          \end{proof}

          Note that, if $(y,x)\in P^G_A$ and $S$ is a shrinkable set obtained from Theorem~\ref{corol:srk}, then by Lemma~\ref{lemma:srkPa} we have that $x_e \leq \max\{y_u, y_v\}$ for every $e=[u,v] \in E[S]$. Hence, it only makes sense to use the rule S2 in combination with the rule S1.

          If a subcycle-safe rule is applied, we know that all the SECs have not vanished. However, new violated SECs for $(y[S], x[S])$ might have appeared, which cannot be adapted to a violated one for $(y, x)$. This situation would lead to identifying unnecessary cuts for $(y, x)$ and therefore to slowing down the separation algorithm (the cut generation part). It is reasonable to ask when the violated SECs for $(y[S], x[S])$ can be transformed to violated SECs for $(y, x)$ and when not. Let us define the mapping by $\pi_S: \mathcal{P}(V[S]) \rightarrow \mathcal{P}(V)$
          \begin{equation}
            \pi_S (Q)  = \left \{
              \begin{aligned}
          & Q - \{s\} \cup S&  \hspace{35mm} \mbox{if $s \in Q $} \\
          & Q & \mbox{otherwise}
            \end{aligned} \right.
          \end{equation}
          For a given $S$, the inverse, $\pi^{-1}_S$, of the mapping $\pi_S$ is the set shrinking defined in~\eqref{def:srkset}, i.e., $\pi^{-1}_S(Q)=Q[S]$.  We have that $Q = \pi^{-1}_S(\pi_S(Q))$ for all $Q\subset V[S]$ and $Q \subset \pi_S(\pi^{-1}_S(Q))$ for all $Q\subset V$. An important property of the mapping $\pi_S$, by the definition~\eqref{def:srkstar}, is that $x(\delta(\pi_S(Q))) = x[S](\delta(Q))$ for all $Q\subset V[S]$. In some cases, we will need to refer to the set obtained by unshrinking completely the contracted sets, where multiple shrinking might have been performed, e.g., $G[S_1][S_2]$. In such cases, we simplify the notation and denote $\pi(Q)$, e.g., $\pi(Q)=\pi_{S_1}(\pi_{S_2}(Q))$.

          When an inequality family is targeted in a separation problem, knowing the representation of such inequalities, as is the case for the SECs, is very valuable to study how an inequality is transformed when shrinking and unshrinking a set. Moreover, since $x(\delta(\pi_S(Q))) = x[S](\delta(Q))$ for all $Q\subset V[S]$, understanding the relationship between $y$ and $y[S]$ values is the key point to see how the violated SEC inequalities behave under the different shrinking rules.

          \begin{lemma}\label{lemma:uvysrk}
            Given a vector $(y, x) \in \mathbb{R}^{V \times E}$ that satisfies the degree equations and a subset $S=\{u, v\}$ of $V$. The following holds:
            \begin{enumerate}[i)]
              \item $y[S](v[S]) > y_v$ if $x_{[u, v]} < y_u$
              \item $y[S](v[S]) < y_v$ if $x_{[u, v]} > y_u$
              \item $y[S](v[S]) = y_v$ if $x_{[u, v]} = y_u$
            \end{enumerate}
          \end{lemma}
          \begin{proof}
            It is a consequence of the definition of $y[S]$ and the identity~\eqref{cp:eq}.
              
          \end{proof}

          \begin{lemma}\label{lemma:ygsec1}
            Under the hypotheses of Theorem~\ref{srk:crowder}, $y[S](v[S]) = y_v$ for all $v\in V$.
          \end{lemma}
          \begin{proof}
            For every $v\in V-S$, we have $y[S](v[S]) = y_v$ by definition. For $u,v\in S$, since $y_u=y_v= x_{[u,v]}$, we obtain the equality by Lemma~\ref{lemma:uvysrk}.
              
          \end{proof}

          \begin{lemma}\label{lemma:secrel}
            Let $G$ be an undirected graph, $(y, x) \in \mathbb{R}^{V \times E}$ be a vector and a vertex subset $S\subset V$. Suppose that $y[S](u) \leq y(v)$ for all  $u \in V[S]$ and $v\in \pi_S(u)$. Then, for each SEC for $(y[S], x[S])$ there exists at least one SEC as violated as it for $(y, x)$.
          \end{lemma}
          \begin{proof}
            Note that, if $r \in Q$ and $t \notin Q$ then $u \in \pi_S(Q)$ and $v \notin \pi_S(Q)$ for all $u \in \pi_S(r)$ and $v \in \pi_S(t)$. Let $\langle Q, r, t \rangle$ be a SEC inequality violated by $(y[S], x[S])$. Therefore, the SEC inequality $\langle \pi_S(Q) , u, v \rangle$ is violated by $(y, x)$ where $u \in \pi_S(r)$ and $v \in \pi_S(t)$.
            \begin{equation}
              x(\delta(\pi_S(Q))) - 2 y_u -2 y_v \leq x[S](\delta(Q)) - 2y[S](r) -2y[S](t) \quad u \in \pi_S(r) \textnormal{~and~} v \in \pi_S(t)
            \end{equation}
              
          \end{proof}

          \begin{corollary}\label{corol:sec}
            Let $G$ be an undirected graph and $(y, x) \in \mathbb{R}^{V \times E}$ be a vector. If $S$ is a shrinkable subset obtained by rules C1, C2, C3 or S1, then $(y, x) \notin P^G_{SEC}$ if and only if $(y[S], x[S]) \notin P^{G[S]}_{SEC}$.
          \end{corollary}
          \begin{proof}
            It is a consequence of Lemma~\ref{lemma:yc} and Lemma~\ref{lemma:ygsec1}.
              
          \end{proof}

          When rule S2 is applied, as a consequence of Lemma~\ref{lemma:uvysrk}, some vertices of the shrunk graph will have lower values than the original ones. Although, by the definition of subcycle-safe shrinking, all the violated SECs for $(y, x)$ are not vanished, we might lose some of them in the shrinking process. However, it could be interesting to identify and save those excluded violated SECs if possible. For that aim we consider a vector $m[S] \in \mathbb{R}^{V[S]}$ defined as $m[S](v) = \max \{y_u: u \in \pi_S(v)\}$. It is clear that if only the rules of Theorem~\ref{corol:srk} and the rule S1 are applied, $m[S](v)=y[S](v)$ for all $v \in V[S]$. Considering the vector $m[S]$, we evaluate a SEC $\langle Q, u, v\rangle$ for a given vector $(y[S], x[S])$ by the expression
          \begin{equation}
            x[S](\delta(Q)) - 2 m[S](u) - 2 m[S](v) \geq -2 \label{ineq:max}
          \end{equation}
          \noindent and only if this is violated, we save the SEC $\langle Q, u, v\rangle$ for $(y, x)$.

        \end{section}

        \begin{section}{Exact Separation Algorithms for SECs}\label{sec:sep-alg}
          In this section, we present two exact separation algorithms for SECs in cycle problems. Given a vector $(y, x)\in P^G_A$, an algorithm which finds violated SECs for $(x,y)$ is called a separation algorithm. A separation algorithm is called exact if it always finds violated inequalities when they exist, otherwise it is called heuristic. Let $\bar{G}=(\bar{V},\bar{E})$ be the support graph of the given vector $(y,x)$ where
          \begin{subequations}
            \begin{align}
              \bar{V} := \{v\in V: y_v >0\} \\
              \bar{E} := \{e\in E: x_e >0\}
            \end{align}
          \end{subequations}

          Before delving into the separation algorithms in depth, we need to make an observation which has important consequences for SEC separation problems in cycle problems.
          In the TSP, the $y$ values are fixed to 1, so the constraints in the family~\eqref{cp:gsec} only depend on the star-set value of subsets of vertices. For this reason, the SEC separation problem for the TSP is closely related with the minimum cut problem, particularly, the most violated SEC for $(y,x)$ is in correspondence with the global minimum cut of $\bar{G}$. However, in cycle problems in general, the SECs $\langle C,v,d \rangle$ obtained from the global minimum cut of $\bar{G}$, $x(C:V-C)$, might not be violated, although other violated SECs for $(y,x)$ can exist. This scenario is shown in the example in Figure~\ref{fig:why}. The global minimum cut in the figure is obtained by $C=\{4\}$ and because $|C|<3$, by definition~\eqref{cp:gsec}, there is no violated SEC inequality of type $\langle C,v,u \rangle$ (or equivalently of type $\langle V-C,v,u \rangle$).
          However, the SECs  $\langle \{2, 3, 8\}, 2, 6 \rangle$ (or $\langle \{1, 4, 5, 6, 7, 9\}, 6, 2 \rangle$),  $\langle \{2, 3, 4, 8\}, 2, 6 \rangle$ (or $\langle \{1, 5, 6, 7, 9\}, 6, 2 \rangle$) and $\langle \{2, 3, 4, 5, 8\}, 2, 6 \rangle$ (or $\langle \{1, 6, 7, 9\}, 6, 2 \rangle$) are violated for the vector $(y,x)$ represented in Figure~\ref{fig:why}.
          \begin{figure}[htb!]
            \centering
            \hspace*{0.2\columnwidth}
            \includegraphics[width=0.8\columnwidth]{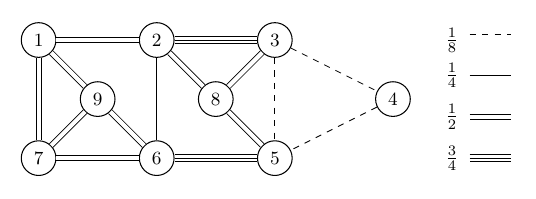}
            \caption{An example of a vector $(y, x)$ where the associated SEC with the global minimum cut of the support graph is not violated, while violated SECs for the vector exist. The edge values of the vector $(y, x)$ are detailed in the legend, while the vertex values are derived by the degree equations.}\label{fig:why}
          \end{figure}

          The straightforward exact algorithm to find violated SECs for $(y,x)$, consists of solving $\binom{|\bar{V}|}{2}$ number of $(s,t)$-minimum cuts problems on $\bar{G}$, one for each pair of different vertices, and then evaluating the associated inequality~\eqref{cp:gsec} using the $y$ values of the pair of vertices. When using the push-relabel algorithm in~\cite{golberg88} with highest-level selection and global relabeling heuristics to solve the $(s,t)$-minimum cut problems (or better said, to solve its dual: the $(s,t)$-maximum flow problems), the straightforward exact strategy has a $O(|\bar{V}|^4 \sqrt{|\bar{E}|})$ time complexity. Note that for cycle problems in general, the algorithm in~\cite{haoorlin1992} cannot be used to find the most violated SEC.\@ Although this algorithm solves the global minimum cut in $O(|\bar{V}|^2 \sqrt{|\bar{E}|})$ steps, which might be very useful, particularly for the TSP, in a general cycle problem the global minimum cut might not correspond with a violated SEC as shown above.

          The proposed separation algorithms in this paper, the Dynamic Hong's algorithm and the Extended Padberg-Gr{\"o}tschel algorithm, are two exact algorithms for cycle problems that run in $O(|\bar{V}|^3 \sqrt{|\bar{E}|})$. They are motivated by two observations made in~\cite{GTSP1997}. First, for a given pair of different vertices $u,v \in V$, the most violated SEC, $\langle Q, u, v \rangle$, corresponds to the subset $Q$ such that $(Q:V-Q)$ is a $(u,v)$-minimum cut. Secondly, for a given subset $Q$, the most violated SEC, $\langle Q, u, v \rangle$, corresponds to the vertices $u = \arg \max\{y_w: w \in Q\}$ and $v = \arg \max\{y_w: w \in V-Q\}$. The next two algorithms exploit these two observations, in order to guarantee that the most violated SEC for $(y,x)$ is identified.

          \begin{subsection}{Dynamic Hong's Exact Separation Algorithm}

            The Hong's exact approach, which emerged in the context of the TSP, consists of solving only $|\bar{V}|-1$ number of $(s,t)$-minimum cut problems, by fixing a random vertex, $s$, as the source of all the minimum cut problems, at the expense of possibly losing a subset of violated cuts, see~\cite{hong1972linear}.

            This exact approach can be extended for cycle problems, by selecting $s$ as a vertex of $\bar{V}$ with maximum $y$ value. Based on the second observation in~\cite{GTSP1997}, an $s$ selected this way will belong to the most violated SEC corresponding to every subset $Q$. However, since to define a SEC we need to select another vertex in $\bar{V}-\{s\}$, based on the first observation, we consider for each $t \in \bar{V}-\{s\}$ the subset $Q$ such that $(Q:V-Q)$ is a $(s,t)$-minimum cut. This shows that the extension of the Hong's approach for cycle problems is also an exact separation algorithm.

            Let us suppose that the vertices $\bar{V}=\{\bar{v}_1, \ldots, \bar{v}_{|\bar{V}|}\}$ are ordered decreasingly by $y$ and define the source $s_i=\bar{v}_1$ and the sink $t_i=\bar{v}_{i+1}$ for all $i\in\{1, \ldots, |\bar{V}|-1\}$. In~\cite{fischetti5} and~\cite{berube2009}, after each ($s_i$,$t_i$)-minimum cut, $(Q:V-Q)$, they increase the weight of the edge $[s_i,t_i]$ by $2-x(\delta(Q))$, in order to prevent collecting the same SEC in subsequent iterations. A disadvantage of this strategy is that the degree equations are not satisfied anymore. In Theorem~\ref{thm:gsec-exact} we achieve the same objective by shrinking the set $\{s_i,t_i\}$, with the extra feature of reducing the size of the graph for the following iterations.

            The underlying idea of Theorem~\ref{thm:gsec-exact} comes from the shrinking rule for minimum cut problems, Theorem $3.3$, in~\cite{Padberg1990a}. This theorem says that the edges having a value greater than or equal to the upper bound of the minimum cut can be contracted. However, this rule is not safe for SECs in cycle problems. For instance, based on Theorem $3.3$, in Figure~\ref{fig:why} we would shrink the set $\{2, 6\}$ because the value of the edge $[2,6]$ is equal to the global minimum cut value $x(C:V-C)$.
            However, because all the violated SECs in the figure consider the vertices $2$ and $6$ as disjoint ones, it is not safe to shrink the set $\{2,6\}$.

            \begin{lemma}\label{lemma:gsec-exact}
              Given a vector $(y,x) \in \mathbb{R}^{V \times E}$ that satisfies the degree constraints and four vertices $u, v, u^{'}, v^{'} \in \bar{V}$ such that $y_u+y_v \geq y_{u^{'}}+y_{v^{'}}$, let $(Q:\bar{V}-Q)$ be a $(u,v)$-minimum cut and $(Q^{'}:\bar{V}-Q^{'})$ be a $(u^{'},v^{'})$-minimum cut in $\bar{G}$. If $\langle Q^{'}, u^{'}, v^{'} \rangle$ is a strictly more violated SEC than $\langle Q, u, v \rangle$, then both $u,v$ vertices belong either to $Q^{'}$ or $\bar{V}-Q^{'}$.
            \end{lemma}
            \begin{proof}
              Suppose that $\langle Q^{'}, u^{'}, v^{'} \rangle$ is a strictly more violated SEC than $\langle Q, u, v \rangle$, then:
              \begin{subequations}
                \begin{align}
                  x(\delta(Q)) - 2y_u -2y_v +2 & >  x(\delta(S)) - 2y_{u^{'}} -2y_{v^{'}} + 2 \\
                  x(\delta(Q))  & >  x(\delta(S)) + 2y_u +2y_v- 2y_{u^{'}} -2y_{v^{'}} \\
                  x(\delta(Q))  & >  x(\delta(S))
                \end{align}
              \end{subequations}
              Since $x(\delta(Q))=x(Q:\bar{V}-Q)$ is the value of the $(u,v)$-minimum cut and $x(\delta(Q^{'}))$ is strictly smaller than it, then both $u$ and $v$ belong either to $Q^{'}$ or $V-Q^{'}$.
              
            \end{proof}

            \begin{theorem}[Rule S3]\label{thm:gsec-exact}
              Given a vector $(y, x)\in \mathbb{R}^{V \times E}$ satisfying the degree equations, consider $u, v \in \bar{V}$ such that $\min \{y_u, y_v\} \geq y_w$ for all $w\in \bar{V}-\{u, v\}$. Then, after solving the $(u, v)$-minimum cut problem and collecting, if any, the associated violated SECs, it is subcycle-safe to shrink $S=\{u, v\}$.
            \end{theorem}
            \begin{proof}
              The theorem is a direct consequence of Lemma~\ref{lemma:gsec-exact}.
              
            \end{proof}

            The dynamic Hong's algorithm is based on Theorem~\ref{thm:gsec-exact}, and it takes its name because the source, $s$, for the $(s,t)$-minimum cut problems might not be the same as in the classical approach. The algorithm works as follows: suppose that the vertices of $\bar{V}$ are ordered decreasingly by $y$, and set for the first minimum cut problem $s_1=\bar{v}_1$ and $t_1=\bar{v}_2$. Next, we solve the ($s_1$,$t_1$)-minimum cut problem, evaluate the obtained SEC candidates and, thereafter, shrink $\{s_1,t_1\}$. To proceed with the subsequent iteration, we need to know if the ordering of the vertices has changed after the $\{s_1,t_1\}$ shrinking, so we consider the Lemma~\ref{lemma:uvysrk}. When the logical constraint $x_{[s_1,t_1]} \leq y_{s_1}$ is satisfied, we have that $y[\{s_1,t_1\}](s_1[\{s_1,t_1\}]) \geq y_{t_1} \geq y_v$ for all $v\in \bar{V}-\{s_1,t_1\}$, and, hence, the vertex $s_1[\{s_1,t_1\}]$ will be ``again'' the source of the subsequent minimum cut problem. However, when $x_{[s_1,t_1]} > y_{s_1} $, it might happen that $y[\{s_1,t_1\}](s_1[\{s_1,t_1\}])  < y_v$ for some $v \in \bar{V}-\{s_1,t_1\}$. In this situation, after shrinking the set $\{s_1,t_1\}$, we will need to reorder the vertices of $\bar{V}[\{s_1,t_1\}]$  decreasingly by $y$ (rearrange  $s_1[\{s_1,t_1\}]$ in the set $\bar{V}$). So now, to proceed, we set as $s_2$ and $t_2$, the first two vertices of $\bar{V}[\{s_1,t_1\}]$, continue by solving the ($s_2$,$t_2$)-minimum cut problem, evaluating the possible violated SECs and shrinking $\{s_2,t_2\}$, and so on.

          \end{subsection}

          \begin{subsection}{Extended Padberg-Gr{\"o}tschel Exact Separation Algorithm}

            \cite{Padberg1985}, showed a different exact separation algorithm for SECs in the TSP, whose key component is the multitermal flow algorithm proposed in~\cite{gomoryhu1961}. A multitermal flow algorithm is solved, in turn, using the so-called Gomory-Hu tree, which can be constructed solving a $|\bar{V}|-1$ number of $(s,t)$-minimum cut problems.

            In~\cite{GTSP1997} it was mentioned that an analogue approach to the one given for the TSP might be used for the SECs in the cycle problems, but no details were given to illustrate how this approach should be extended. However, note that the adaptation of the Padberg-Gr{\"o}tschel approach for cycle problems is not trivial. The algorithm in~\cite{Padberg1985} for the TSP\@ relies on the correspondence between the most violated subtour elimination constraint for $(y,x)$ and the global minimum cut of $\bar{G}$, which is not always the case in general cycle problems (this might not even be violated while other exist).

            In cycle problems, Gomory-Hu trees were used to find violated SECs in~\cite{bauer2002} for the Cardinality Constrained Cycle Problem (CCCP) and in~\cite{jepsen2014} for the Capacited Profitable Tour Problem (CPTP).
            Nevertheless, in absence of details of the approach used to identify the violated SECs, we understand that in both papers the selected inequality corresponds with the global minimum cut.
            Therefore, these separation algorithms for SECs should be considered as heuristics.
            As far as we know, an exact extension for the Padberg-Gr{\"o}tschel separation algorithm for SECs in cycle problems has not been detailed in the literature.

            In order to extend the separation algorithm for cycle problems, we need to construct a Gomory-Hu tree, $T=(\bar{V},A_T)$, of the support graph $\bar{G}$ with weights $(y,x)$. However, unlike in the original approach, the tree $T$ has to be constructed as a directed rooted tree, where the root is set as a vertex of $\bar{V}$ with maximum $y$ value. Let us denote by $\Delta(v)$ the set of descendant vertices of $v\in \bar{V}$ and by $r$ the root of the tree $T$. We consider that every vertex is descendant of itself, i.e., $v \in \Delta(v)$.  Suppose that the arcs of $A_T$ are in the descendant orientation, and call $h_e$ the head vertex of an arc $a$. Given $a\in A_T$, we define
            \begin{subequations}
              \begin{align}
                u_a = & \arg \max \{y_v:v \in \Delta(h_a)\} \\
                v_a = & \arg \max \{y_v:v \in \bar{V}-\Delta(h_a)\}
              \end{align}
            \end{subequations}
            which identifies the vertices, $u_a$ and $v_a$, with the maximum $y$ value for each of the two connected components of the graph $(\bar{V}, A_T-\{a\})$. Note that, from the way that we have chosen the root, we can assume that $v_a=r$. Then, once the directed rooted Gomory-Hu tree is constructed, the violated SECs are collected in $O(\bar{V})$ computational time. With that aim, we check for each arc $a \in A_T$ ($|A_T|< |\bar{V}|$) if the inequality $w_a -2y_{u_a} -2y_r \geq -2$ is violated, being $w_a$ the weight of the arc $a$ in the Gomory-Hu tree $T$ representing the $(s,t)$-minimum cut for the two extreme vertices of the arc $a$. If this happens, the violated SEC is defined by $\langle \Delta(h_a), u_a, r \rangle$.

            Note that this can be done efficiently because the $u_a$ vertices of the arcs can be updated without an extra computational overhead. At every step of the Gomory-Hu algorithm, when a new arc is added to the tree, the descendant vertices are identified, which can be grasped to update the $u_a$ vertices. Also, with a proper implementation of the Gomory-Hu algorithm, it is possible to maintain the subset that contains the selected $r$ as the root of the subsequent trees. For more details, see the pseudocode in the Appendix A.

            In a similar way to the extension of Hong's approach, it can be shown that the extension of Padberg-Gr{\"o}tschel is exact for cycle problems. In this case, the root vertex $r$ plays the role of $s$, whereas each arc $a\in A_T$ identifies simultaneously a vertex in $V-\{r\}$, $t=h_a$, and its associated $(s,t)$-minimum cut. Furthermore, it goes one step beyond, based on the second observation, it considers $u_a$ instead of $h_a$. Hence, the number of violated cuts found by the extension of the classical Hong's approach is dominated by the extension of the Padberg-Gr{\"o}tschel approach.

            According to our experiments in Section~\ref{sec:exp}, the Extended Padberg-Gr{\"o}tschel approach consumes a much lower computational time than the Extended Hong approach, although both approaches have the same worst case running time complexity. This happens because the subsequent $(s,t)$-minimum cut problems are solved in subgraphs of $\bar{G}$ in the Gomory-Hu tree based approach. When the problem size increases, the time needed for the shrinking and unshrinking operations during the Gomory-Hu tree construction is insignificant compared to the time needed to solve the $(s,t)$-minimum cut problems. Therefore, in addition to potentially finding more violated SECs, the Extended Padberg-Gr{\"o}tschel is a faster exact separation algorithm than the Extended Hong's Algorithm.

            In Figure~\ref{fig:gomoryhu}, we illustrate the Extended Padberg-Gr{\"o}tschel approach to find the violated SECs for the vector $(y, x)$ defined in Figure~\ref{fig:why}. The weight $w_a$ of each $a \in A_T$ in the tree is detailed above the arcs, and the $y$ values of the vertices $u_a$ and $v_a$ are detailed inside a box, at the top and at the bottom respectively, near the head vertex of the arc. Two violated SECs are identified $\langle \{2,3,4,5,8\}, 2, 6 \rangle$ and $\langle \{2,3,8\}, 2, 6 \rangle$. Note that, if in this particular tree, the vertex $2$ is chosen to be the root, only the violated SEC $\langle \{1,6,7,9\}, 6, 2 \rangle$ (equivalent to $\langle \{2,3,8\}, 2, 6 \rangle$) is collected, which shows that the exact algorithm is sensible to the directed rooted Gomory-Hu tree construction.

            \begin{figure}[htb!]
              \begin{center}
                \includegraphics[width=0.7\columnwidth]{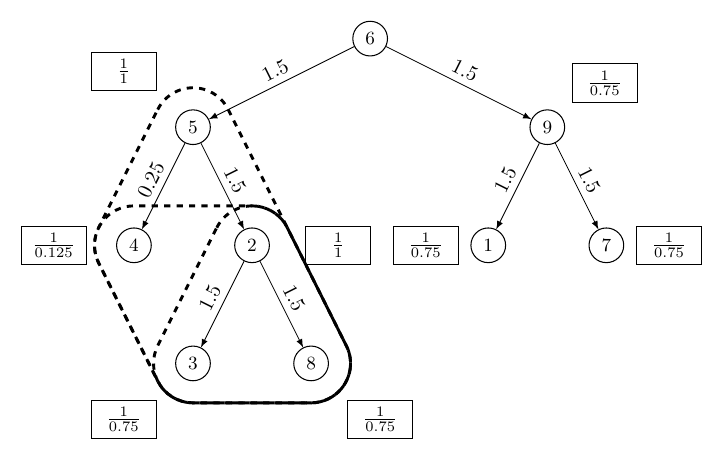}
                \caption{An example of the directed rooted Gomory-Hu tree for the SEC separation problem of Figure~\ref{fig:why}. The $u_a$ (below) and $v_a$ (above) values are detailed in the boxes. The arc weights are detailed next to the arcs.}\label{fig:gomoryhu}
              \end{center}
            \end{figure}
          \end{subsection}

          Although, the detailed approach until now always finds violated inequalities when they exist, extra violated SECs can be collected using a more exhaustive search whose cost is $O(|\bar{V}|^2)$. Observe that $x(\delta(\Delta(h_a) \cup \Delta(h_f))) \leq w_a+w_f$ for every $a,f \in A_T$. Then, we can define $y_{u(e,f)}=\max\{y_{u_a}, y_{u_f}\}$  and check if $w_a +w_f - 2y_{u(a, f)} -2y_r < -2$ for each pair arcs of $A_T$.
          This way, the violated SEC $\langle \{2,3,4,8\}, 2, 6 \rangle$ in Figure~\ref{fig:gomoryhu} can be identified. We have not made use of this kind of extra SECs in our experiments.

        \end{section}

        \begin{section}{Computational Results for Shrinking and Separation Algorithms for SECs}\label{sec:exp}

          In this section we describe the results of the computational experiments for the shrinking and the exact separation algorithms for SECs. These experiments have been designed with two goals in mind. First, to show the importance of the shrinking technique for cycle problems, and second, to evaluate the performance of different combination of shrinking and separation algorithms for SECs.

          The computational study of this section is inspired by two studies for the minimum cut algorithms:~\cite{junger2000} and~\cite{goldberg2001}. In both papers, the minimum cut algorithms are tested in instances originated, among others, from the solution of the TSP by a B\&C algorithm. Note that, as explained in Section~\ref{sec:sep-alg}, the global minimum cut algorithms tested in these papers are not suitable for our purpose.

          \cite{junger2000} studied the performance of different algorithms in combination with the shrinking rules defined for the minimum cut problems in~\cite{Padberg1990a}. Similarly, in this paper, we show the performance of the combination of shrinking rules and separation algorithms for SECs in cycle problems.~\cite{goldberg2001} compared different Gomory-Hu tree building strategies:~\cite{gusfield1990} implementation and three variants of the classical implementation. It was shown, for the SEC separation problem in the TSP, that the classical Gomory-Hu building based strategies outperform Gusfield's implementation, whereas they have not obtained significant differences among the variants of the classical implementation. The directed rooted Gomory-Hu tree algorithm presented in Section~\ref{sec:sep-alg} can be considered within the class of classical implementations.

          \begin{subsection}{Benchmark Instances}\label{sec:bench}
            The cycle problems could have a very large variety of origins, where the cycle constraints might be combined with additional constraints (e.g., a limit in the length of the cycle) and different objective functions (e.g., maximizing the profits and/or minimizing the length). These different natures of the cycle problems might vary the results obtained by each proposed strategy.
            However, we assume that in general terms the behaviour of the strategies for SECs is similar for all the cycle problems.
            So, instead of presenting an extensive comparison for different cycle problems, we focus our experiments on a well-known cycle problem, the Orienteering Problem (OP).

            With the purpose of evaluating our shrinking and separation algorithms for SECs, we have built the SEC separation instances by obtaining vectors $(y,x)\notin P^G_C$ during a B\&C algorithm for the OP\@. The OP instances are constructed based on the TSPLIB instances in~\cite{reinelt1991} following the approach in~\cite{fischetti5}. Particularly, we have chosen the TSPLIB instances selected in~\cite{goldberg2001}: pr76, att532, vm1084, rl1323, vm1748, rl5934, usa13509, d15112. Based on these 8 TSP instances, we have constructed 24 OP instances following the approach in the OP literature. The depot vertex is considered to be the first vertex of the TSPLIB instance, the maximum cycle length in the OP is set as half of the TSP value of the instance (values reported in~\cite{concorde}) and the profits of the vertices are generated in three different ways: Gen1, all the vertices have equal profit; Gen2, the scores are generated pseudorandomly; and Gen3, the vertices which are further from the depot vertex have a greater profit. Once the OP instances have been constructed, the SEC separation instances are generated by considering the first support graph during a B\&C algorithm for the OP which satisfies the degree constraints, the logical constraints and the connectivity. We have classified the instances into two equal-sized groups: Medium, instances whose original OP problem has less than 1500 vertices, and Large, the rest of the instances. All the used OP instances and SEC separation problem instances are available at~\cite{cpsrksec}.
          \end{subsection}

          \begin{subsection}{Shrinking Strategies for SECs}\label{sec:srk-strat}

            Relying on the results of Section~\ref{sec:srk} and Section~\ref{sec:gsec}, we have considered 5 different shrinking strategies for SECs.
            We have named the obtained strategies, by concatenating the names of the involved rules: C1, C1C2, C1C2C3, S1, S1S2. The pseudocodes of these strategies are detailed in Appendix A.

            In each strategy, each involved rule is applied exhaustively. For instance, for the rule C1, the hypotheses of Theorem~\ref{corol:srk} are checked for every possible set $S \subset \bar{V}$ and vertex $t \in \bar{V}-S$. Moreover, when a shrinkable set $S$ is found and shrunk, new shrinkable sets might appear in the graph obtained  after applying the shrinking. In order to handle these scenarios, we make use of a heap set, $H \subset \bar{V}$, which stores all the vertices that need to be checked to see whether they belong to a candidate $S$. For that, first, the set $H$ is initialized considering all the vertices of $\bar{V}$. During the search procedure, whenever the heap set $H$ is not empty, we draw one of its vertex, $v$, and consider it as contained in $S$. Then, we find neighbour vertices of $v$ that, if they incorporate to $S$, might make $S$ shrinkable. If a shrinkable set $S$ is found, first we remove the vertices in the set $S$ from $H$, and then we shrink the graph $\bar{G}$ and the vectors $(y, x)$ and $m$ (remember that $m_v = \max \{y_u: u \in \pi(v)\}$ for $v \in \bar{V}$). Immediately thereafter, we add the newly created vertex $s$ and its neighbours to the heap $H$. Additionally, when the support graph has vertices with value one, we check
            if violated SECs exist as suggested by Lemma~\ref{lemma:egt} and Theorem~\ref{thm:srkPa2}.

          \end{subsection}

          \begin{subsection}{Exact Separation Algorithms for SECs}\label{sec:sep-strat}
            We study the performance of four exact separation algorithms for SECs:
            \begin{enumerate}[i)]
              \item Algorithm \hypertarget{item:sep-eh}{EH}: Extended Hong's algorithm.
              \item Algorithm \hypertarget{item:sep-dh}{DH}: Dynamic Hong's algorithm.
              \item Algorithm \hypertarget{item:sep-dhi}{DHI}: Dynamic Hong's algorithm with internal shrinking.
              \item Algorithm \hypertarget{item:sep-epg}{EPG}: Extended Padberg-Gr{\"o}tschel algorithm.
            \end{enumerate}

            The Algorithm~\hyperlink{item:sep-eh}{EH} is the Hong separation algorithm extended for cycle problems in~\cite{GTSP1997}. The Algorithm~\hyperlink{item:sep-dh}{DH} refers to the Dynamic Hong separation algorithm explained in Section~\ref{sec:sep-alg}, i.e., after each minimum cut, we shrink the source and sink vertices based on rule S3. In Algorithm~\hyperlink{item:sep-dhi}{DHI}, in analogy to the approach used in~\cite{concorde} for the TSP, inside the~\hyperlink{item:sep-dh}{DH} separation algorithm, after shrinking the source and the sink vertices, we apply the given shrinking strategy to the newly obtained graph.
            The Algorithm~\hyperlink{item:sep-epg}{EPG} refers to the extended Padberg-Gr{\"o}tschel algorithm explained in Section~\ref{sec:sep-alg}.

            When a violated SEC, $\langle Q, u,v \rangle$, is found, we save in a repository only the $Q$ set of the violated SEC\@. During the whole separation procedure each $Q$ set is saved only once to avoid generating unnecessary cuts. Moreover, if $|Q| > |\bar{V}|/2$, we save $\bar{V}-Q$ instead of $Q$ in order to decrease memory resource requirements. Once the separation algorithm is completed, we generate the SEC cuts from the saved $Q$ sets in the following way: we consider for candidate vertices, $u$ and $v$, the vertices with maximum $y$ value inside $Q$, $M(Q)=\{u\in Q: y_u \geq y_v \;\forall v \in Q\}$, and outside $Q$, $M(\bar{V}-Q)=\{u\in \bar{V}-Q: y_u \geq y_v \;\forall v \in \bar{V}-Q\}$. Since the amount of generated SECs might be huge (producing memory problems), we consider only $k_{in}$ and $k_{out}$ randomly selected vertices from $M(Q)$ and $M(\bar{V}-Q)$, respectively. Note that in a cycle problem with depot, we have either $d \in M(Q)$ or $d \in M(\bar{V}-Q)$ for every $Q$, so it would be sufficient to select the depot instead of the randomly selected vertices. In other words, in these problems, it is enough to consider $u=d$ and $k_{in}=1$ if $d \in M(Q)$  and $v=d$ and $k_{out}=1$ otherwise. However, with the aim of obtaining insights about the SEC generation process in general cases, in the experiments, we have ignored that the OP is a cycle problem with depot.
          \end{subsection}

          The pseudocodes of the considered shrinking and separation strategies can be found in Appendix A
          and the source code of the implementation used for the experiments is publicly available at~\cite{cpsrksec}.

          \begin{subsection}{Results}\label{sec:exp-results}
            For the experiments, we have run 10 times each combination of shrinking and separation strategies with two objectives in mind: evaluate the influence of the random choices during the algorithm (ties are broken randomly when ordering $\bar{V}$; source and sink vertices are selected randomly in the Gomory-Hu tree construction) and obtain a better approximation of the running times. We have divided the process of finding the violated cuts into three parts: (1) the preprocess, which considers the shrinking carried out before the separation, (2) the separation, which consists of finding the $Q$ sets that define violated cuts, and (3) the generation of the violated SEC from the $Q$ sets. Since the SEC generation is closely related to the obtained $Q$ sets in the previous parts, and it is independent of the considered shrinking and separation strategies, we have limited the discussion of results to the preprocess and the separation parts.

            The computational results are summarized in two tables. In Table~\ref{tab:speedup}, we present the information about the graph simplification and the relative time needed by each combination of strategies compared to the reference strategy (Algorithm \hyperlink{item:sep-eh}{EH} with NO shrinking). In Table~\ref{tab:cxms}, we show the absolute values (on average) about the collected Q sets and the time needed (in milliseconds) by each combination of strategies. Although these tables give a general picture of the behaviour of the strategies, we consider that the results reflect what happens instance by instance. The detailed results of the experiments can be found in Appendix B.

            \begin{table}
              \begin{center}
\begin{tabular*}{\columnwidth}{@{\extracolsep{\fill}}llcccccc@{}}
\toprule
& & \multicolumn{2}{c}{Preprocess} & \multicolumn{4}{c}{Separation} \\ \cmidrule(lr){3-4}\cmidrule(lr){5-8}
& & \multicolumn{2}{c}{Graph Size} & \multicolumn{4}{c}{Speedup} \\ \cmidrule(lr){3-4}\cmidrule(lr){5-8}
Size & Shrinking & $\%|\bar{V}|$ & $\%|\bar{E}|$ & EH & DH & DHI & \multicolumn{1}{c}{EPG} \\ 
\midrule
Medium & NO  & $100.00$ & $100.00$ & $\phantom{00}1$ & $\phantom{00}9$ & $\phantom{00}9$ & $\phantom{00}9$ \\
 & C1  & $\phantom{0}42.55$ & $\phantom{0}50.61$ & $\phantom{00}6$ & $\phantom{0}29$ & $\phantom{0}23$ & $\phantom{0}19$ \\
 & C1C2  & $\phantom{0}39.73$ & $\phantom{0}46.40$ & $\phantom{00}7$ & $\phantom{0}32$ & $\phantom{0}27$ & $\phantom{0}20$ \\
 & C1C2C3  & $\phantom{0}39.73$ & $\phantom{0}46.40$ & $\phantom{00}7$ & $\phantom{0}33$ & $\phantom{0}25$ & $\phantom{0}20$ \\
 & S1  & $\phantom{0}22.88$ & $\phantom{0}26.43$ & $\phantom{0}16$ & $\phantom{0}57$ & $\phantom{0}51$ & $\phantom{0}28$ \\
 & S1S2  & $\phantom{0}21.26$ & $\phantom{0}24.53$ & $\phantom{0}17$ & $\phantom{0}60$ & $\phantom{0}53$ & $\phantom{0}27$ \\
\midrule Large & NO  & $100.00$ & $100.00$ & $\phantom{00}1$ & $\phantom{0}15$ & $\phantom{0}15$ & $\phantom{0}16$ \\
 & C1  & $\phantom{0}30.45$ & $\phantom{0}37.88$ & $\phantom{0}17$ & $107$ & $\phantom{0}74$ & $139$ \\
 & C1C2  & $\phantom{0}27.95$ & $\phantom{0}34.10$ & $\phantom{0}20$ & $122$ & $\phantom{0}86$ & $151$ \\
 & C1C2C3  & $\phantom{0}27.95$ & $\phantom{0}34.10$ & $\phantom{0}20$ & $121$ & $\phantom{0}80$ & $150$ \\
 & S1  & $\phantom{0}16.15$ & $\phantom{0}19.91$ & $\phantom{0}44$ & $221$ & $203$ & $215$ \\
 & S1S2  & $\phantom{0}14.34$ & $\phantom{0}17.43$ & $\phantom{0}53$ & $252$ & $227$ & $225$ \\
\bottomrule 
\end{tabular*}
                \caption{ Average speedup of the proposed algorithms using the Algorithm EH with no shrinking preprocess as a baseline.}\label{tab:speedup}
              \end{center}
            \end{table}

            In Table~\ref{tab:speedup} it can be seen that the graph is contracted considerably by means of the shrinking, especially in large problems. The largest contractions are achieved with strategy S1S2. An interesting point of the results is that with the rules derived from Theorem~\ref{corol:srk} (C1,C2,C3) the support graph is simplified significantly, which encourages us to apply the shrinking preprocess for other valid inequalities, such as combs.  Note that, rule C3 does not contract the graph more than what is already achieved by the combination of rules C2 and C3, see Section~\ref{sec:discus} for the discussion concerning this result.

            Regarding the speedup up obtained by the shrinking strategies, the results are clear and show the importance of performing the shrinking preprocess before the separation algorithms. If we observe the column related to Algorithm \hyperlink{item:sep-eh}{EH} in Table~\ref{tab:speedup}, the speedup obtained by each shrinking strategy is meaningful. In Medium instances, on average, the speedup is about $6$ times for the least aggressive strategy (C1), and $17$ times in Large instances. By means of the most aggressive strategy (S1S2) the speedup on average is $17$ for Medium-sized instances and $53$ in Large-sized instances.

            With respect to the time needed, the separation algorithms, Algorithm \hyperlink{item:sep-dh}{DH} and Algorithm \hyperlink{item:sep-epg}{EPG}, are both faster than the commonly used Algorithm \hyperlink{item:sep-eh}{EH}, which shows the relevance of the detailed exact separation algorithms in Section~\ref{sec:sep-alg}. If we compare Algorithm \hyperlink{item:sep-dh}{DH} and Algorithm \hyperlink{item:sep-epg}{EPG}, without considering any shrinking strategy, the speedups on average are similar
            (9 and 9 times, respectively) and Algorithm \hyperlink{item:sep-epg}{EPG} in larger instances (15 and 16 times, respectively). The table also suggests, based on the results of Algorithm \hyperlink{item:sep-dh}{DH} and Algorithm \hyperlink{item:sep-dhi}{DHI}, that it might not be convenient in the Dynamic Hong's separation algorithm to internally carry out extra shrinking procedures.

            Taking into account jointly the shrinking and separation strategies, the largest speedups are obtained when rules S1 and S2 are combined in the preprocess and, after that, alternatives to the standard Hong separation algorithms are used.
            In terms of running time, the Algorithm \hyperlink{item:sep-dh}{DH} with the S1S2 shrinking preprocess obtains the best results in the experiments, with an average speedup of $60$ in Medium-sized instances and $252$ in Large-sized instances. The results obtained by Algorithm \hyperlink{item:sep-epg}{EPG} with the S1S2 preprocess strategy are also very good, especially in Large-sized instances with an average speedup of $225$.

            \begin{table}[htb!]
              \centering
              \begin{tabular*}{\columnwidth}{@{\extracolsep{\fill}}llcccccccccc@{}}
                \toprule
& & \multicolumn{2}{c}{Preprocess} & \multicolumn{8}{c}{Separation} \\ \cmidrule(lr){3-4}\cmidrule(lr){5-12}
& & \multicolumn{2}{c}{All} & \multicolumn{2}{c}{EH} & \multicolumn{2}{c}{DH} & \multicolumn{2}{c}{DHI} & \multicolumn{2}{c}{EPG} \\ \cmidrule(lr){3-4}\cmidrule(lr){5-6}\cmidrule(lr){7-8}\cmidrule(lr){9-10}\cmidrule(lr){11-12}
                Size & Shrinking & \#Q & Time & \#Q & Time & \#Q & Time & \#Q & Time & \#Q & \multicolumn{1}{c}{Time} \\
                \midrule
                Medium & NO  & $\phantom{00}0.0$ & $\phantom{0}0.5$ & $\phantom{00}83.8$ & $\phantom{0}211.6$ & $\phantom{00}79.9$ & $\phantom{00}17.1$ & $\phantom{00}79.9$ & $\phantom{00}17.1$ & $\phantom{0}438.2$ & $\phantom{00}16.3$ \\
                       & C1  & $\phantom{00}0.0$ & $\phantom{0}0.8$ & $\phantom{00}27.8$ & $\phantom{00}30.2$ & $\phantom{00}58.4$ & $\phantom{000}5.2$ & $\phantom{00}58.4$ & $\phantom{000}6.5$ & $\phantom{0}149.0$ & $\phantom{000}7.8$ \\
                       & C1C2  & $\phantom{00}5.5$ & $\phantom{0}0.8$ & $\phantom{00}31.6$ & $\phantom{00}25.2$ & $\phantom{00}59.4$ & $\phantom{000}4.6$ & $\phantom{00}59.4$ & $\phantom{000}5.5$ & $\phantom{0}139.7$ & $\phantom{000}7.3$ \\
                       & C1C2C3  & $\phantom{00}5.5$ & $\phantom{0}0.9$ & $\phantom{00}31.6$ & $\phantom{00}25.5$ & $\phantom{00}59.4$ & $\phantom{000}4.5$ & $\phantom{00}59.4$ & $\phantom{000}5.9$ & $\phantom{0}139.8$ & $\phantom{000}7.4$ \\
                       & S1  & $\phantom{0}29.3$ & $\phantom{0}0.9$ & $\phantom{00}43.4$ & $\phantom{00}10.2$ & $\phantom{00}63.1$ & $\phantom{000}2.6$ & $\phantom{00}63.1$ & $\phantom{000}2.9$ & $\phantom{0}101.3$ & $\phantom{000}5.3$ \\
                       & S1S2  & $\phantom{0}35.1$ & $\phantom{0}0.9$ & $\phantom{00}48.8$ & $\phantom{000}9.5$ & $\phantom{00}69.0$ & $\phantom{000}2.5$ & $\phantom{00}69.9$ & $\phantom{000}2.8$ & $\phantom{00}98.3$ & $\phantom{000}5.3$ \\
                \midrule Large & NO  & $\phantom{00}0.0$ & $\phantom{0}9.9$ & $\phantom{0}679.4$ & $26578.2$ & $\phantom{0}372.6$ & $\phantom{}2140.0$ & $\phantom{0}372.6$ & $\phantom{}2140.0$ & $\phantom{}3395.1$ & $\phantom{}1828.8$ \\
                               & C1  & $\phantom{00}0.0$ & $\phantom{}22.5$ & $\phantom{0}154.2$ & $\phantom{}1513.4$ & $\phantom{0}266.8$ & $\phantom{0}203.7$ & $\phantom{0}266.8$ & $\phantom{0}295.8$ & $\phantom{0}756.6$ & $\phantom{0}146.7$ \\
                               & C1C2  & $\phantom{0}17.0$ & $\phantom{}22.8$ & $\phantom{0}166.8$ & $\phantom{}1320.0$ & $\phantom{0}271.7$ & $\phantom{0}179.3$ & $\phantom{0}271.7$ & $\phantom{0}257.2$ & $\phantom{0}717.9$ & $\phantom{0}135.2$ \\
                               & C1C2C3  & $\phantom{0}16.8$ & $\phantom{}23.2$ & $\phantom{0}166.6$ & $\phantom{}1321.0$ & $\phantom{0}271.5$ & $\phantom{0}181.0$ & $\phantom{0}271.5$ & $\phantom{0}277.1$ & $\phantom{0}717.8$ & $\phantom{0}136.2$ \\
                               & S1  & $169.2$ & $\phantom{}25.1$ & $\phantom{0}225.4$ & $\phantom{0}515.4$ & $\phantom{0}287.0$ & $\phantom{00}95.4$ & $\phantom{0}287.0$ & $\phantom{0}103.8$ & $\phantom{0}507.1$ & $\phantom{00}94.7$ \\
                               & S1S2  & $248.8$ & $\phantom{}25.3$ & $\phantom{0}293.1$ & $\phantom{0}427.2$ & $\phantom{0}372.2$ & $\phantom{00}83.5$ & $\phantom{0}374.3$ & $\phantom{00}91.5$ & $\phantom{0}528.0$ & $\phantom{00}91.1$ \\
                               \bottomrule
              \end{tabular*}
              \caption{On average, the number of $Q$ sets found and the time needed by strategy and size.}\label{tab:cxms}
            \end{table}

            Apart from the running time, an aspect to consider when making a choice about the separation algorithm is the number of violated cuts found. As we have already mentioned, in the cycle problems, the number of collected violated SECs is closely related with the Q sets obtained by the separation algorithms.
            Therefore, we have measured the obtained amount of Q sets instead of the number of violated SECs. In Table~\ref{tab:cxms}, the average number of Q sets and time of each combination of strategies is shown.

            The first aspect to note is that, by means of the shrinking preprocess, which is considerably faster than the exact separation procedure, we are able to find violated SECs in many instances (via Theorem~\ref{thm:srkPa2} and Lemma~\ref{lemma:egt}). These violated SECs might be enough for the separation goal and, in practice, we could skip the exact separation algorithm if violated inequalities are found in the preprocess. In the separation process, in general, the largest amount of Q sets are obtained by Algorithm \hyperlink{item:sep-epg}{EPG}, as was anticipated theoretically in Section~\ref{sec:sep-alg}. Note that, the quantity of obtained $Q$ sets is sensitive to the randomness of the shrinking and separation strategies (it can be concluded because $\#Q$ is not always an integer).

            In the view of these results, the S1S2 shrinking strategy is the best choice to use as the preprocess of SEC separation algorithms. Bearing in mind both the time and the obtained amount of $Q$ sets, either Algorithm \hyperlink{item:sep-dh}{DH} or Algorithm \hyperlink{item:sep-epg}{EPG} might be a good choice as the separation algorithm. However, it is not clear from these results which of the two exact approaches should be used in practice. It probably depends on the nature and the size of the cycle problem under consideration.

          \end{subsection}

          \begin{subsection}{Discussion}\label{sec:discus}

            Finally, we would like to open a discussion about the following concerns as a consequence of the computational results. It might be helpful, to look at the detailed computational results in Appendix B.
            to understand the motivation behind the discussion below.

            In Figure~\ref{fig:sec1}, an example of a vector $(y, x) \in P^G_A$ was shown where rule C3 can be applied but rules C1 nor C2 cannot. However, in the experiments, although rule C3 has been applied in some instances, we have not obtained any situation in which rule C3 was able to simplify the support graph more than with the rest of the rules.

            An open question is then to explain why rule C3 does not improve the results obtained by means of the rules C1 and C2. We believe that this is related with the planarity property of the support graphs, which is satisfied in the considered instances. Note that the graph in the example of Figure~\ref{fig:sec1} is not planar because the complete graph of 5 vertices, $K_5$, is a subgraph of it.
            \begin{conjecture}\label{conj:c3}
              Given a graph $G$, let $(y,x) \in P^G_A$ be a vector. If the support graph $\bar{G}$ of $(y,x)$ is planar, then the combination of the rules $C1$ and $C2$ dominate the rule $C3$.
            \end{conjecture}
            Note that the rules C1, C2, and C3 induce a contraction of an edge (a sequence of contractions for C3), which is a closed operation in planar graphs. Therefore, if $\bar{G}$ is planar then $\bar{G}[S]$ is also planar for every subset $S$ obtained from these rules. While working with the OP, we have empirically seen that in geometrical instances the support graph obtained within a B\&C is planar most of the time.

            Another interesting fact that can be extracted from the experiments is that the number of vertices and edges in the shrunk graph (the final result) is independent of the ordering of the considered rules and the shrinkable sets. This suggests the idea that the obtained shrunk graphs are isomorphic.
            \begin{conjecture}\label{conj:iso}
              Given a graph $G$, let $(y,x) \in P^G_A$ be a vector and SRK $\in \{$C1, C1C2, C1C2C3, S1, S1S2$\}$ be a fixed shrinking strategy, then the graphs obtained by applying SRK to $(y,x)$ are isomorphic.
            \end{conjecture}
            If the conjecture is true, the complexity of the separation algorithm carried out in the shrunk graph does not depend on the different implementations of a shrinking strategy. As a consequence, in the future, we might focus on identifying the implementations of the shrinking strategies that might obtain the largest amount of $Q$ sets, especially for the preprocess, e.g., by reordering the vertices in the heap.

          \end{subsection}
        \end{section}

        \begin{section}{Conclusions and Future Work}

          In this paper, for cycle problems, we have successfully generalized the global (C1, C2 and C3) and SEC specific (S1, S2 and S3) shrinking rules proposed in the literature of the TSP\@. The obtained computational results for the shrinking in the OP are remarkable and, hence, very promising for other cycle problems. The results clearly show that the shrinking technique considerably improves the running time of the separation algorithm for SECs.
          This opens the possibility to investigate in two directions in cycle problems: (1) studying the shrinking for other valid cycle inequalities of the OP (e.g., combs) and (2)  evaluating for other cycle problems the shrinking technique in SEC separation problems.

          Part of the paper focuses on exact SEC separation algorithms for cycle problems. We have extended from the TSP two exact algorithms (Algorithm \hyperlink{item:sep-dh}{DH} and Algorithm \hyperlink{item:sep-epg}{EPG}). The proposed separation algorithms were shown to be more efficient in the OP than the exact algorithm used so far in the literature (the adaptation of the classical Hong's approach). The importance of the detailed extension of the Padberg-Gr{\"o}tschel approach, Algorithm \hyperlink{item:sep-epg}{EPG}, lies in the fact that in cycle problems, in general, the global minimum cut of a support graph might not generate a violated SEC, while violated SECs in the same graph exist.
          An example is given
          where this claim is shown,
          which implies that the adaptions of the Padberg-Gr{\"o}tschel approach used so far in the literature of cycle problems should be viewed as heuristic separation algorithms.
          Therefore, this might be the first exact extension of the Padberg-Gr{\"o}tschel approach in the literature for cycle problems.

        \end{section}

        {\textbf{Acknowledgements}}  \ \
        The first and third authors are partially supported by the projects BERC 2018{-}2021 (Basque Government) and by SEV{-}2017{-}0718 (Spanish Ministry of Economy and Competitiveness).
        The first author is also supported by the grant BES{-}2015{-}072036 (Spanish Ministry of Economy and Competitiveness) and the project ELKARTEK (Basque Government).
        The second author is supported by IT{-}1252{-}19 (Basque Government) and PPG17/32 and GIU17/011 (University of the Basque Country).
        The third author is also supported by IT1244{-}19 (Basque Government) and TIN2016{-}78365R (Spanish Ministry of Science and Innovation).
        This article has been partially written while the first author was at Operational Research Group at University of Brescia as a visiting PhD student under the supervision of Prof.\ Speranza. He wishes to sincerely thank the OR group and the university for their warm hospitality and an excellent atmosphere.
        We gratefully acknowledge the authors of the TSP solver Concorde for making their code available to the public, since it has been the working basis of our implementations.

        \bibliographystyle{plainnat}

  \clearpage
  \appendix
  \renewcommand*{\thesection}{\Alph{section}}

  \section*{Appendices}\label{appendix}

  \section{Pseudocodes of the Shrinking and Separation Strategies}\label{appendix:pseudo}
  \setcounter{table}{0}
  \renewcommand*{\thetable}{A.\arabic{table}}

  In this appendix, we detail the pseudocodes of the shrinking and separation strategies used in the computational experiments for Section~\ref{sec:exp}. These strategies are combinations of the shrinking rules proposed in Section~\ref{sec:srk} and Section~\ref{sec:gsec}, and the exact separation algorithms proposed in Section~\ref{sec:sep-alg}.

  The pseudocodes should be considered as illustrations of the implementations of strategies whose aim is to help the reader to understand how the strategies work.
  The source code in C of the computational implementations is available at~\cite{cpsrksec}. In Table~\ref{tab:symbol}, we detail the meaning of the symbols used in the pseudocodes.

  \begin{table}[htb!]
    \centering
    \begin{tabular}{cll}
      \toprule
      \multicolumn{2}{c}{Symbol} & \multicolumn{1}{c}{Meaning} \\
      \midrule
      $G= (V, E)$& & Input graph of the cycle problem \\
      $\bar{G}= (\bar{V},\bar{E})$&  & Support graph \\
      $(y,x)$ &  $\in P^G_A$ & A solution of the $LP_0$ \\
      $m$ & $\in \mathbb{R}_{+}^{\bar{V}}$ & A vector where $m_v=\max\{y_u:u\in \pi(v)\}$ \\
      $H$ & $\subset \bar{V}$ & Heap: vertices remaining to check \\
      $S$ & $\subset \bar{V}$ & A subset candidate for the shrinking \\
      $Q$ & $\subset V$ & A subset of $V$ \\
      $\bar{Q}$ & $\subset \bar{V}$ & A subset of $\bar{V}$ \\
      $\mathcal{Q}$ & $\subset \mathcal{P}(V)$ & List of $Q$ sets of $V$ \\
      $\mathcal{L}$ & & List of violated SECs \\
      $D$ & $\subset \bar{V}$ & Set of fixed vertices. In a cycle problem with depot: $D=\{d\}$ \\
      $O$ & $\subset \bar{V}$ & Set of vertices with value one \\
      $(k_{in}\times k_{out})$ & $\in \mathbb{N}_{+} \times \mathbb{N}_{+}$ & Maximum vertices (inside and outside) considered when \\
                               & & generating the violated SECs from the $Q$ sets \\
      $T= (V, A_T)$ &  & A directed rooted tree \\
      $parent$ & $V \rightarrow V$ & Successive parent of each $v$ in the tree \\
      $child$ & $V \rightarrow V$ & Successive children of each $v$ in the tree \\
      $w$ & $\in \mathbb{R}_{+}^{A_T}$ & Weights of the arcs of the Gomory-Hu tree \\
      $G^{*}= (V^{*},E^{*})$&  & Generic graph used in the Gomory-Hu tree construction \\
      \bottomrule
    \end{tabular}
    \caption{A summary of the symbols used in the pseudocodes}\label{tab:symbol}
  \end{table}

  \subsection{Shrinking Strategies}

  The shrinking strategies are combinations of the shrinking rules of Section~\ref{sec:srk} and Section~\ref{sec:gsec}. In total, 5 different shrinking strategies for SECs are obtained:  C1, C1C2, C1C2C3, S1 and S1S2. The~\ref{alg:srkupdate} procedure refers to a process performed every time a set is shrunk.

  \begin{algorithm}[htb!]
    \SetAlgoRefName{SHRINK/UPDATE}
    \caption{Shrink graph and vectors. Save $Q$ sets. Update heap.}\label{alg:srkupdate}
    \SetKwInOut{Input}{input}\SetKwInOut{Output}{output}
    \Input{$\bar{G}$, $(y, x)$, $m$, $H$, $S$ and $\mathcal{Q}$}
    \Output{$\bar{G}$, $(y, x)$, $m$, $H$, $s$ and $\mathcal{Q}$}

    $\bar{G} \leftarrow \bar{G}[S]$\;
    $(y, x) \leftarrow (y[S], x[S])$\;
    $m \leftarrow m[S]$\;
    $H \leftarrow H[S]$\;
    $O \leftarrow \{v \in \bar{V}: m_v \geq 1\}$\;
    \For{$n \in N(s) $}{

      \If{$y_n < x_{[n,s]}$}
      {\For{$r \in O$}{
          \If{$r \neq s$}
          {\If{$\langle \{s,n\}, s, r \rangle$ violates~\eqref{ineq:max}}
            {
              $Q \leftarrow \{\pi(\{s,n\})\}$\;
              \If{$|Q| > |V|/2$}{
                $Q \leftarrow V -Q$\;
              }
              $\mathcal{Q} \leftarrow \mathcal{Q} \cup \{Q\}$\;
              \textbf{goto} line~\ref{gotosrkupdate}\;
            }
          }
        }
      }\label{gotosrkupdate}

      $H \leftarrow H \cup \{n\}$\;
    }
  \end{algorithm}

  \begin{algorithm}[htb!]
    \SetAlgoRefName{C1}
    \caption{Shrinking: Rule C1}\label{alg:C1}
    \SetKwInOut{Input}{input}\SetKwInOut{Output}{output}
    \Input{$\bar{G}$, $(y,x)$, $m$, $H$ and $\mathcal{Q}$}
    \Output{$\bar{G}$, $(y, x)$, $m$, $H$ and $\mathcal{Q}$}

    \While{$|H|\neq \emptyset$}{
      Select a vertex $u \in H$\;
      $H \leftarrow H - \{u\}$\;
      $c \leftarrow y_{u}$\;
      \For{$v \in N(u)$}{

        \If{$ y_{v} = c $ and $x_{[u, v]} = c$}{
          \For{$t \in N(v)-\{u\}$}
          {
            \If{$y_{t} = c$ and $x_{[v, t]}=c$}
            {
              $S \leftarrow \{u, v\}$\;
              \ref{alg:srkupdate} ($\bar{G}, (y, x), m, H, S, \mathcal{Q}$)\;
              \textbf{goto} line~\ref{found:c1}\;
            }
          }
        }
      }\label{found:c1}
    }
  \end{algorithm}

  \begin{algorithm}[htb!]
    \SetAlgoRefName{C1C2}
    \caption{Shrinking: Rule C1 and Rule C2}\label{alg:C1C2}
    \SetKwInOut{Input}{input}\SetKwInOut{Output}{output}
    \Input{$\bar{G}$, $(y,x)$, $m$, $H$  and $\mathcal{Q}$}
    \Output{$\bar{G}$, $(y, x)$, $m$, $H$ and $\mathcal{Q}$}

    \While{$|H|\neq \emptyset$}{
      Select a vertex $u \in H$\;
      $H \leftarrow H - \{u\}$\;
      $c \leftarrow y_{u}$\;

      \For{$v \in N(u) $}{

        \If{$ y_{v} = c$ and $x_{[u, v]} = c$}{
          \For{$t \in N(v)-\{u\} $}{
            \If{$y_{t} = c$ and $x_{[u, t]} + x_{[v, t]} = c$}{
              $S \leftarrow \{u, v\}$\;
              \ref{alg:srkupdate} ($\bar{G}, (y, x), m, H, S, \mathcal{Q}$)\;
              \textbf{goto} line~\ref{found:c1c2}\;
            }
          }
        }
      }\label{found:c1c2}
    }
  \end{algorithm}

  \begin{algorithm}[htb!]
    \SetAlgoRefName{C1C2C3}
    \caption{Shrinking: Rule C1, C2 and C3}\label{alg:C1C2C3}
    \SetKwInOut{Input}{input}\SetKwInOut{Output}{output}
    \Input{$\bar{G}$, $(y,x)$, $m$, $H$ and $\mathcal{Q}$}
    \Output{$\bar{G}$, $(y, x)$, $m$, $H$ and $\mathcal{Q}$}

    \While{$|H| \neq \emptyset$}{
      Select a vertex $u \in H$\;
      $H \leftarrow H - \{u\}$\;
      $c \leftarrow y_{u}$\;

      \For{$v \in N(u)$}{

        \If{$ y_{v} = c $ and $x_{[u, v]} = c$}{
          \For{$t \in N(v)-\{u\} $}{
            \If{$ y_{t} = c $ and $x_{[u, t]} + x_{[v, t]}=c$}{
              $S \leftarrow \{u, v\}$\;
              \ref{alg:srkupdate} ($\bar{G}, (y, x), m, H, S,  \mathcal{Q}$)\;
              \textbf{goto} line~\ref{found:c1c2c3}\;
            }
          }
          \For{$w \in N(v)-\{u\} $}{
            \If{$x_{[u, t]} + x_{[u, w]} + x_{[v,w]} = 2c$}{
              \For{$t \in N(w)-\{v,u\}$}{
                \If{$y_{t} = c$ and $x_{[u, t]} + x_{[v, t]} = c$}{
                  $S \leftarrow \{u, v, w\}$\;
                  \ref{alg:srkupdate} ($\bar{G}, (y, x), m, H, S, \mathcal{Q}$)\;
                  \textbf{goto} line~\ref{found:c1c2c3}\;
                }
              }
            }
          }
        }
      }\label{found:c1c2c3}
    }
  \end{algorithm}

  \begin{algorithm}[htb!]
    \SetAlgoRefName{S1}
    \caption{Shrinking: Rule S1}\label{alg:S1}
    \SetKwInOut{Input}{input}\SetKwInOut{Output}{output}
    \Input{$\bar{G}$, $(y,x)$, $m$, $H$ and $\mathcal{Q}$}
    \Output{$\bar{G}$, $(y, x)$, $m$, $H$ and $\mathcal{Q}$}

    \While{$|H|\neq \emptyset$}{
      Select a vertex $u \in H$\;
      $H \leftarrow H - \{u\}$\;
      $c \leftarrow y_{u}$\;
      \For{$v \in N(u) $}{

        \If{$ y_{v} = c $ and $x_{[u, v]} = c$}
        {
          \If{$\exists w \in \bar{V}-\{u,v\}$ such that $y_w \geq c$}
          {$S \leftarrow \{u, v\}$\;
            \ref{alg:srkupdate} ($\bar{G}, (y, x), m, H, S, \mathcal{Q}$)\;
            \textbf{goto} line~\ref{found:s1}\;
          }
        }
      }\label{found:s1}
    }
  \end{algorithm}

  \begin{algorithm}[htb!]
    \SetAlgoRefName{S1S2}
    \caption{Shrinking: Rule S1 and S2}\label{alg:S1S2}
    \SetKwInOut{Input}{input}\SetKwInOut{Output}{output}
    \Input{$\bar{G}$, $(y,x)$, $m$, $H$, $D$ and $\mathcal{Q}$}
    \Output{$\bar{G}$, $(y, x)$, $m$, $H$, $D$ and $\mathcal{Q}$}

    \While{$|H|\neq \emptyset$}{
      Select a vertex $u \in H$\;
      $H \leftarrow H - \{u\}$\;
      $c \leftarrow y_{u}$\;
      \For{$v \in N(u) $}{

        \uIf{$y_{v} = c $ and $x_{[u, v]} = c$}
        {
          \If{$\exists w \in \bar{V}-\{u,v\}$ such that $y_w \geq c$}
          {$S \leftarrow \{u, v\}$\;
            \ref{alg:srkupdate} ($\bar{G}, (y, x), m, H, S, \mathcal{Q}$)\;
            \textbf{goto} line~\ref{found:s1s2}\;
          }
        }
        \ElseIf{$x_{[u, v]}>y_u$ and $x_{[u, v]}>y_v$}
        {
          {$S \leftarrow \{u, v\}$\;
            \ref{alg:srkupdate} ($\bar{G}, (y, x), m, H, S, \mathcal{Q}$)\;
            \textbf{goto} line~\ref{found:s1s2}\;
          }
        }
      }\label{found:s1s2}
    }
  \end{algorithm}

  \clearpage
  \subsection{Exact SEC Separation Strategies}
  The exact separation strategies detailed in this appendix refer to the separation algorithms used for the experiments in Section~\ref{sec:exp}. We assume that the vertex set $\bar{V}=\{v_1, \ldots, v_{|\bar{V}|}\}$ is an ordered set. The~\ref{alg:cutgen} algorithm is the procedure detailed in Section~\ref{sec:exp} to generate the most violated SECs corresponding to set $Q$ given the parameter $(k_{in}, k_{out}) \in \mathbb{N}_{+} \times \mathbb{N}_{+}$. The vector $(k_{in}, k_{out})$ represents the maximum amount of vertices that are considered inside and outside $Q$. Note that,~\ref{alg:cutgen} is defined to select, for each inside vertex, a number of $k_{out}$ different random outside vertices to maximize the randomness of the obtained violated SECs.

  \begin{algorithm}[htb!]
    \SetAlgoRefName{EH}
    \caption{Extended Hong's exact separation algorithm}\label{alg:H}
    \SetKwInOut{Input}{input}\SetKwInOut{Output}{output}
    \Input{$\bar{G} $, $(y,x)$, $D$ and $(k_{in}, k_{out})$}
    \Output{A list $\mathcal{L}$ of violated SECs}

    $\bar{V} \leftarrow$ sort $\bar{V}$ decreasingly by $y$;
    $m \leftarrow y$\;
    $H \leftarrow \bar{V}$\;
    Apply shrinking strategy ($\bar{G}, (y, x), m, H, D, \mathcal{Q}$)\;

    \While{$|\bar{V}|>1$}{
      $Q \leftarrow$ $(v_1, v_2)$-minimum\, cut in the graph $\bar{G}$\;
      \If{$\langle Q, v_1, v_2 \rangle $ violates~\eqref{ineq:max}}{
        \If{$|Q| > |V|/2$}{
          $Q \leftarrow V -Q$\;
        }
        $\mathcal{Q} \leftarrow \mathcal{Q} \cup \{\pi(Q)\}$\;
      }
    }
    $\mathcal{L} \leftarrow$~\ref{alg:cutgen} ($\bar{G}, (y, x), D, \mathcal{Q}, (k_{in}, k_{out})$)\;
  \end{algorithm}

  \begin{algorithm}[htb!]
    \SetAlgoRefName{DH}
    \caption{Dynamic Hong's exact separation algorithm }\label{alg:DH}
    \SetKwInOut{Input}{input}\SetKwInOut{Output}{output}
    \Input{$\bar{G} $, $(y,x)$, $D$ and $(k_{in}, k_{out})$}
    \Output{A list $\mathcal{L}$ of violated SECs}

    $\bar{V} \leftarrow$ sort $\bar{V}$ decreasingly by $y$\;
    $m \leftarrow y$\;
    $H \leftarrow \bar{V}$\;

    Apply shrinking strategy ($\bar{G}, (y, x), m, H, \mathcal{Q}$)\;

    \While{$|\bar{V}|>1$}{
      $Q \leftarrow$ $(v_1, v_2)$-minimum\, cut in the graph $\bar{G}$\;
      \If{$\langle Q, v_1, v_2 \rangle $ violates~\eqref{ineq:max}}{
        \If{$|Q| > |V|/2$}{
          $Q \leftarrow V -Q$\;
        }
        $\mathcal{Q} \leftarrow \mathcal{Q} \cup \{\pi(Q)\}$\;
      }

      \eIf{$x_{[v_1,v_2]}>y_{v_2}$}{
        reorder $\leftarrow 1$\;
        }{
        reorder $\leftarrow 0$\;
      }

      $S \leftarrow \{v_1, v_2\}$\;
      $\bar{G} \leftarrow \bar{G}[S]$\;
      $(y, x) \leftarrow (y[S], x[S])$\;
      $m \leftarrow m[S]$\;
      \If{reorder}{
        $\bar{V} \leftarrow$ sort $\bar{V}$ decreasingly by $y$\;
      }
    }
    $\mathcal{L} \leftarrow$~\ref{alg:cutgen} ($\bar{G}, (y, x), D, \mathcal{Q}, (k_{in}, k_{out})$)\;
  \end{algorithm}

  \begin{algorithm}[htb!]
    \SetAlgoRefName{DHI}
    \caption{Dynamic Hong with extra shrinking separation algorithm}\label{alg:DHE}
    \SetKwInOut{Input}{input}\SetKwInOut{Output}{output}
    \Input{$\bar{G} $, $(y,x)$, $D$ and $(k_{in}, k_{out})$}
    \Output{A family $\mathcal{Q}$ of violated SECs}

    $\bar{V} \leftarrow$ sort $\bar{V}$ decreasingly by $y$\;
    $m \leftarrow y$\;
    $H \leftarrow \bar{V}$\;

    Apply shrinking strategy ($\bar{G}, (y, x), m, H, \mathcal{Q}$)\;

    \While{$|\bar{V}|>1$}{
      $Q \leftarrow$ $(v_1, v_2)$-minimum\, cut in the graph $\bar{G}$\;
      \If{$\langle Q, v_1, v_2 \rangle $ violates~\eqref{ineq:max}}{
        \If{$|Q| > |V|/2$}{
          $Q \leftarrow V -Q$\;
        }
        $\mathcal{Q} \leftarrow \mathcal{Q} \cup \{\pi(Q)\}$\;
      }

      \eIf{$x_{[v_1,v_2]}>y_{v_2}$}{
        reorder $\leftarrow 1$\;
        }{
        reorder $\leftarrow 0$\;
      }

      $S \leftarrow \{v_1, v_2\}$\;
      \ref{alg:srkupdate} ($\bar{G}, (y, x), m, H, S, \mathcal{Q}$)\;
      Apply shrinking strategy ($\bar{G}, (y, x), m, H, \mathcal{Q}$)\;
      \If{reorder}{
        $\bar{V} \leftarrow$ sort $\bar{V}$ decreasingly by $y$\;
      }
    }
    $\mathcal{L} \leftarrow$~\ref{alg:cutgen} ($\bar{G}, (y, x), D, \mathcal{Q}, (k_{in}, k_{out})$)\;
  \end{algorithm}

  \begin{algorithm}[htb!]
    \SetAlgoRefName{EPG}
    \caption{Extended Padberg-Gr{\"o}tschel exact separation algorithm}\label{alg:EPG}
    \SetKwInOut{Input}{input}\SetKwInOut{Output}{output}
    \Input{$\bar{G} $, $(y,x)$, $D$ and $(k_{in}, k_{out})$}
    \Output{A family $\mathcal{Q}$ of violated SECs}

    $\bar{V} \leftarrow$ sort $\bar{V}$ decreasingly by $y$\;
    $m \leftarrow y$\;
    Apply shrinking strategy ($\bar{G}, (y, x), m, H, \mathcal{Q}$)\;

    $(T,w, u) \leftarrow$\ref{alg:ghtree} $(\bar{G}, (y,x), v_1)$\;

    \For{a $\in A_T$}{
      $Q \leftarrow$ $d_a$\;
      \If{$ w_a - 2 \cdot u_a - 2\cdot v_a <2$}{
        \If{$|Q| > |V|/2$}{
          $Q \leftarrow V -Q$\;
        }
        $\mathcal{Q} \leftarrow \mathcal{Q} \cup \{\pi(Q)\}$\;
      }
    }
    $\mathcal{L} \leftarrow$~\ref{alg:cutgen} ($\bar{G}, (y, x), D, \mathcal{Q}, (k_{in}, k_{out})$)\;
  \end{algorithm}

  \begin{algorithm}[htb!]
    \SetAlgoRefName{CUTGEN}
    \caption{SEC generation}\label{alg:cutgen}
    \SetKwInOut{Input}{input}\SetKwInOut{Output}{output}
    \Input{$\bar{G}$, $(y,x)$, $D$, $\mathcal{Q}$, $(k_{in}, k_{out})$}
    \Output{A family $\mathcal{L}$ of violated SECs}

    \For{$Q \in \mathcal{Q}$}{
      \eIf{$D \cap Q = \emptyset$}
      {
        $M_{in} \leftarrow \{v \in Q: y_v \geq y_u \; \forall u \in Q\}$\;
        $S_{in} \leftarrow$ randomly select $k_{in}$ vertices from $M_{in}$\;
        }{
        $S_{in} \leftarrow$ a vertex in $D \cap Q$\;
      }

      \eIf{$D - Q = \emptyset$}
      {
        $M_{out} \leftarrow \{v \in \bar{V}-Q: y_v \geq y_u \; \forall u \in \bar{V}-Q\}$\;
        }{
        $S_{out} \leftarrow$ a vertex in $D - Q$\;
      }

      \For{$u \in S_{in}$}{

        \If{$D - Q = \emptyset$}
        {
          $S_{out} \leftarrow$ randomly select $k_{out}$ vertices from $M_{out}$\;
        }

        \For{$v \in S_{out}$}{
          Add the violated SEC $\langle Q, u, v\rangle$ to $\mathcal{L} $\;
        }
      }
    }
    $\mathcal{L} \leftarrow$~\ref{alg:cutgen} ($\bar{G}, (y, x), D, \mathcal{Q}, (k_{in}, k_{out})$)\;

  \end{algorithm}

  \clearpage
  \subsection{Directed Rooted Gomory-Hu Tree}\label{appendix:pseudo:gh}

  As was explained in Section~\ref{sec:sep-alg}, the key for an efficient extension of the Padberg-Gr{\"o}tschel exact separation algorithm is the construction of the directed rooted Gomory-Hu tree, which is detailed in the following pseudocodes. The novelty is the~\ref{alg:ghtree-addedge-reorder} procedure, where we show how the Gomory-Hu construction must be adapted to evaluate the $u_v$ values ($u_v=\arg\max\{y_u: u \in \Delta(v) \}$) and reorder the tree in order to maintain a given vertex in the top of the tree.

  \begin{algorithm}[htb!]
    \SetAlgoRefName{GHTREE}
    \caption{Rooted directed Gomory-Hu tree}\label{alg:ghtree}
    \SetKwInOut{Input}{input}\SetKwInOut{Output}{output}
    \Input{$\bar{G}$, $(y,x)$, $r$}
    \Output{$T,w, u$: a rooted directed weighted tree}

    $T \leftarrow (V, \emptyset)$\;

    \For{$v \in V$}
    {
      $u_v = m_v = \arg\max\{y_w : w \in \pi(v) \in \bar{G}\}$\;
    }

    $G^{*} \leftarrow \bar{G}$ and consider $|\pi(v)|=1$ for every $v \in V^{*}$\;

    $(T, w, u) \leftarrow$~\ref{alg:ghtree-recursive}$(G^{*}, (y,x), r, T, w, u)$\;

  \end{algorithm}

  \begin{algorithm}[htb!]
    \SetAlgoRefName{GHTREE-RECURSIVE}
    \caption{Recursive operator to build the Gomory-Hu tree}\label{alg:ghtree-recursive}
    \SetKwInOut{Input}{input}\SetKwInOut{Output}{output}
    \Input{$G^{*}$, $(y,x)$, $r$, $T$, $w, u$}
    \Output{$T, w, u$}

    $C \leftarrow \{v \in V^{*}: |\pi(v)|=1\}$\;
    \If{$|C|>1$}
    {
      $(a,b) \leftarrow$ randomly select two different vertices from $C$\;
      $(A:B) \leftarrow (a,b)$-minimum cut in $G^*$\;

      $(T, w, u, r_a, r_b) \leftarrow$~\ref{alg:ghtree-addedge-reorder}$(T, (y,x), m, u, r, A, B)$\;

      $(T, w, u) \leftarrow$~\ref{alg:ghtree-recursive}$(G^{*}[B], (y[B],x[B]), r_a,T,w, u)$\;

      $(T, w, u) \leftarrow$~\ref{alg:ghtree-recursive}$(G^{*}[A], (y[A],x[A]), r_b,T,w, u)$\;
    }

  \end{algorithm}

  \begin{algorithm}[htb!]
    \linespread{1.35}\selectfont
    \begin{scriptsize}
      \SetAlgoRefName{ADD-ARC/REORDER-TREE}
      \caption{Add arc and reorder the tree}\label{alg:ghtree-addedge-reorder}
      \SetKwInOut{Input}{input}\SetKwInOut{Output}{output}
      \Input{$T$, $(y,x)$, $m, u$, $r$, $A$, $B$}
      \Output{$T$, $w, u$, $r_a, r_b$}

      \eIf{$r \in A$}
      {
        $r_a \leftarrow r$\;
        $r_b \leftarrow b$\;

        \eIf{$parent(r) \in A $ or $parent(r) = \emptyset$}
        {
          $e = (r, b)$\;
        }
        {
          $e = (b, r)$\;
          $f = (p(r), r)$\;
          $g = (p(r), b)$\;
          $w_g \leftarrow w_f$\;
          $A_T = A_T - \{f\} \cup \{g\}$\;
          $m_r = \max \{m_r, m_b\}$\;
        }
        $u_r = m_r$\;
        $u_b = m_b$\;
        \For{$c \in child(r)$}
        {
          \eIf{$c \in A$}
          {
            $u_r = \max \{u_r, u_c\}$\;
          }
          {
            $A_T = A_T - \{(r, c)\} \cup \{(a, c)\}$\;
            $u_b = \max \{u_b, u_c\}$\;
          }
        }
      }
      {
        $r_a \leftarrow a$\;
        $r_b \leftarrow r$\;

        \eIf{$parent(r) \in B $ or $parent(r) = \emptyset$}
        {
          $e = (r, a)$\;
        }
        {
          $e = (a, r)$\;
          $f = (p(r), r)$\;
          $g = (p(r), a)$\;
          $w_g \leftarrow w_f$\;
          $A_T = A_T - \{f\} \cup \{g\}$\;
        }
        $u_r = m_r$\;
        $u_a = m_a$\;
        \For{$c \in child(r)$}
        {
          \eIf{$c \in B$}
          {
            $u_r = \max \{u_r, u_c\}$\;
          }
          {
            $A_T = A_T - \{(r, c)\} \cup \{(a, c)\}$\;
            $u_a = \max \{u_a, u_c\}$\;
          }
        }
      }

      $A_T = A_T \cup \{e\}$\;
      $w_e \leftarrow x(A:B)$\;

    \end{scriptsize}
  \end{algorithm}

  \clearpage
  \section{Detailed Computational Results}\label{appendix:detailed}
  \setcounter{table}{0}
  \renewcommand*{\thetable}{B.\arabic{table}}

  In this section, we show the computational results obtained in each considered SEC instance. For each instance, we present three tables: two are related with the shrinking processes and one is related with separation and SEC generation processes. In addition, the results are separated into three groups (Gen1, Gen2 and Gen3). These groups represent the generation strategy proposed in \citep{fischetti5} to build the OP vertex scores which are then used to obtain the support graphs.

  In tables Table~\ref{tab:pr76-1}, Table~\ref{tab:att532-1},
  $\ldots$
  and Table~\ref{tab:d15112-1}, we report the details of the shrinking preprocess. One can see, below the support graph and shrunk graph columns, the size of the given support graph and the size of the shrunk support graph for each shrinking strategy. In the preprocess columns, we show the number of $Q$ sets obtained and the time (in milliseconds) needed by each shrinking preprocess. As can be seen, the shrinking is very fast, needing very few dozens of millisecond to be accomplished in the larger instances. An interesting point of these tables is that within the shrinking preprocess we are already able to obtain $Q$ sets that correspond with violated SECs. In particular, the largest amount of $Q$ sets are obtained with the shrinking strategy S1S2.

  In tables Table~\ref{tab:pr76-2}, Table~\ref{tab:att532-2},
  $\ldots$
  and Table~\ref{tab:d15112-2}, we report the number of times a rule is applied by each shrinking strategy. Regarding the Conjecture 1 in the discussion of the computational experiments of the main paper, it can be seen that Rule C3 is rarely applied in the shrinking preprocess. Moreover, the strategy C1C2C3 does not provide further contractions of the support graph and, in all the compared instances, the obtained final shrunk graphs have the same amount of vertices and edges as with strategy C1C2.

  The extra column in these tables represents how many extra vertices are contracted in the internal shrinking process of Algorithm DHI, i.e, Extra is increased by one if rule C1, C2 or S1 is applied and by two if rule C3 is applied. The results show that this extra shrinking is rarely achieved.

  In tables Table~\ref{tab:pr76-3}, Table~\ref{tab:att532-3},
  $\ldots$
  and Table~\ref{tab:d15112-3} ,we report the details about the separation process and SEC generation. We can see that EPG approach always obtains more violated SECs than Algorithm EH as suggested theoretically in the main paper. Moreover, without using the shrinking preprocess, the EPG algorithm is always faster than Algorithm EH except for the smallest instance pr76.

  Regarding the SEC generation process, we compare two strategies $1 \times 1$ and  $10 \times 10$, which refer to the amount of vertices considered inside and outside $Q$ sets when generating the violated SECs. What we see is that, in medium-sized instances, the generation of violated SECs is the most time-consuming part (see the results regarding Algorithm EPG), but in large-sized, this difference is shortened. Nevertheless, it is likely that most of the generated violated cuts by $10 \times 10$ (around half a million of different violated SECs were obtained in large-sized instances by EPG) are useless and counterproductive to consider them, in practice, for a B\&C.

  \begin{landscape}
    \begin{center}\bf pr76\end{center}
    \begin{center} Shrinking: Preprocess and Extra\end{center}
    \vspace*{\fill}
    \begin{table}[htb!]
      \centering
      \scriptsize
\begin{tabular*}{\columnwidth}{@{\extracolsep{\fill}}lccccccccccccccccccc@{}}
\toprule
 &  & \multicolumn{6}{c}{Gen1} & \multicolumn{6}{c}{Gen2} & \multicolumn{6}{c}{Gen3} \\ \cmidrule(lr){3-8}\cmidrule(lr){9-14}\cmidrule(lr){15-20}
 &  & \multicolumn{2}{c}{Support graph} & \multicolumn{2}{c}{Shrunk graph} & \multicolumn{2}{c}{Preprocess} & \multicolumn{2}{c}{Support graph} & \multicolumn{2}{c}{Shrunk graph} & \multicolumn{2}{c}{Preprocess} & \multicolumn{2}{c}{Support graph} & \multicolumn{2}{c}{Shrunk graph} & \multicolumn{2}{c}{Preprocess} \\ \cmidrule(lr){3-4}\cmidrule(lr){5-6}\cmidrule(lr){7-8}\cmidrule(lr){9-10}\cmidrule(lr){11-12}\cmidrule(lr){13-14}\cmidrule(lr){15-16}\cmidrule(lr){17-18}\cmidrule(lr){19-20}
Shrinking  & $|V|$ & $|\bar{V}|$ & $|\bar{E}|$ & $|\bar{V}|$ & $|\bar{E}|$ & \#Q & Time & $|\bar{V}|$ & $|\bar{E}|$ & $|\bar{V}|$ & $|\bar{E}|$ & \#Q & Time & $|\bar{V}|$ & $|\bar{E}|$ & $|\bar{V}|$ & $|\bar{E}|$ & \#Q & \multicolumn{1}{c}{Time} \\ 
\midrule
NO  & $76$ & $65$ & $71$ & $65$ & $71$ & $0$ & $0.05$ & $50$ & $59$ & $50$ & $59$ & $0$ & $0.06$ & $54$ & $63$ & $54$ & $63$ & $0$ & $0.05$ \\
C1  & $76$ & $65$ & $71$ & $26$ & $32$ & $0$ & $0.10$ & $50$ & $59$ & $20$ & $29$ & $0$ & $0.08$ & $54$ & $63$ & $24$ & $33$ & $0$ & $0.09$ \\
C1C2  & $76$ & $65$ & $71$ & $26$ & $32$ & $0$ & $0.09$ & $50$ & $59$ & $20$ & $29$ & $0$ & $0.08$ & $54$ & $63$ & $24$ & $33$ & $0$ & $0.09$ \\
C1C2C3  & $76$ & $65$ & $71$ & $26$ & $32$ & $0$ & $0.10$ & $50$ & $59$ & $20$ & $29$ & $0$ & $0.09$ & $54$ & $63$ & $24$ & $33$ & $0$ & $0.10$ \\
S1  & $76$ & $65$ & $71$ & $14$ & $15$ & $4$ & $0.10$ & $50$ & $59$ & $\phantom{0}9$ & $13$ & $0$ & $0.08$ & $54$ & $63$ & $13$ & $18$ & $0$ & $0.09$ \\
S1S2  & $76$ & $65$ & $71$ & $11$ & $12$ & $4$ & $0.10$ & $50$ & $59$ & $\phantom{0}9$ & $13$ & $0$ & $0.08$ & $54$ & $63$ & $13$ & $18$ & $0$ & $0.09$ \\
\bottomrule 
\end{tabular*}
      \caption{Graph sizes, number of obtained $Q$ sets and running time of the preprocess by shrinking strategy and OP instance generation in pr76.}\label{tab:pr76-1}
    \end{table}
    \vspace*{\fill}
    \begin{table}[htb!]
      \centering
      \scriptsize

\begin{tabular*}{\columnwidth}{@{\extracolsep{\fill}}lccccccccccccccccccccc@{}}
\toprule
 & \multicolumn{7}{c}{Gen1} & \multicolumn{7}{c}{Gen2} & \multicolumn{7}{c}{Gen3} \\ \cmidrule(lr){2-8}\cmidrule(lr){9-15}\cmidrule(lr){16-22}
 & \multicolumn{6}{c}{Preprocess} & DHI & \multicolumn{6}{c}{Preprocess} & DHI & \multicolumn{6}{c}{Preprocess} & \multicolumn{1}{c}{DHI} \\ \cmidrule(lr){2-7}\cmidrule(lr){9-14}\cmidrule(lr){16-21}
Shrinking  & C1 & C2 & C3 & S1 & S2 & H & Extra & C1 & C2 & C3 & S1 & S2 & H & Extra & C1 & C2 & C3 & S1 & S2 & H & \multicolumn{1}{c}{Extra} \\ 
\midrule
NO  & $\phantom{0}0$ & $0$ & $0$ & $\phantom{0}0$ & $0$ & $\phantom{0}0$ & $0$ & $\phantom{0}0$ & $0$ & $0$ & $\phantom{0}0$ & $0$ & $\phantom{0}0$ & $0$ & $\phantom{0}0$ & $0$ & $0$ & $\phantom{0}0$ & $0$ & $\phantom{0}0$ & $0$ \\
C1  & $39$ & $0$ & $0$ & $\phantom{0}0$ & $0$ & $75$ & $0$ & $30$ & $0$ & $0$ & $\phantom{0}0$ & $0$ & $53$ & $0$ & $30$ & $0$ & $0$ & $\phantom{0}0$ & $0$ & $65$ & $0$ \\
C1C2  & $39$ & $0$ & $0$ & $\phantom{0}0$ & $0$ & $75$ & $0$ & $30$ & $0$ & $0$ & $\phantom{0}0$ & $0$ & $53$ & $0$ & $30$ & $0$ & $0$ & $\phantom{0}0$ & $0$ & $65$ & $0$ \\
C1C2C3  & $39$ & $0$ & $0$ & $\phantom{0}0$ & $0$ & $75$ & $0$ & $30$ & $0$ & $0$ & $\phantom{0}0$ & $0$ & $53$ & $0$ & $30$ & $0$ & $0$ & $\phantom{0}0$ & $0$ & $65$ & $0$ \\
S1  & $\phantom{0}0$ & $0$ & $0$ & $51$ & $0$ & $73$ & $0$ & $\phantom{0}0$ & $0$ & $0$ & $41$ & $0$ & $62$ & $0$ & $\phantom{0}0$ & $0$ & $0$ & $41$ & $0$ & $63$ & $0$ \\
S1S2  & $\phantom{0}0$ & $0$ & $0$ & $53$ & $1$ & $72$ & $0$ & $\phantom{0}0$ & $0$ & $0$ & $41$ & $0$ & $62$ & $0$ & $\phantom{0}0$ & $0$ & $0$ & $41$ & $0$ & $63$ & $0$ \\
\bottomrule 
\end{tabular*}
      \caption{Number of applications of each rule in the preprocess by shrinking strategy and OP instance generation in pr76. The extra column is particular of DHI separation strategy (during separation). }\label{tab:pr76-2}
    \end{table}
    \vspace*{\fill}
  \end{landscape}

  \begin{landscape}
    \begin{center}\bf att532\end{center}
    \begin{center} Shrinking: Preprocess and Extra\end{center}
    \vspace*{\fill}
    \begin{table}[htb!]
      \centering
      \scriptsize
\begin{tabular*}{\columnwidth}{@{\extracolsep{\fill}}lccccccccccccccccccc@{}}
\toprule
 &  & \multicolumn{6}{c}{Gen1} & \multicolumn{6}{c}{Gen2} & \multicolumn{6}{c}{Gen3} \\ \cmidrule(lr){3-8}\cmidrule(lr){9-14}\cmidrule(lr){15-20}
 &  & \multicolumn{2}{c}{Support graph} & \multicolumn{2}{c}{Shrunk graph} & \multicolumn{2}{c}{Preprocess} & \multicolumn{2}{c}{Support graph} & \multicolumn{2}{c}{Shrunk graph} & \multicolumn{2}{c}{Preprocess} & \multicolumn{2}{c}{Support graph} & \multicolumn{2}{c}{Shrunk graph} & \multicolumn{2}{c}{Preprocess} \\ \cmidrule(lr){3-4}\cmidrule(lr){5-6}\cmidrule(lr){7-8}\cmidrule(lr){9-10}\cmidrule(lr){11-12}\cmidrule(lr){13-14}\cmidrule(lr){15-16}\cmidrule(lr){17-18}\cmidrule(lr){19-20}
Shrinking  & $|V|$ & $|\bar{V}|$ & $|\bar{E}|$ & $|\bar{V}|$ & $|\bar{E}|$ & \#Q & Time & $|\bar{V}|$ & $|\bar{E}|$ & $|\bar{V}|$ & $|\bar{E}|$ & \#Q & Time & $|\bar{V}|$ & $|\bar{E}|$ & $|\bar{V}|$ & $|\bar{E}|$ & \#Q & \multicolumn{1}{c}{Time} \\ 
\midrule
NO  & $532$ & $458$ & $528$ & $458$ & $528$ & $\phantom{0}0$ & $0.33$ & $413$ & $503$ & $413$ & $503$ & $\phantom{0}0$ & $0.29$ & $412$ & $512$ & $412$ & $512$ & $\phantom{0}0$ & $0.32$ \\
C1  & $532$ & $458$ & $528$ & $166$ & $236$ & $\phantom{0}0$ & $0.62$ & $413$ & $503$ & $212$ & $302$ & $\phantom{0}0$ & $0.51$ & $412$ & $512$ & $240$ & $340$ & $\phantom{0}0$ & $0.50$ \\
C1C2  & $532$ & $458$ & $528$ & $142$ & $196$ & $\phantom{0}3$ & $0.67$ & $413$ & $503$ & $200$ & $279$ & $\phantom{0}4$ & $0.53$ & $412$ & $512$ & $221$ & $305$ & $\phantom{0}6$ & $0.54$ \\
C1C2C3  & $532$ & $458$ & $528$ & $142$ & $196$ & $\phantom{0}4$ & $0.70$ & $413$ & $503$ & $200$ & $279$ & $\phantom{0}4$ & $0.59$ & $412$ & $512$ & $221$ & $305$ & $\phantom{0}6$ & $0.58$ \\
S1  & $532$ & $458$ & $528$ & $\phantom{0}77$ & $111$ & $15$ & $0.74$ & $413$ & $503$ & $119$ & $164$ & $29$ & $0.63$ & $412$ & $512$ & $135$ & $185$ & $24$ & $0.59$ \\
\bottomrule 
\end{tabular*}
      \caption{Graph sizes, number of obtained $Q$ sets and running time of the preprocess by shrinking strategy and OP instance generation in att532. }\label{tab:att532-1}
    \end{table}
    \vspace*{\fill}

    \begin{table}[htb!]
      \centering
      \scriptsize
\begin{tabular*}{\columnwidth}{@{\extracolsep{\fill}}lccccccccccccccccccccc@{}}
\toprule
 & \multicolumn{7}{c}{Gen1} & \multicolumn{7}{c}{Gen2} & \multicolumn{7}{c}{Gen3} \\ \cmidrule(lr){2-8}\cmidrule(lr){9-15}\cmidrule(lr){16-22}
 & \multicolumn{6}{c}{Preprocess} & DHI & \multicolumn{6}{c}{Preprocess} & DHI & \multicolumn{6}{c}{Preprocess} & \multicolumn{1}{c}{DHI} \\ \cmidrule(lr){2-7}\cmidrule(lr){9-14}\cmidrule(lr){16-21}
Shrinking  & C1 & C2 & C3 & S1 & S2 & H & Extra & C1 & C2 & C3 & S1 & S2 & H & Extra & C1 & C2 & C3 & S1 & S2 & H & \multicolumn{1}{c}{Extra} \\ 
\midrule
NO  & $\phantom{00}0$ & $\phantom{0}0$ & $0$ & $\phantom{00}0$ & $0$ & $\phantom{00}0$ & $0.0$ & $\phantom{00}0$ & $\phantom{0}0$ & $0$ & $\phantom{00}0$ & $0$ & $\phantom{00}0$ & $0.0$ & $\phantom{00}0$ & $\phantom{0}0$ & $0$ & $\phantom{00}0$ & $0$ & $\phantom{00}0$ & $0.0$ \\
C1  & $292$ & $\phantom{0}0$ & $0$ & $\phantom{00}0$ & $0$ & $531$ & $0.0$ & $201$ & $\phantom{0}0$ & $0$ & $\phantom{00}0$ & $0$ & $462$ & $0.0$ & $172$ & $\phantom{0}0$ & $0$ & $\phantom{00}0$ & $0$ & $466$ & $0.0$ \\
C1C2  & $302$ & $14$ & $0$ & $\phantom{00}0$ & $0$ & $534$ & $0.0$ & $205$ & $\phantom{0}8$ & $0$ & $\phantom{00}0$ & $0$ & $469$ & $0.0$ & $178$ & $13$ & $0$ & $\phantom{00}0$ & $0$ & $470$ & $0.0$ \\
C1C2C3  & $302$ & $10$ & $2$ & $\phantom{00}0$ & $0$ & $528$ & $0.0$ & $205$ & $\phantom{0}8$ & $0$ & $\phantom{00}0$ & $0$ & $469$ & $0.0$ & $178$ & $13$ & $0$ & $\phantom{00}0$ & $0$ & $470$ & $0.0$ \\
S1  & $\phantom{00}0$ & $\phantom{0}0$ & $0$ & $381$ & $0$ & $513$ & $0.0$ & $\phantom{00}0$ & $\phantom{0}0$ & $0$ & $294$ & $0$ & $470$ & $0.0$ & $\phantom{00}0$ & $\phantom{0}0$ & $0$ & $277$ & $0$ & $467$ & $0.0$ \\
S1S2  & $\phantom{00}0$ & $\phantom{0}0$ & $0$ & $381$ & $4$ & $508$ & $0.5$ & $\phantom{00}0$ & $\phantom{0}0$ & $0$ & $296$ & $8$ & $469$ & $0.0$ & $\phantom{00}0$ & $\phantom{0}0$ & $0$ & $278$ & $5$ & $470$ & $0.2$ \\
\bottomrule 
\end{tabular*}
      \caption{Number of applications of each rule in the preprocess by shrinking strategy and OP instance generation in att532. The extra column is particular of DHI separation strategy (during separation). }\label{tab:att532-2}
    \end{table}
    \vspace*{\fill}
  \end{landscape}

  \begin{landscape}
    \begin{center}\bf vm1084\end{center}
    \begin{center} Shrinking: Preprocess and Extra\end{center}
    \vspace*{\fill}
    \begin{table}[htb!]
      \centering
      \scriptsize
\begin{tabular*}{\columnwidth}{@{\extracolsep{\fill}}lccccccccccccccccccc@{}}
\toprule
 &  & \multicolumn{6}{c}{Gen1} & \multicolumn{6}{c}{Gen2} & \multicolumn{6}{c}{Gen3} \\ \cmidrule(lr){3-8}\cmidrule(lr){9-14}\cmidrule(lr){15-20}
 &  & \multicolumn{2}{c}{Support graph} & \multicolumn{2}{c}{Shrunk graph} & \multicolumn{2}{c}{Preprocess} & \multicolumn{2}{c}{Support graph} & \multicolumn{2}{c}{Shrunk graph} & \multicolumn{2}{c}{Preprocess} & \multicolumn{2}{c}{Support graph} & \multicolumn{2}{c}{Shrunk graph} & \multicolumn{2}{c}{Preprocess} \\ \cmidrule(lr){3-4}\cmidrule(lr){5-6}\cmidrule(lr){7-8}\cmidrule(lr){9-10}\cmidrule(lr){11-12}\cmidrule(lr){13-14}\cmidrule(lr){15-16}\cmidrule(lr){17-18}\cmidrule(lr){19-20}
Shrinking  & $|V|$ & $|\bar{V}|$ & $|\bar{E}|$ & $|\bar{V}|$ & $|\bar{E}|$ & \#Q & Time & $|\bar{V}|$ & $|\bar{E}|$ & $|\bar{V}|$ & $|\bar{E}|$ & \#Q & Time & $|\bar{V}|$ & $|\bar{E}|$ & $|\bar{V}|$ & $|\bar{E}|$ & \#Q & \multicolumn{1}{c}{Time} \\ 
\midrule
NO  & $1084$ & $861$ & $980$ & $861$ & $980$ & $\phantom{0}0$ & $0.66$ & $863$ & $1012$ & $863$ & $1012$ & $\phantom{0}0$ & $0.68$ & $785$ & $917$ & $785$ & $917$ & $\phantom{0}0$ & $0.63$ \\
C1  & $1084$ & $861$ & $980$ & $297$ & $416$ & $\phantom{0}0$ & $1.16$ & $863$ & $1012$ & $377$ & $\phantom{0}526$ & $\phantom{0}0$ & $1.18$ & $785$ & $917$ & $297$ & $429$ & $\phantom{0}0$ & $1.08$ \\
C1C2  & $1084$ & $861$ & $980$ & $260$ & $354$ & $\phantom{0}8$ & $1.20$ & $863$ & $1012$ & $341$ & $\phantom{0}465$ & $\phantom{0}7$ & $1.17$ & $785$ & $917$ & $267$ & $378$ & $\phantom{0}7$ & $1.11$ \\
C1C2C3  & $1084$ & $861$ & $980$ & $260$ & $354$ & $\phantom{0}8$ & $1.30$ & $863$ & $1012$ & $341$ & $\phantom{0}465$ & $\phantom{0}6$ & $1.27$ & $785$ & $917$ & $267$ & $378$ & $\phantom{0}7$ & $1.22$ \\
S1  & $1084$ & $861$ & $980$ & $147$ & $202$ & $40$ & $1.30$ & $863$ & $1012$ & $213$ & $\phantom{0}289$ & $45$ & $1.30$ & $785$ & $917$ & $160$ & $233$ & $34$ & $1.14$ \\
\bottomrule 
\end{tabular*}
      \caption{Graph sizes, number of obtained $Q$ sets and running time of the preprocess by shrinking strategy and OP instance generation in vm1084. }\label{tab:vm1084-1}
    \end{table}
    \vspace*{\fill}

    \begin{table}[htb!]
      \centering
      \scriptsize
\begin{tabular*}{\columnwidth}{@{\extracolsep{\fill}}lccccccccccccccccccccc@{}}
\toprule
 & \multicolumn{7}{c}{Gen1} & \multicolumn{7}{c}{Gen2} & \multicolumn{7}{c}{Gen3} \\ \cmidrule(lr){2-8}\cmidrule(lr){9-15}\cmidrule(lr){16-22}
 & \multicolumn{6}{c}{Preprocess} & DHI & \multicolumn{6}{c}{Preprocess} & DHI & \multicolumn{6}{c}{Preprocess} & \multicolumn{1}{c}{DHI} \\ \cmidrule(lr){2-7}\cmidrule(lr){9-14}\cmidrule(lr){16-21}
Shrinking  & C1 & C2 & C3 & S1 & S2 & H & Extra & C1 & C2 & C3 & S1 & S2 & H & Extra & C1 & C2 & C3 & S1 & S2 & H & \multicolumn{1}{c}{Extra} \\ 
\midrule
NO  & $\phantom{00}0$ & $\phantom{0}0$ & $0$ & $\phantom{00}0$ & $\phantom{0}0$ & $\phantom{00}0$ & $0$ & $\phantom{00}0$ & $\phantom{0}0$ & $0$ & $\phantom{00}0$ & $\phantom{0}0$ & $\phantom{00}0$ & $0$ & $\phantom{00}0$ & $\phantom{0}0$ & $0$ & $\phantom{00}0$ & $\phantom{0}0$ & $\phantom{00}0$ & $0$ \\
C1  & $564$ & $\phantom{0}0$ & $0$ & $\phantom{00}0$ & $\phantom{0}0$ & $945$ & $0$ & $486$ & $\phantom{0}0$ & $0$ & $\phantom{00}0$ & $\phantom{0}0$ & $970$ & $0$ & $488$ & $\phantom{0}0$ & $0$ & $\phantom{00}0$ & $\phantom{0}0$ & $881$ & $0$ \\
C1C2  & $582$ & $19$ & $0$ & $\phantom{00}0$ & $\phantom{0}0$ & $980$ & $0$ & $502$ & $20$ & $0$ & $\phantom{00}0$ & $\phantom{0}0$ & $971$ & $0$ & $500$ & $18$ & $0$ & $\phantom{00}0$ & $\phantom{0}0$ & $887$ & $0$ \\
C1C2C3  & $582$ & $19$ & $0$ & $\phantom{00}0$ & $\phantom{0}0$ & $980$ & $0$ & $502$ & $18$ & $1$ & $\phantom{00}0$ & $\phantom{0}0$ & $968$ & $0$ & $500$ & $16$ & $1$ & $\phantom{00}0$ & $\phantom{0}0$ & $882$ & $0$ \\
S1  & $\phantom{00}0$ & $\phantom{0}0$ & $0$ & $714$ & $\phantom{0}0$ & $962$ & $0$ & $\phantom{00}0$ & $\phantom{0}0$ & $0$ & $650$ & $\phantom{0}0$ & $975$ & $0$ & $\phantom{00}0$ & $\phantom{0}0$ & $0$ & $625$ & $\phantom{0}0$ & $876$ & $0$ \\
\bottomrule 
\end{tabular*}
      \caption{Number of applications of each rule in the preprocess by shrinking strategy and OP instance generation in vm1084. The extra column is particular of DHI separation strategy (during separation). }\label{tab:vm1084-2}
    \end{table}
    \vspace*{\fill}
  \end{landscape}

  \begin{landscape}
    \begin{center}\bf rl1323\end{center}
    \begin{center} Shrinking: Preprocess and Extra\end{center}
    \vspace*{\fill}
    \begin{table}[htb!]
      \centering
      \scriptsize
\begin{tabular*}{\columnwidth}{@{\extracolsep{\fill}}lccccccccccccccccccc@{}}
\toprule
 &  & \multicolumn{6}{c}{Gen1} & \multicolumn{6}{c}{Gen2} & \multicolumn{6}{c}{Gen3} \\ \cmidrule(lr){3-8}\cmidrule(lr){9-14}\cmidrule(lr){15-20}
 &  & \multicolumn{2}{c}{Support graph} & \multicolumn{2}{c}{Shrunk graph} & \multicolumn{2}{c}{Preprocess} & \multicolumn{2}{c}{Support graph} & \multicolumn{2}{c}{Shrunk graph} & \multicolumn{2}{c}{Preprocess} & \multicolumn{2}{c}{Support graph} & \multicolumn{2}{c}{Shrunk graph} & \multicolumn{2}{c}{Preprocess} \\ \cmidrule(lr){3-4}\cmidrule(lr){5-6}\cmidrule(lr){7-8}\cmidrule(lr){9-10}\cmidrule(lr){11-12}\cmidrule(lr){13-14}\cmidrule(lr){15-16}\cmidrule(lr){17-18}\cmidrule(lr){19-20}
Shrinking  & $|V|$ & $|\bar{V}|$ & $|\bar{E}|$ & $|\bar{V}|$ & $|\bar{E}|$ & \#Q & Time & $|\bar{V}|$ & $|\bar{E}|$ & $|\bar{V}|$ & $|\bar{E}|$ & \#Q & Time & $|\bar{V}|$ & $|\bar{E}|$ & $|\bar{V}|$ & $|\bar{E}|$ & \#Q & \multicolumn{1}{c}{Time} \\ 
\midrule
NO  & $1323$ & $1011$ & $1165$ & $1011$ & $1165$ & $\phantom{0}0$ & $0.83$ & $933$ & $1073$ & $933$ & $1073$ & $\phantom{0}0$ & $0.76$ & $956$ & $1124$ & $956$ & $1124$ & $\phantom{0}0$ & $0.77$ \\
C1  & $1323$ & $1011$ & $1165$ & $\phantom{0}421$ & $\phantom{0}575$ & $\phantom{0}0$ & $1.39$ & $933$ & $1073$ & $375$ & $\phantom{0}515$ & $\phantom{0}0$ & $1.32$ & $956$ & $1124$ & $406$ & $\phantom{0}574$ & $\phantom{0}0$ & $1.32$ \\
C1C2  & $1323$ & $1011$ & $1165$ & $\phantom{0}401$ & $\phantom{0}538$ & $10$ & $1.44$ & $933$ & $1073$ & $335$ & $\phantom{0}445$ & $12$ & $1.36$ & $956$ & $1124$ & $382$ & $\phantom{0}529$ & $\phantom{0}9$ & $1.35$ \\
C1C2C3  & $1323$ & $1011$ & $1165$ & $\phantom{0}401$ & $\phantom{0}538$ & $10$ & $1.48$ & $933$ & $1073$ & $335$ & $\phantom{0}445$ & $12$ & $1.46$ & $956$ & $1124$ & $382$ & $\phantom{0}529$ & $\phantom{0}9$ & $1.50$ \\
\bottomrule 
\end{tabular*}
      \caption{Graph sizes, number of obtained $Q$ sets and running time of the preprocess by shrinking strategy and OP instance generation in rl1323. }\label{tab:rl1323-1}
    \end{table}
    \vspace*{\fill}

    \begin{table}[htb!]
      \centering
      \scriptsize
\begin{tabular*}{\columnwidth}{@{\extracolsep{\fill}}lccccccccccccccccccccc@{}}
\toprule
 & \multicolumn{7}{c}{Gen1} & \multicolumn{7}{c}{Gen2} & \multicolumn{7}{c}{Gen3} \\ \cmidrule(lr){2-8}\cmidrule(lr){9-15}\cmidrule(lr){16-22}
 & \multicolumn{6}{c}{Preprocess} & DHI & \multicolumn{6}{c}{Preprocess} & DHI & \multicolumn{6}{c}{Preprocess} & \multicolumn{1}{c}{DHI} \\ \cmidrule(lr){2-7}\cmidrule(lr){9-14}\cmidrule(lr){16-21}
Shrinking  & C1 & C2 & C3 & S1 & S2 & H & Extra & C1 & C2 & C3 & S1 & S2 & H & Extra & C1 & C2 & C3 & S1 & S2 & H & \multicolumn{1}{c}{Extra} \\ 
\midrule
NO  & $\phantom{00}0$ & $\phantom{0}0$ & $0$ & $\phantom{00}0$ & $\phantom{0}0$ & $\phantom{000}0$ & $0.0$ & $\phantom{00}0$ & $\phantom{0}0$ & $0$ & $\phantom{00}0$ & $\phantom{0}0$ & $\phantom{000}0$ & $0.0$ & $\phantom{00}0$ & $\phantom{0}0$ & $0$ & $\phantom{00}0$ & $\phantom{0}0$ & $\phantom{000}0$ & $0.0$ \\
C1  & $590$ & $\phantom{0}0$ & $0$ & $\phantom{00}0$ & $\phantom{0}0$ & $1149$ & $0.0$ & $558$ & $\phantom{0}0$ & $0$ & $\phantom{00}0$ & $\phantom{0}0$ & $1045$ & $0.0$ & $550$ & $\phantom{0}0$ & $0$ & $\phantom{00}0$ & $\phantom{0}0$ & $1093$ & $0.0$ \\
C1C2  & $598$ & $12$ & $0$ & $\phantom{00}0$ & $\phantom{0}0$ & $1141$ & $0.0$ & $579$ & $19$ & $0$ & $\phantom{00}0$ & $\phantom{0}0$ & $1048$ & $0.0$ & $559$ & $15$ & $0$ & $\phantom{00}0$ & $\phantom{0}0$ & $1091$ & $0.0$ \\
C1C2C3  & $598$ & $12$ & $0$ & $\phantom{00}0$ & $\phantom{0}0$ & $1141$ & $0.0$ & $578$ & $18$ & $1$ & $\phantom{00}0$ & $\phantom{0}0$ & $1044$ & $0.0$ & $559$ & $13$ & $1$ & $\phantom{00}0$ & $\phantom{0}0$ & $1086$ & $0.0$ \\
S1  & $\phantom{00}0$ & $\phantom{0}0$ & $0$ & $763$ & $\phantom{0}0$ & $1148$ & $0.0$ & $\phantom{00}0$ & $\phantom{0}0$ & $0$ & $724$ & $\phantom{0}0$ & $1055$ & $0.0$ & $\phantom{00}0$ & $\phantom{0}0$ & $0$ & $731$ & $\phantom{0}0$ & $1092$ & $0.0$ \\
\bottomrule 
\end{tabular*}
      \caption{Number of applications of each rule in the preprocess by shrinking strategy and OP instance generation in rl1323. The extra column is particular of DHI separation strategy (during separation). }\label{tab:rl1323-2}
    \end{table}
    \vspace*{\fill}
  \end{landscape}

  \begin{landscape}
    \begin{center}\bf vm1748\end{center}
    \begin{center} Shrinking: Preprocess and Extra\end{center}
    \vspace*{\fill}
    \begin{table}[htb!]
      \centering
      \scriptsize
\begin{tabular*}{\columnwidth}{@{\extracolsep{\fill}}lccccccccccccccccccc@{}}
\toprule
 &  & \multicolumn{6}{c}{Gen1} & \multicolumn{6}{c}{Gen2} & \multicolumn{6}{c}{Gen3} \\ \cmidrule(lr){3-8}\cmidrule(lr){9-14}\cmidrule(lr){15-20}
 &  & \multicolumn{2}{c}{Support graph} & \multicolumn{2}{c}{Shrunk graph} & \multicolumn{2}{c}{Preprocess} & \multicolumn{2}{c}{Support graph} & \multicolumn{2}{c}{Shrunk graph} & \multicolumn{2}{c}{Preprocess} & \multicolumn{2}{c}{Support graph} & \multicolumn{2}{c}{Shrunk graph} & \multicolumn{2}{c}{Preprocess} \\ \cmidrule(lr){3-4}\cmidrule(lr){5-6}\cmidrule(lr){7-8}\cmidrule(lr){9-10}\cmidrule(lr){11-12}\cmidrule(lr){13-14}\cmidrule(lr){15-16}\cmidrule(lr){17-18}\cmidrule(lr){19-20}
Shrinking  & $|V|$ & $|\bar{V}|$ & $|\bar{E}|$ & $|\bar{V}|$ & $|\bar{E}|$ & \#Q & Time & $|\bar{V}|$ & $|\bar{E}|$ & $|\bar{V}|$ & $|\bar{E}|$ & \#Q & Time & $|\bar{V}|$ & $|\bar{E}|$ & $|\bar{V}|$ & $|\bar{E}|$ & \#Q & \multicolumn{1}{c}{Time} \\ 
\midrule
NO  & $1748$ & $1490$ & $1756$ & $1490$ & $1756$ & $\phantom{00}0$ & $1.32$ & $1487$ & $1837$ & $1487$ & $1837$ & $\phantom{00}0$ & $1.36$ & $1361$ & $1586$ & $1361$ & $1586$ & $\phantom{00}0$ & $1.21$ \\
C1  & $1748$ & $1490$ & $1756$ & $\phantom{0}642$ & $\phantom{0}908$ & $\phantom{00}0$ & $2.39$ & $1487$ & $1837$ & $\phantom{0}808$ & $1158$ & $\phantom{00}0$ & $2.33$ & $1361$ & $1586$ & $\phantom{0}515$ & $\phantom{0}740$ & $\phantom{00}0$ & $2.18$ \\
C1C2  & $1748$ & $1490$ & $1756$ & $\phantom{0}596$ & $\phantom{0}823$ & $\phantom{0}18$ & $2.49$ & $1487$ & $1837$ & $\phantom{0}727$ & $1005$ & $\phantom{0}32$ & $2.47$ & $1361$ & $1586$ & $\phantom{0}480$ & $\phantom{0}680$ & $\phantom{00}6$ & $2.28$ \\
C1C2C3  & $1748$ & $1490$ & $1756$ & $\phantom{0}596$ & $\phantom{0}823$ & $\phantom{0}18$ & $2.72$ & $1487$ & $1837$ & $\phantom{0}727$ & $1005$ & $\phantom{0}32$ & $2.64$ & $1361$ & $1586$ & $\phantom{0}480$ & $\phantom{0}680$ & $\phantom{00}6$ & $2.42$ \\
\bottomrule 
\end{tabular*}
      \caption{Graph sizes, number of obtained $Q$ sets and running time of the preprocess by shrinking strategy and OP instance generation in vm1748. }\label{tab:vm1748-1}
    \end{table}
    \vspace*{\fill}

    \begin{table}[htb!]
      \centering
      \scriptsize
\begin{tabular*}{\columnwidth}{@{\extracolsep{\fill}}lccccccccccccccccccccc@{}}
\toprule
 & \multicolumn{7}{c}{Gen1} & \multicolumn{7}{c}{Gen2} & \multicolumn{7}{c}{Gen3} \\ \cmidrule(lr){2-8}\cmidrule(lr){9-15}\cmidrule(lr){16-22}
 & \multicolumn{6}{c}{Preprocess} & DHI & \multicolumn{6}{c}{Preprocess} & DHI & \multicolumn{6}{c}{Preprocess} & \multicolumn{1}{c}{DHI} \\ \cmidrule(lr){2-7}\cmidrule(lr){9-14}\cmidrule(lr){16-21}
Shrinking  & C1 & C2 & C3 & S1 & S2 & H & Extra & C1 & C2 & C3 & S1 & S2 & H & Extra & C1 & C2 & C3 & S1 & S2 & H & \multicolumn{1}{c}{Extra} \\ 
\midrule
NO  & $\phantom{00}0$ & $\phantom{0}0$ & $0$ & $\phantom{000}0$ & $\phantom{0}0$ & $\phantom{000}0$ & $0.0$ & $\phantom{00}0$ & $\phantom{0}0$ & $0$ & $\phantom{000}0$ & $\phantom{0}0$ & $\phantom{000}0$ & $0.0$ & $\phantom{00}0$ & $\phantom{0}0$ & $0$ & $\phantom{000}0$ & $\phantom{0}0$ & $\phantom{000}0$ & $0.0$ \\
C1  & $848$ & $\phantom{0}0$ & $0$ & $\phantom{000}0$ & $\phantom{0}0$ & $1653$ & $0.0$ & $679$ & $\phantom{0}0$ & $0$ & $\phantom{000}0$ & $\phantom{0}0$ & $1659$ & $0.0$ & $846$ & $\phantom{0}0$ & $0$ & $\phantom{000}0$ & $\phantom{0}0$ & $1518$ & $0.0$ \\
C1C2  & $866$ & $28$ & $0$ & $\phantom{000}0$ & $\phantom{0}0$ & $1676$ & $0.0$ & $711$ & $49$ & $0$ & $\phantom{000}0$ & $\phantom{0}0$ & $1690$ & $0.0$ & $859$ & $22$ & $0$ & $\phantom{000}0$ & $\phantom{0}0$ & $1533$ & $0.0$ \\
C1C2C3  & $866$ & $28$ & $0$ & $\phantom{000}0$ & $\phantom{0}0$ & $1676$ & $0.0$ & $711$ & $43$ & $3$ & $\phantom{000}0$ & $\phantom{0}0$ & $1680$ & $0.0$ & $859$ & $20$ & $1$ & $\phantom{000}0$ & $\phantom{0}0$ & $1529$ & $0.0$ \\
S1  & $\phantom{00}0$ & $\phantom{0}0$ & $0$ & $1116$ & $\phantom{0}0$ & $1692$ & $0.0$ & $\phantom{00}0$ & $\phantom{0}0$ & $0$ & $1000$ & $\phantom{0}0$ & $1729$ & $0.0$ & $\phantom{00}0$ & $\phantom{0}0$ & $0$ & $1077$ & $\phantom{0}0$ & $1525$ & $0.0$ \\
\bottomrule 
\end{tabular*}
      \caption{Number of applications of each rule in the preprocess by shrinking strategy and OP instance generation in vm1748. The extra column is particular of DHI separation strategy (during separation). }\label{tab:vm1748-2}
    \end{table}
    \vspace*{\fill}
  \end{landscape}

  \begin{landscape}
    \begin{center}\bf rl5934\end{center}
    \begin{center} Shrinking: Preprocess and Extra\end{center}
    \vspace*{\fill}
    \begin{table}[htb!]
      \centering
      \scriptsize
\begin{tabular*}{\columnwidth}{@{\extracolsep{\fill}}lccccccccccccccccccc@{}}
\toprule
 &  & \multicolumn{6}{c}{Gen1} & \multicolumn{6}{c}{Gen2} & \multicolumn{6}{c}{Gen3} \\ \cmidrule(lr){3-8}\cmidrule(lr){9-14}\cmidrule(lr){15-20}
 &  & \multicolumn{2}{c}{Support graph} & \multicolumn{2}{c}{Shrunk graph} & \multicolumn{2}{c}{Preprocess} & \multicolumn{2}{c}{Support graph} & \multicolumn{2}{c}{Shrunk graph} & \multicolumn{2}{c}{Preprocess} & \multicolumn{2}{c}{Support graph} & \multicolumn{2}{c}{Shrunk graph} & \multicolumn{2}{c}{Preprocess} \\ \cmidrule(lr){3-4}\cmidrule(lr){5-6}\cmidrule(lr){7-8}\cmidrule(lr){9-10}\cmidrule(lr){11-12}\cmidrule(lr){13-14}\cmidrule(lr){15-16}\cmidrule(lr){17-18}\cmidrule(lr){19-20}
Shrinking  & $|V|$ & $|\bar{V}|$ & $|\bar{E}|$ & $|\bar{V}|$ & $|\bar{E}|$ & \#Q & Time & $|\bar{V}|$ & $|\bar{E}|$ & $|\bar{V}|$ & $|\bar{E}|$ & \#Q & Time & $|\bar{V}|$ & $|\bar{E}|$ & $|\bar{V}|$ & $|\bar{E}|$ & \#Q & \multicolumn{1}{c}{Time} \\ 
\midrule
NO  & $5934$ & $4303$ & $4871$ & $4303$ & $4871$ & $\phantom{00}0$ & $\phantom{0}5.67$ & $4101$ & $4651$ & $4101$ & $4651$ & $\phantom{00}0$ & $\phantom{0}5.31$ & $3970$ & $4424$ & $3970$ & $4424$ & $\phantom{00}0$ & $\phantom{0}5.43$ \\
C1  & $5934$ & $4303$ & $4871$ & $1477$ & $2045$ & $\phantom{00}0$ & $11.10$ & $4101$ & $4651$ & $1533$ & $2083$ & $\phantom{00}0$ & $10.06$ & $3970$ & $4424$ & $1127$ & $1581$ & $\phantom{00}0$ & $10.44$ \\
C1C2  & $5934$ & $4303$ & $4871$ & $1312$ & $1754$ & $\phantom{0}54$ & $11.50$ & $4101$ & $4651$ & $1415$ & $1873$ & $\phantom{0}44$ & $10.45$ & $3970$ & $4424$ & $1014$ & $1381$ & $\phantom{0}42$ & $10.60$ \\
C1C2C3  & $5934$ & $4303$ & $4871$ & $1312$ & $1754$ & $\phantom{0}54$ & $11.61$ & $4101$ & $4651$ & $1415$ & $1873$ & $\phantom{0}43$ & $10.71$ & $3970$ & $4424$ & $1014$ & $1381$ & $\phantom{0}41$ & $10.57$ \\
\bottomrule 
\end{tabular*}
      \caption{Graph sizes, number of obtained $Q$ sets and running time of the preprocess by shrinking strategy and OP instance generation in rl5934. }\label{tab:rl5934-1}
    \end{table}
    \vspace*{\fill}

    \begin{table}[htb!]
      \centering
      \scriptsize
\begin{tabular*}{\columnwidth}{@{\extracolsep{\fill}}lccccccccccccccccccccc@{}}
\toprule
 & \multicolumn{7}{c}{Gen1} & \multicolumn{7}{c}{Gen2} & \multicolumn{7}{c}{Gen3} \\ \cmidrule(lr){2-8}\cmidrule(lr){9-15}\cmidrule(lr){16-22}
 & \multicolumn{6}{c}{Preprocess} & DHI & \multicolumn{6}{c}{Preprocess} & DHI & \multicolumn{6}{c}{Preprocess} & \multicolumn{1}{c}{DHI} \\ \cmidrule(lr){2-7}\cmidrule(lr){9-14}\cmidrule(lr){16-21}
Shrinking  & C1 & C2 & C3 & S1 & S2 & H & Extra & C1 & C2 & C3 & S1 & S2 & H & Extra & C1 & C2 & C3 & S1 & S2 & H & \multicolumn{1}{c}{Extra} \\ 
\midrule
NO  & $\phantom{000}0$ & $\phantom{0}0$ & $0$ & $\phantom{000}0$ & $\phantom{0}0$ & $\phantom{000}0$ & $0$ & $\phantom{000}0$ & $\phantom{0}0$ & $0$ & $\phantom{000}0$ & $\phantom{0}0$ & $\phantom{000}0$ & $0$ & $\phantom{000}0$ & $\phantom{0}0$ & $0$ & $\phantom{000}0$ & $\phantom{0}0$ & $\phantom{000}0$ & $0$ \\
C1  & $2826$ & $\phantom{0}0$ & $0$ & $\phantom{000}0$ & $\phantom{0}0$ & $4864$ & $0$ & $2568$ & $\phantom{0}0$ & $0$ & $\phantom{000}0$ & $\phantom{0}0$ & $4598$ & $0$ & $2843$ & $\phantom{0}0$ & $0$ & $\phantom{000}0$ & $\phantom{0}0$ & $4389$ & $0$ \\
C1C2  & $2904$ & $87$ & $0$ & $\phantom{000}0$ & $\phantom{0}0$ & $4854$ & $0$ & $2622$ & $64$ & $0$ & $\phantom{000}0$ & $\phantom{0}0$ & $4603$ & $0$ & $2897$ & $59$ & $0$ & $\phantom{000}0$ & $\phantom{0}0$ & $4374$ & $0$ \\
C1C2C3  & $2904$ & $85$ & $1$ & $\phantom{000}0$ & $\phantom{0}0$ & $4852$ & $0$ & $2620$ & $58$ & $4$ & $\phantom{000}0$ & $\phantom{0}0$ & $4585$ & $0$ & $2896$ & $54$ & $3$ & $\phantom{000}0$ & $\phantom{0}0$ & $4363$ & $0$ \\
S1  & $\phantom{000}0$ & $\phantom{0}0$ & $0$ & $3503$ & $\phantom{0}0$ & $4790$ & $0$ & $\phantom{000}0$ & $\phantom{0}0$ & $0$ & $3224$ & $\phantom{0}0$ & $4563$ & $0$ & $\phantom{000}0$ & $\phantom{0}0$ & $0$ & $3346$ & $\phantom{0}0$ & $4361$ & $0$ \\
\bottomrule 
\end{tabular*}
      \caption{Number of applications of each rule in the preprocess by shrinking strategy and OP instance generation in rl5934. The extra column is particular of DHI separation strategy (during separation). }\label{tab:rl5934-2}
    \end{table}
    \vspace*{\fill}
  \end{landscape}

  \begin{landscape}
    \begin{center}\bf usa13509\end{center}
    \begin{center} Shrinking: Preprocess and Extra\end{center}
    \vspace*{\fill}
    \begin{table}[htb!]
      \centering
      \scriptsize
\begin{tabular*}{\columnwidth}{@{\extracolsep{\fill}}lccccccccccccccccccc@{}}
\toprule
 &  & \multicolumn{6}{c}{Gen1} & \multicolumn{6}{c}{Gen2} & \multicolumn{6}{c}{Gen3} \\ \cmidrule(lr){3-8}\cmidrule(lr){9-14}\cmidrule(lr){15-20}
 &  & \multicolumn{2}{c}{Support graph} & \multicolumn{2}{c}{Shrunk graph} & \multicolumn{2}{c}{Preprocess} & \multicolumn{2}{c}{Support graph} & \multicolumn{2}{c}{Shrunk graph} & \multicolumn{2}{c}{Preprocess} & \multicolumn{2}{c}{Support graph} & \multicolumn{2}{c}{Shrunk graph} & \multicolumn{2}{c}{Preprocess} \\ \cmidrule(lr){3-4}\cmidrule(lr){5-6}\cmidrule(lr){7-8}\cmidrule(lr){9-10}\cmidrule(lr){11-12}\cmidrule(lr){13-14}\cmidrule(lr){15-16}\cmidrule(lr){17-18}\cmidrule(lr){19-20}
Shrinking  & $|V|$ & $|\bar{V}|$ & $|\bar{E}|$ & $|\bar{V}|$ & $|\bar{E}|$ & \#Q & Time & $|\bar{V}|$ & $|\bar{E}|$ & $|\bar{V}|$ & $|\bar{E}|$ & \#Q & Time & $|\bar{V}|$ & $|\bar{E}|$ & $|\bar{V}|$ & $|\bar{E}|$ & \#Q & \multicolumn{1}{c}{Time} \\ 
\midrule
NO  & $13509$ & $9084$ & $9990$ & $9084$ & $9990$ & $\phantom{00}0$ & $18.58$ & $8015$ & $8735$ & $8015$ & $8735$ & $\phantom{00}0$ & $14.58$ & $7245$ & $7992$ & $7245$ & $7992$ & $\phantom{00}0$ & $12.73$ \\
C1  & $13509$ & $9084$ & $9990$ & $2017$ & $2923$ & $\phantom{00}0$ & $44.98$ & $8015$ & $8735$ & $1723$ & $2443$ & $\phantom{00}0$ & $35.07$ & $7245$ & $7992$ & $1644$ & $2391$ & $\phantom{00}0$ & $28.85$ \\
C1C2  & $13509$ & $9084$ & $9990$ & $1891$ & $2725$ & $\phantom{00}1$ & $44.86$ & $8015$ & $8735$ & $1623$ & $2284$ & $\phantom{00}1$ & $35.70$ & $7245$ & $7992$ & $1490$ & $2148$ & $\phantom{00}0$ & $29.28$ \\
C1C2C3  & $13509$ & $9084$ & $9990$ & $1891$ & $2725$ & $\phantom{00}1$ & $45.12$ & $8015$ & $8735$ & $1623$ & $2284$ & $\phantom{00}1$ & $35.42$ & $7245$ & $7992$ & $1490$ & $2148$ & $\phantom{00}0$ & $30.03$ \\
\bottomrule 
\end{tabular*}
      \caption{Graph sizes, number of obtained $Q$ sets and running time of the preprocess by shrinking strategy and OP instance generation in usa13509. }\label{tab:usa1359-1}
    \end{table}
    \vspace*{\fill}

    \begin{table}[htb!]
      \centering
      \scriptsize
\begin{tabular*}{\columnwidth}{@{\extracolsep{\fill}}lccccccccccccccccccccc@{}}
\toprule
 & \multicolumn{7}{c}{Gen1} & \multicolumn{7}{c}{Gen2} & \multicolumn{7}{c}{Gen3} \\ \cmidrule(lr){2-8}\cmidrule(lr){9-15}\cmidrule(lr){16-22}
 & \multicolumn{6}{c}{Preprocess} & DHI & \multicolumn{6}{c}{Preprocess} & DHI & \multicolumn{6}{c}{Preprocess} & \multicolumn{1}{c}{DHI} \\ \cmidrule(lr){2-7}\cmidrule(lr){9-14}\cmidrule(lr){16-21}
Shrinking  & C1 & C2 & C3 & S1 & S2 & H & Extra & C1 & C2 & C3 & S1 & S2 & H & Extra & C1 & C2 & C3 & S1 & S2 & H & \multicolumn{1}{c}{Extra} \\ 
\midrule
NO  & $\phantom{000}0$ & $\phantom{0}0$ & $0$ & $\phantom{000}0$ & $\phantom{00}0$ & $\phantom{0000}0$ & $0$ & $\phantom{000}0$ & $\phantom{0}0$ & $0$ & $\phantom{000}0$ & $\phantom{00}0$ & $\phantom{0000}0$ & $0$ & $\phantom{000}0$ & $\phantom{0}0$ & $0$ & $\phantom{000}0$ & $\phantom{00}0$ & $\phantom{0000}0$ & $0$ \\
C1  & $7067$ & $\phantom{0}0$ & $0$ & $\phantom{000}0$ & $\phantom{00}0$ & $10005$ & $0$ & $6292$ & $\phantom{0}0$ & $0$ & $\phantom{000}0$ & $\phantom{00}0$ & $\phantom{0}8751$ & $0$ & $5601$ & $\phantom{0}0$ & $0$ & $\phantom{000}0$ & $\phantom{00}0$ & $\phantom{0}7971$ & $0$ \\
C1C2  & $7123$ & $70$ & $0$ & $\phantom{000}0$ & $\phantom{00}0$ & $\phantom{0}9969$ & $0$ & $6334$ & $58$ & $0$ & $\phantom{000}0$ & $\phantom{00}0$ & $\phantom{0}8722$ & $0$ & $5666$ & $89$ & $0$ & $\phantom{000}0$ & $\phantom{00}0$ & $\phantom{0}7936$ & $0$ \\
C1C2C3  & $7123$ & $64$ & $3$ & $\phantom{000}0$ & $\phantom{00}0$ & $\phantom{0}9960$ & $0$ & $6334$ & $56$ & $1$ & $\phantom{000}0$ & $\phantom{00}0$ & $\phantom{0}8718$ & $0$ & $5666$ & $83$ & $3$ & $\phantom{000}0$ & $\phantom{00}0$ & $\phantom{0}7926$ & $0$ \\
S1  & $\phantom{000}0$ & $\phantom{0}0$ & $0$ & $8202$ & $\phantom{00}0$ & $\phantom{0}9753$ & $0$ & $\phantom{000}0$ & $\phantom{0}0$ & $0$ & $7169$ & $\phantom{00}0$ & $\phantom{0}8628$ & $0$ & $\phantom{000}0$ & $\phantom{0}0$ & $0$ & $6527$ & $\phantom{00}0$ & $\phantom{0}7771$ & $0$ \\
\bottomrule 
\end{tabular*}
      \caption{Number of applications of each rule in the preprocess by shrinking strategy and OP instance generation in usa13509. The extra column is particular of DHI separation strategy (during separation). }\label{tab:usa1359-2}
    \end{table}
    \vspace*{\fill}
  \end{landscape}

  \begin{landscape}
    \begin{center}\bf d15112\end{center}
    \begin{center} Shrinking: Preprocess and Extra\end{center}
    \vspace*{\fill}
    \begin{table}[htb!]
      \centering
      \scriptsize
\begin{tabular*}{\columnwidth}{@{\extracolsep{\fill}}lccccccccccccccccccc@{}}
\toprule
 &  & \multicolumn{6}{c}{Gen1} & \multicolumn{6}{c}{Gen2} & \multicolumn{6}{c}{Gen3} \\ \cmidrule(lr){3-8}\cmidrule(lr){9-14}\cmidrule(lr){15-20}
 &  & \multicolumn{2}{c}{Support graph} & \multicolumn{2}{c}{Shrunk graph} & \multicolumn{2}{c}{Preprocess} & \multicolumn{2}{c}{Support graph} & \multicolumn{2}{c}{Shrunk graph} & \multicolumn{2}{c}{Preprocess} & \multicolumn{2}{c}{Support graph} & \multicolumn{2}{c}{Shrunk graph} & \multicolumn{2}{c}{Preprocess} \\ \cmidrule(lr){3-4}\cmidrule(lr){5-6}\cmidrule(lr){7-8}\cmidrule(lr){9-10}\cmidrule(lr){11-12}\cmidrule(lr){13-14}\cmidrule(lr){15-16}\cmidrule(lr){17-18}\cmidrule(lr){19-20}
Shrinking  & $|V|$ & $|\bar{V}|$ & $|\bar{E}|$ & $|\bar{V}|$ & $|\bar{E}|$ & \#Q & Time & $|\bar{V}|$ & $|\bar{E}|$ & $|\bar{V}|$ & $|\bar{E}|$ & \#Q & Time & $|\bar{V}|$ & $|\bar{E}|$ & $|\bar{V}|$ & $|\bar{E}|$ & \#Q & \multicolumn{1}{c}{Time} \\ 
\midrule
NO  & $15112$ & $9075$ & $9866$ & $9075$ & $9866$ & $\phantom{00}0$ & $18.73$ & $7682$ & $8322$ & $7682$ & $8322$ & $\phantom{00}0$ & $14.25$ & $9393$ & $10378$ & $9393$ & $10378$ & $\phantom{00}0$ & $19.63$ \\
C1  & $15112$ & $9075$ & $9866$ & $1793$ & $2584$ & $\phantom{00}0$ & $43.70$ & $7682$ & $8322$ & $1600$ & $2240$ & $\phantom{00}0$ & $32.45$ & $9393$ & $10378$ & $2170$ & $\phantom{0}3155$ & $\phantom{00}0$ & $46.91$ \\
C1C2  & $15112$ & $9075$ & $9866$ & $1656$ & $2373$ & $\phantom{00}0$ & $44.42$ & $7682$ & $8322$ & $1510$ & $2097$ & $\phantom{00}1$ & $32.28$ & $9393$ & $10378$ & $1994$ & $\phantom{0}2876$ & $\phantom{00}5$ & $47.70$ \\
C1C2C3  & $15112$ & $9075$ & $9866$ & $1656$ & $2373$ & $\phantom{00}0$ & $45.28$ & $7682$ & $8322$ & $1510$ & $2097$ & $\phantom{00}1$ & $32.70$ & $9393$ & $10378$ & $1994$ & $\phantom{0}2876$ & $\phantom{00}5$ & $48.93$ \\
\bottomrule 
\end{tabular*}
      \caption{Graph sizes, number of obtained $Q$ sets and running time of the preprocess by shrinking strategy and OP instance generation in d15112. }\label{tab:d15112-1}
    \end{table}
    \vspace*{\fill}

    \begin{table}[htb!]
      \centering
      \scriptsize
\begin{tabular*}{\columnwidth}{@{\extracolsep{\fill}}lccccccccccccccccccccc@{}}
\toprule
 & \multicolumn{7}{c}{Gen1} & \multicolumn{7}{c}{Gen2} & \multicolumn{7}{c}{Gen3} \\ \cmidrule(lr){2-8}\cmidrule(lr){9-15}\cmidrule(lr){16-22}
 & \multicolumn{6}{c}{Preprocess} & DHI & \multicolumn{6}{c}{Preprocess} & DHI & \multicolumn{6}{c}{Preprocess} & \multicolumn{1}{c}{DHI} \\ \cmidrule(lr){2-7}\cmidrule(lr){9-14}\cmidrule(lr){16-21}
Shrinking  & C1 & C2 & C3 & S1 & S2 & H & Extra & C1 & C2 & C3 & S1 & S2 & H & Extra & C1 & C2 & C3 & S1 & S2 & H & \multicolumn{1}{c}{Extra} \\ 
\midrule
NO  & $\phantom{000}0$ & $\phantom{00}0$ & $0$ & $\phantom{000}0$ & $\phantom{00}0$ & $\phantom{0000}0$ & $0$ & $\phantom{000}0$ & $\phantom{00}0$ & $0$ & $\phantom{000}0$ & $\phantom{00}0$ & $\phantom{0000}0$ & $0$ & $\phantom{000}0$ & $\phantom{00}0$ & $0$ & $\phantom{000}0$ & $\phantom{00}0$ & $\phantom{0000}0$ & $0$ \\
C1  & $7282$ & $\phantom{00}0$ & $0$ & $\phantom{000}0$ & $\phantom{00}0$ & $\phantom{0}9933$ & $0$ & $6082$ & $\phantom{00}0$ & $0$ & $\phantom{000}0$ & $\phantom{00}0$ & $\phantom{0}8382$ & $0$ & $7223$ & $\phantom{00}0$ & $0$ & $\phantom{000}0$ & $\phantom{00}0$ & $10350$ & $0$ \\
C1C2  & $7345$ & $\phantom{0}74$ & $0$ & $\phantom{000}0$ & $\phantom{00}0$ & $\phantom{0}9864$ & $0$ & $6120$ & $\phantom{0}52$ & $0$ & $\phantom{000}0$ & $\phantom{00}0$ & $\phantom{0}8357$ & $0$ & $7299$ & $100$ & $0$ & $\phantom{000}0$ & $\phantom{00}0$ & $10305$ & $0$ \\
C1C2C3  & $7345$ & $\phantom{0}70$ & $2$ & $\phantom{000}0$ & $\phantom{00}0$ & $\phantom{0}9858$ & $0$ & $6120$ & $\phantom{0}52$ & $0$ & $\phantom{000}0$ & $\phantom{00}0$ & $\phantom{0}8357$ & $0$ & $7299$ & $\phantom{0}98$ & $1$ & $\phantom{000}0$ & $\phantom{00}0$ & $10302$ & $0$ \\
\bottomrule 
\end{tabular*}
      \caption{Number of applications of each rule in the preprocess by shrinking strategy and OP instance generation in d15112. The extra column is particular of DHI separation strategy (during separation). }\label{tab:d15112-2}
    \end{table}
    \vspace*{\fill}
  \end{landscape}

  \begin{landscape}
    \begin{center}\bf pr76\end{center}
    \begin{center} Separation and SEC Generation\end{center}
    \vspace*{\fill}
    \begin{table}[htb!]
      \centering
      \tiny
\begin{tabular*}{\columnwidth}{@{\extracolsep{\fill}}llcccccccccccccccccc@{}}
\toprule
& & \multicolumn{6}{c}{Gen1} & \multicolumn{6}{c}{Gen2} & \multicolumn{6}{c}{Gen3} \\ \cmidrule(lr){3-8}\cmidrule(lr){9-14}\cmidrule(lr){15-20}
& & \multicolumn{2}{c}{Separation} & \multicolumn{4}{c}{SEC Generation} & \multicolumn{2}{c}{Separation} & \multicolumn{4}{c}{SEC Generation} & \multicolumn{2}{c}{Separation} & \multicolumn{4}{c}{SEC Generation} \\ \cmidrule(lr){3-4}\cmidrule(lr){5-8}\cmidrule(lr){9-10}\cmidrule(lr){11-14}\cmidrule(lr){15-16}\cmidrule(lr){17-20}
& & \multicolumn{2}{c}{(20 runs)} & \multicolumn{2}{c}{1x1 (10 runs)} & \multicolumn{2}{c}{10x10 (10 runs)} & \multicolumn{2}{c}{(20 runs)} & \multicolumn{2}{c}{1x1 (10 runs)} & \multicolumn{2}{c}{10x10 (10 runs)} & \multicolumn{2}{c}{(20 runs)} & \multicolumn{2}{c}{1x1 (10 runs)} & \multicolumn{2}{c}{10x10 (10 runs)} \\ \cmidrule(lr){3-4}\cmidrule(lr){5-6}\cmidrule(lr){7-8}\cmidrule(lr){9-10}\cmidrule(lr){11-12}\cmidrule(lr){13-14}\cmidrule(lr){15-16}\cmidrule(lr){17-18}\cmidrule(lr){19-20}
Sep. & Shrinking & \#Q & Time & \#SEC & Time & \#SEC & Time & \#Q & Time & \#SEC & Time & \#SEC & Time & \#Q & Time & \#SEC & Time & \#SEC & \multicolumn{1}{c}{Time} \\ 
\midrule
EH & NO  & $46$ & $0.4$ & $\phantom{00}46$ & $0.5$ & $3030$ & $1.5$ & $16$ & $0.6$ & $\phantom{00}16$ & $0.7$ & $\phantom{0}860$ & $1.1$ & $37$ & $0.3$ & $\phantom{00}37$ & $0.4$ & $2750$ & $1.4$ \\
 & C1  & $11$ & $0.2$ & $\phantom{00}11$ & $0.2$ & $\phantom{0}610$ & $0.5$ & $10$ & $0.2$ & $\phantom{00}10$ & $0.2$ & $\phantom{0}390$ & $0.5$ & $10$ & $0.2$ & $\phantom{00}10$ & $0.2$ & $\phantom{0}620$ & $0.4$ \\
 & C1C2  & $11$ & $0.2$ & $\phantom{00}11$ & $0.2$ & $\phantom{0}610$ & $0.5$ & $10$ & $0.2$ & $\phantom{00}10$ & $0.2$ & $\phantom{0}390$ & $0.5$ & $10$ & $0.2$ & $\phantom{00}10$ & $0.2$ & $\phantom{0}620$ & $0.5$ \\
\bottomrule 
\end{tabular*}
      \caption{Number of obtained $Q$ sets in separation, number of generated SECs when $k_{in}\times k_{out}$ is set to $1\times 1$ and $10\times 10$ and their running times by separation strategy, shrinking strategy  and OP instance generation in pr76.}\label{tab:pr76-3}
    \end{table}
    \vspace*{\fill}
  \end{landscape}

  \begin{landscape}
    \begin{center}\bf att532\end{center}
    \begin{center} Separation and SEC Generation\end{center}
    \vspace*{\fill}
    \begin{table}[htb!]
      \centering
      \tiny
\begin{tabular*}{\columnwidth}{@{\extracolsep{\fill}}llcccccccccccccccccc@{}}
\toprule
& & \multicolumn{6}{c}{Gen1} & \multicolumn{6}{c}{Gen2} & \multicolumn{6}{c}{Gen3} \\ \cmidrule(lr){3-8}\cmidrule(lr){9-14}\cmidrule(lr){15-20}
& & \multicolumn{2}{c}{Separation} & \multicolumn{4}{c}{SEC Generation} & \multicolumn{2}{c}{Separation} & \multicolumn{4}{c}{SEC Generation} & \multicolumn{2}{c}{Separation} & \multicolumn{4}{c}{SEC Generation} \\ \cmidrule(lr){3-4}\cmidrule(lr){5-8}\cmidrule(lr){9-10}\cmidrule(lr){11-14}\cmidrule(lr){15-16}\cmidrule(lr){17-20}
& & \multicolumn{2}{c}{(20 runs)} & \multicolumn{2}{c}{1x1 (10 runs)} & \multicolumn{2}{c}{10x10 (10 runs)} & \multicolumn{2}{c}{(20 runs)} & \multicolumn{2}{c}{1x1 (10 runs)} & \multicolumn{2}{c}{10x10 (10 runs)} & \multicolumn{2}{c}{(20 runs)} & \multicolumn{2}{c}{1x1 (10 runs)} & \multicolumn{2}{c}{10x10 (10 runs)} \\ \cmidrule(lr){3-4}\cmidrule(lr){5-6}\cmidrule(lr){7-8}\cmidrule(lr){9-10}\cmidrule(lr){11-12}\cmidrule(lr){13-14}\cmidrule(lr){15-16}\cmidrule(lr){17-18}\cmidrule(lr){19-20}
Sep. & Shrinking & \#Q & Time & \#SEC & Time & \#SEC & Time & \#Q & Time & \#SEC & Time & \#SEC & Time & \#Q & Time & \#SEC & Time & \#SEC & \multicolumn{1}{c}{Time} \\ 
\midrule
EH & NO  & $\phantom{0}94$ & $65.3$ & $\phantom{000}94$ & $67.8$ & $\phantom{0}8550$ & $67.4$ & $\phantom{0}32$ & $54.2$ & $\phantom{000}32$ & $53.9$ & $\phantom{0}2630$ & $56.0$ & $\phantom{0}30$ & $65.6$ & $\phantom{000}30$ & $65.9$ & $\phantom{0}2750$ & $67.4$ \\
 & C1  & $\phantom{0}43$ & $\phantom{0}5.4$ & $\phantom{000}43$ & $\phantom{0}6.1$ & $\phantom{0}3560$ & $\phantom{0}7.8$ & $\phantom{0}17$ & $17.0$ & $\phantom{000}17$ & $17.4$ & $\phantom{0}1470$ & $17.6$ & $\phantom{0}21$ & $20.4$ & $\phantom{000}21$ & $20.6$ & $\phantom{0}1850$ & $21.7$ \\
\bottomrule 
\end{tabular*}
      \caption{Number of obtained $Q$ sets in separation, number of generated SECs when $k_{in} \times k_{out}$ is set to $1 \times 1$ and $10 \times 10$ and their running times by separation strategy, shrinking strategy  and OP instance generation in att532.}\label{tab:att532-3}
    \end{table}
    \vspace*{\fill}
  \end{landscape}

  \begin{landscape}
    \begin{center}\bf vm1084\end{center}
    \begin{center} Separation and SEC Generation\end{center}
    \vspace*{\fill}
    \begin{table}[htb!]
      \centering
      \tiny
\begin{tabular*}{\columnwidth}{@{\extracolsep{\fill}}llcccccccccccccccccc@{}}
\toprule
& & \multicolumn{6}{c}{Gen1} & \multicolumn{6}{c}{Gen2} & \multicolumn{6}{c}{Gen3} \\ \cmidrule(lr){3-8}\cmidrule(lr){9-14}\cmidrule(lr){15-20}
& & \multicolumn{2}{c}{Separation} & \multicolumn{4}{c}{SEC Generation} & \multicolumn{2}{c}{Separation} & \multicolumn{4}{c}{SEC Generation} & \multicolumn{2}{c}{Separation} & \multicolumn{4}{c}{SEC Generation} \\ \cmidrule(lr){3-4}\cmidrule(lr){5-8}\cmidrule(lr){9-10}\cmidrule(lr){11-14}\cmidrule(lr){15-16}\cmidrule(lr){17-20}
& & \multicolumn{2}{c}{(20 runs)} & \multicolumn{2}{c}{1x1 (10 runs)} & \multicolumn{2}{c}{10x10 (10 runs)} & \multicolumn{2}{c}{(20 runs)} & \multicolumn{2}{c}{1x1 (10 runs)} & \multicolumn{2}{c}{10x10 (10 runs)} & \multicolumn{2}{c}{(20 runs)} & \multicolumn{2}{c}{1x1 (10 runs)} & \multicolumn{2}{c}{10x10 (10 runs)} \\ \cmidrule(lr){3-4}\cmidrule(lr){5-6}\cmidrule(lr){7-8}\cmidrule(lr){9-10}\cmidrule(lr){11-12}\cmidrule(lr){13-14}\cmidrule(lr){15-16}\cmidrule(lr){17-18}\cmidrule(lr){19-20}
Sep. & Shrinking & \#Q & Time & \#SEC & Time & \#SEC & Time & \#Q & Time & \#SEC & Time & \#SEC & Time & \#Q & Time & \#SEC & Time & \#SEC & \multicolumn{1}{c}{Time} \\ 
\midrule
EH & NO  & $102$ & $538.5$ & $\phantom{00}102$ & $539.8$ & $\phantom{0}8350$ & $545.3$ & $156$ & $211.2$ & $\phantom{00}156$ & $212.3$ & $14470$ & $224.9$ & $116$ & $291.9$ & $\phantom{00}116$ & $293.6$ & $10280$ & $303.0$ \\
 & C1  & $\phantom{0}40$ & $\phantom{0}42.9$ & $\phantom{000}40$ & $\phantom{0}44.3$ & $\phantom{0}3070$ & $\phantom{0}45.0$ & $\phantom{0}32$ & $\phantom{0}47.6$ & $\phantom{000}32$ & $\phantom{0}47.8$ & $\phantom{0}2540$ & $\phantom{0}49.8$ & $\phantom{0}33$ & $\phantom{0}36.5$ & $\phantom{000}33$ & $\phantom{0}36.9$ & $\phantom{0}2410$ & $\phantom{0}39.0$ \\
\bottomrule 
\end{tabular*}
      \caption{Number of obtained $Q$ sets in separation, number of generated SECs when $k_{in}\times k_{out}$ is set to $1\times 1$ and $10\times 10$ and their running times by separation strategy, shrinking strategy  and OP instance generation in vm1084.}\label{tab:vm1084-3}
    \end{table}
    \vspace*{\fill}
  \end{landscape}

  \begin{landscape}
    \begin{center}\bf rl1323\end{center}
    \begin{center} Separation and SEC Generation\end{center}
    \vspace*{\fill}
    \begin{table}[htb!]
      \centering
      \tiny
\begin{tabular*}{\columnwidth}{@{\extracolsep{\fill}}llcccccccccccccccccc@{}}
\toprule
& & \multicolumn{6}{c}{Gen1} & \multicolumn{6}{c}{Gen2} & \multicolumn{6}{c}{Gen3} \\ \cmidrule(lr){3-8}\cmidrule(lr){9-14}\cmidrule(lr){15-20}
& & \multicolumn{2}{c}{Separation} & \multicolumn{4}{c}{SEC Generation} & \multicolumn{2}{c}{Separation} & \multicolumn{4}{c}{SEC Generation} & \multicolumn{2}{c}{Separation} & \multicolumn{4}{c}{SEC Generation} \\ \cmidrule(lr){3-4}\cmidrule(lr){5-8}\cmidrule(lr){9-10}\cmidrule(lr){11-14}\cmidrule(lr){15-16}\cmidrule(lr){17-20}
& & \multicolumn{2}{c}{(20 runs)} & \multicolumn{2}{c}{1x1 (10 runs)} & \multicolumn{2}{c}{10x10 (10 runs)} & \multicolumn{2}{c}{(20 runs)} & \multicolumn{2}{c}{1x1 (10 runs)} & \multicolumn{2}{c}{10x10 (10 runs)} & \multicolumn{2}{c}{(20 runs)} & \multicolumn{2}{c}{1x1 (10 runs)} & \multicolumn{2}{c}{10x10 (10 runs)} \\ \cmidrule(lr){3-4}\cmidrule(lr){5-6}\cmidrule(lr){7-8}\cmidrule(lr){9-10}\cmidrule(lr){11-12}\cmidrule(lr){13-14}\cmidrule(lr){15-16}\cmidrule(lr){17-18}\cmidrule(lr){19-20}
Sep. & Shrinking & \#Q & Time & \#SEC & Time & \#SEC & Time & \#Q & Time & \#SEC & Time & \#SEC & Time & \#Q & Time & \#SEC & Time & \#SEC & \multicolumn{1}{c}{Time} \\ 
\midrule
EH & NO  & $100$ & $582.4$ & $\phantom{00}100$ & $582.4$ & $\phantom{0}8930$ & $589.9$ & $\phantom{0}27$ & $441.1$ & $\phantom{000}27$ & $442.9$ & $\phantom{0}1660$ & $441.2$ & $250$ & $288.0$ & $\phantom{00}250$ & $291.2$ & $20440$ & $307.4$ \\
 & C1  & $\phantom{0}28$ & $\phantom{0}83.7$ & $\phantom{000}28$ & $\phantom{0}83.6$ & $\phantom{0}2630$ & $\phantom{0}87.1$ & $\phantom{0}26$ & $\phantom{0}58.0$ & $\phantom{000}26$ & $\phantom{0}58.1$ & $\phantom{0}1710$ & $\phantom{0}59.9$ & $\phantom{0}62$ & $\phantom{0}50.2$ & $\phantom{000}62$ & $\phantom{0}51.3$ & $\phantom{0}4030$ & $\phantom{0}53.4$ \\
\bottomrule 
\end{tabular*}
      \caption{Number of obtained $Q$ sets in separation, number of generated SECs when $k_{in}\times k_{out}$ is set to $1\times 1$ and $10\times 10$ and their running times by separation strategy, shrinking strategy  and OP instance generation in rl1323.}\label{tab:rl1323-3}
    \end{table}
    \vspace*{\fill}
  \end{landscape}

  \begin{landscape}
    \begin{center}\bf vm1748\end{center}
    \begin{center} Separation and SEC Generation\end{center}
    \vspace*{\fill}
    \begin{table}[htb!]
      \centering
      \tiny
\begin{tabular*}{\columnwidth}{@{\extracolsep{\fill}}llcccccccccccccccccc@{}}
\toprule
& & \multicolumn{6}{c}{Gen1} & \multicolumn{6}{c}{Gen2} & \multicolumn{6}{c}{Gen3} \\ \cmidrule(lr){3-8}\cmidrule(lr){9-14}\cmidrule(lr){15-20}
& & \multicolumn{2}{c}{Separation} & \multicolumn{4}{c}{SEC Generation} & \multicolumn{2}{c}{Separation} & \multicolumn{4}{c}{SEC Generation} & \multicolumn{2}{c}{Separation} & \multicolumn{4}{c}{SEC Generation} \\ \cmidrule(lr){3-4}\cmidrule(lr){5-8}\cmidrule(lr){9-10}\cmidrule(lr){11-14}\cmidrule(lr){15-16}\cmidrule(lr){17-20}
& & \multicolumn{2}{c}{(20 runs)} & \multicolumn{2}{c}{1x1 (10 runs)} & \multicolumn{2}{c}{10x10 (10 runs)} & \multicolumn{2}{c}{(20 runs)} & \multicolumn{2}{c}{1x1 (10 runs)} & \multicolumn{2}{c}{10x10 (10 runs)} & \multicolumn{2}{c}{(20 runs)} & \multicolumn{2}{c}{1x1 (10 runs)} & \multicolumn{2}{c}{10x10 (10 runs)} \\ \cmidrule(lr){3-4}\cmidrule(lr){5-6}\cmidrule(lr){7-8}\cmidrule(lr){9-10}\cmidrule(lr){11-12}\cmidrule(lr){13-14}\cmidrule(lr){15-16}\cmidrule(lr){17-18}\cmidrule(lr){19-20}
Sep. & Shrinking & \#Q & Time & \#SEC & Time & \#SEC & Time & \#Q & Time & \#SEC & Time & \#SEC & Time & \#Q & Time & \#SEC & Time & \#SEC & \multicolumn{1}{c}{Time} \\ 
\midrule
EH & NO  & $\phantom{0}130$ & $2198.8$ & $\phantom{000}130$ & $2197.0$ & $\phantom{0}11200$ & $2222.3$ & $\phantom{00}40$ & $1068.1$ & $\phantom{0000}40$ & $1070.7$ & $\phantom{00}2970$ & $1071.0$ & $\phantom{0}188$ & $1140.2$ & $\phantom{000}188$ & $1144.0$ & $\phantom{0}15540$ & $1176.5$ \\
 & C1  & $\phantom{00}47$ & $\phantom{0}175.9$ & $\phantom{0000}47$ & $\phantom{0}176.9$ & $\phantom{00}3250$ & $\phantom{0}181.9$ & $\phantom{00}28$ & $\phantom{0}227.5$ & $\phantom{0000}28$ & $\phantom{0}228.1$ & $\phantom{00}1930$ & $\phantom{0}231.3$ & $\phantom{00}82$ & $\phantom{0}138.3$ & $\phantom{0000}82$ & $\phantom{0}140.6$ & $\phantom{00}5740$ & $\phantom{0}147.2$ \\
\bottomrule 
\end{tabular*}
      \caption{Number of obtained $Q$ sets in separation, number of generated SECs when $k_{in}\times k_{out}$ is set to $1\times 1$ and $10\times 10$ and their running times by separation strategy, shrinking strategy  and OP instance generation in vm1748.}\label{tab:vm1748-3}
    \end{table}
    \vspace*{\fill}
  \end{landscape}

  \begin{landscape}
    \begin{center}\bf rl5934\end{center}
    \begin{center} Separation and SEC Generation\end{center}
    \vspace*{\fill}
    \begin{table}[htb!]
      \centering
      \tiny
\begin{tabular*}{\columnwidth}{@{\extracolsep{\fill}}llcccccccccccccccccc@{}}
\toprule
& & \multicolumn{6}{c}{Gen1} & \multicolumn{6}{c}{Gen2} & \multicolumn{6}{c}{Gen3} \\ \cmidrule(lr){3-8}\cmidrule(lr){9-14}\cmidrule(lr){15-20}
& & \multicolumn{2}{c}{Separation} & \multicolumn{4}{c}{SEC Generation} & \multicolumn{2}{c}{Separation} & \multicolumn{4}{c}{SEC Generation} & \multicolumn{2}{c}{Separation} & \multicolumn{4}{c}{SEC Generation} \\ \cmidrule(lr){3-4}\cmidrule(lr){5-8}\cmidrule(lr){9-10}\cmidrule(lr){11-14}\cmidrule(lr){15-16}\cmidrule(lr){17-20}
& & \multicolumn{2}{c}{(20 runs)} & \multicolumn{2}{c}{1x1 (10 runs)} & \multicolumn{2}{c}{10x10 (10 runs)} & \multicolumn{2}{c}{(20 runs)} & \multicolumn{2}{c}{1x1 (10 runs)} & \multicolumn{2}{c}{10x10 (10 runs)} & \multicolumn{2}{c}{(20 runs)} & \multicolumn{2}{c}{1x1 (10 runs)} & \multicolumn{2}{c}{10x10 (10 runs)} \\ \cmidrule(lr){3-4}\cmidrule(lr){5-6}\cmidrule(lr){7-8}\cmidrule(lr){9-10}\cmidrule(lr){11-12}\cmidrule(lr){13-14}\cmidrule(lr){15-16}\cmidrule(lr){17-18}\cmidrule(lr){19-20}
Sep. & Shrinking & \#Q & Time & \#SEC & Time & \#SEC & Time & \#Q & Time & \#SEC & Time & \#SEC & Time & \#Q & Time & \#SEC & Time & \#SEC & \multicolumn{1}{c}{Time} \\ 
\midrule
EH & NO  & $\phantom{00}21$ & $\phantom{0}7923.7$ & $\phantom{0000}21$ & $\phantom{0}7914.3$ & $\phantom{00}1760$ & $\phantom{0}7942.1$ & $\phantom{00}60$ & $12227.3$ & $\phantom{0000}60$ & $12243.7$ & $\phantom{00}5060$ & $12233.7$ & $\phantom{0}307$ & $\phantom{0}7466.6$ & $\phantom{000}307$ & $\phantom{0}7486.2$ & $\phantom{0}28100$ & $\phantom{0}7530.7$ \\
 & C1  & $\phantom{00}19$ & $\phantom{00}991.6$ & $\phantom{0000}19$ & $\phantom{00}997.7$ & $\phantom{00}1630$ & $\phantom{00}995.7$ & $\phantom{00}27$ & $\phantom{00}700.9$ & $\phantom{0000}27$ & $\phantom{00}705.3$ & $\phantom{00}2410$ & $\phantom{00}707.9$ & $\phantom{00}57$ & $\phantom{00}877.0$ & $\phantom{0000}57$ & $\phantom{00}879.1$ & $\phantom{00}3650$ & $\phantom{00}888.1$ \\
\bottomrule 
\end{tabular*}
      \caption{Number of obtained $Q$ sets in separation, number of generated SECs when $k_{in}\times k_{out}$ is set to $1\times1$ and $10\times10$ and their running times by separation strategy, shrinking strategy  and OP instance generation in rl5934.}\label{tab:rl5934-3}
    \end{table}
    \vspace*{\fill}
  \end{landscape}

  \begin{landscape}
    \begin{center}\bf usa13509\end{center}
    \begin{center} Separation and SEC Generation\end{center}
    \vspace*{\fill}
    \begin{table}[htb!]
      \centering
      \tiny
\begin{tabular*}{\columnwidth}{@{\extracolsep{\fill}}llcccccccccccccccccc@{}}
\toprule
& & \multicolumn{6}{c}{Gen1} & \multicolumn{6}{c}{Gen2} & \multicolumn{6}{c}{Gen3} \\ \cmidrule(lr){3-8}\cmidrule(lr){9-14}\cmidrule(lr){15-20}
& & \multicolumn{2}{c}{Separation} & \multicolumn{4}{c}{SEC Generation} & \multicolumn{2}{c}{Separation} & \multicolumn{4}{c}{SEC Generation} & \multicolumn{2}{c}{Separation} & \multicolumn{4}{c}{SEC Generation} \\ \cmidrule(lr){3-4}\cmidrule(lr){5-8}\cmidrule(lr){9-10}\cmidrule(lr){11-14}\cmidrule(lr){15-16}\cmidrule(lr){17-20}
& & \multicolumn{2}{c}{(20 runs)} & \multicolumn{2}{c}{1x1 (10 runs)} & \multicolumn{2}{c}{10x10 (10 runs)} & \multicolumn{2}{c}{(20 runs)} & \multicolumn{2}{c}{1x1 (10 runs)} & \multicolumn{2}{c}{10x10 (10 runs)} & \multicolumn{2}{c}{(20 runs)} & \multicolumn{2}{c}{1x1 (10 runs)} & \multicolumn{2}{c}{10x10 (10 runs)} \\ \cmidrule(lr){3-4}\cmidrule(lr){5-6}\cmidrule(lr){7-8}\cmidrule(lr){9-10}\cmidrule(lr){11-12}\cmidrule(lr){13-14}\cmidrule(lr){15-16}\cmidrule(lr){17-18}\cmidrule(lr){19-20}
Sep. & Shrinking & \#Q & Time & \#SEC & Time & \#SEC & Time & \#Q & Time & \#SEC & Time & \#SEC & Time & \#Q & Time & \#SEC & Time & \#SEC & \multicolumn{1}{c}{Time} \\ 
\midrule
EH & NO  & $1162$ & $46201.1$ & $\phantom{00}1162$ & $46550.1$ & $108550$ & $46915.2$ & $1183$ & $22163.0$ & $\phantom{00}1183$ & $22528.4$ & $109110$ & $22903.0$ & $\phantom{0}676$ & $33770.8$ & $\phantom{000}676$ & $34109.3$ & $\phantom{0}60910$ & $34008.7$ \\
 & C1  & $\phantom{00}95$ & $\phantom{0}4798.9$ & $\phantom{0000}95$ & $\phantom{0}4827.7$ & $\phantom{00}7930$ & $\phantom{0}4857.0$ & $\phantom{0}229$ & $\phantom{0}1080.0$ & $\phantom{000}229$ & $\phantom{0}1138.9$ & $\phantom{0}19380$ & $\phantom{0}1221.3$ & $\phantom{0}180$ & $\phantom{0}2369.9$ & $\phantom{000}180$ & $\phantom{0}2410.5$ & $\phantom{0}15090$ & $\phantom{0}2467.2$ \\
\bottomrule 
\end{tabular*}
      \caption{Number of obtained $Q$ sets in separation, number of generated SECs when $k_{in}\times k_{out}$ is set to $1\times1$ and $10\times10$ and their running times by separation strategy, shrinking strategy  and OP instance generation in usa13509.}\label{tab:usa1359-3}
    \end{table}
    \vspace*{\fill}
  \end{landscape}

  \begin{landscape}
    \begin{center}\bf d15112\end{center}
    \begin{center} Separation and SEC Generation\end{center}
    \vspace*{\fill}
    \begin{table}[htb!]
      \centering
      \tiny
\begin{tabular*}{\columnwidth}{@{\extracolsep{\fill}}llcccccccccccccccccc@{}}
\toprule
& & \multicolumn{6}{c}{Gen1} & \multicolumn{6}{c}{Gen2} & \multicolumn{6}{c}{Gen3} \\ \cmidrule(lr){3-8}\cmidrule(lr){9-14}\cmidrule(lr){15-20}
& & \multicolumn{2}{c}{Separation} & \multicolumn{4}{c}{SEC Generation} & \multicolumn{2}{c}{Separation} & \multicolumn{4}{c}{SEC Generation} & \multicolumn{2}{c}{Separation} & \multicolumn{4}{c}{SEC Generation} \\ \cmidrule(lr){3-4}\cmidrule(lr){5-8}\cmidrule(lr){9-10}\cmidrule(lr){11-14}\cmidrule(lr){15-16}\cmidrule(lr){17-20}
& & \multicolumn{2}{c}{(20 runs)} & \multicolumn{2}{c}{1x1 (10 runs)} & \multicolumn{2}{c}{10x10 (10 runs)} & \multicolumn{2}{c}{(20 runs)} & \multicolumn{2}{c}{1x1 (10 runs)} & \multicolumn{2}{c}{10x10 (10 runs)} & \multicolumn{2}{c}{(20 runs)} & \multicolumn{2}{c}{1x1 (10 runs)} & \multicolumn{2}{c}{10x10 (10 runs)} \\ \cmidrule(lr){3-4}\cmidrule(lr){5-6}\cmidrule(lr){7-8}\cmidrule(lr){9-10}\cmidrule(lr){11-12}\cmidrule(lr){13-14}\cmidrule(lr){15-16}\cmidrule(lr){17-18}\cmidrule(lr){19-20}
Sep. & Shrinking & \#Q & Time & \#SEC & Time & \#SEC & Time & \#Q & Time & \#SEC & Time & \#SEC & Time & \#Q & Time & \#SEC & Time & \#SEC & \multicolumn{1}{c}{Time} \\ 
\midrule
EH & NO  & $3284$ & $42274.1$ & $\phantom{00}3284$ & $43874.3$ & $307840$ & $43266.4$ & $\phantom{0}128$ & $55626.9$ & $\phantom{000}128$ & $55295.8$ & $\phantom{0}12080$ & $56185.5$ & $\phantom{0}974$ & $86877.2$ & $\phantom{000}974$ & $87707.5$ & $\phantom{0}90920$ & $86895.5$ \\
 & C1  & $\phantom{0}513$ & $\phantom{0}1864.8$ & $\phantom{000}513$ & $\phantom{0}2010.1$ & $\phantom{0}44790$ & $\phantom{0}2152.1$ & $\phantom{0}447$ & $\phantom{00}926.5$ & $\phantom{000}447$ & $\phantom{0}1048.0$ & $\phantom{0}38160$ & $\phantom{0}1138.6$ & $\phantom{0}127$ & $\phantom{0}4010.1$ & $\phantom{000}127$ & $\phantom{0}4026.4$ & $\phantom{0}11000$ & $\phantom{0}4095.6$ \\
\bottomrule 
\end{tabular*}
      \caption{Number of obtained $Q$ sets in separation, number of generated SECs when $k_{in}\times k_{out}$ is set to $1\times1$ and $10\times10$ and their running times by separation strategy, shrinking strategy  and OP instance generation in d15112.}\label{tab:d15112-3}
    \end{table}
    \vspace*{\fill}
  \end{landscape}

  \clearpage
  \section{Figures}\label{appendix:figures}
  \setcounter{figure}{0}
  \renewcommand*{\thefigure}{C.\arabic{figure}}
  In this section, we show some shrunk graphs obtained from the proposed shrinking strategies. The goal is to help the reader to obtain insights about the alternative strategies. We focus on the pr76-Gen1 SEC instance to do so. For each strategy two figures are presented, one preserving the geometry of the original OP instance and other showing the topological representation. In the figures, the vertices and the edges with value 1 are represented in black. The vertices and the edges with value in $[0.5, 1)$ are represented in red. The vertices in white and the edges with dashed style represent those with value in $ (0, 0.5)$. The edges in blue and double lined style represent those with value greater than 1. The depot vertex of the OP, the vertex 1, is colored in green.

  \clearpage
  \begin{figure}[htb!]
    \centering
    \includegraphics[height=9.0cm]{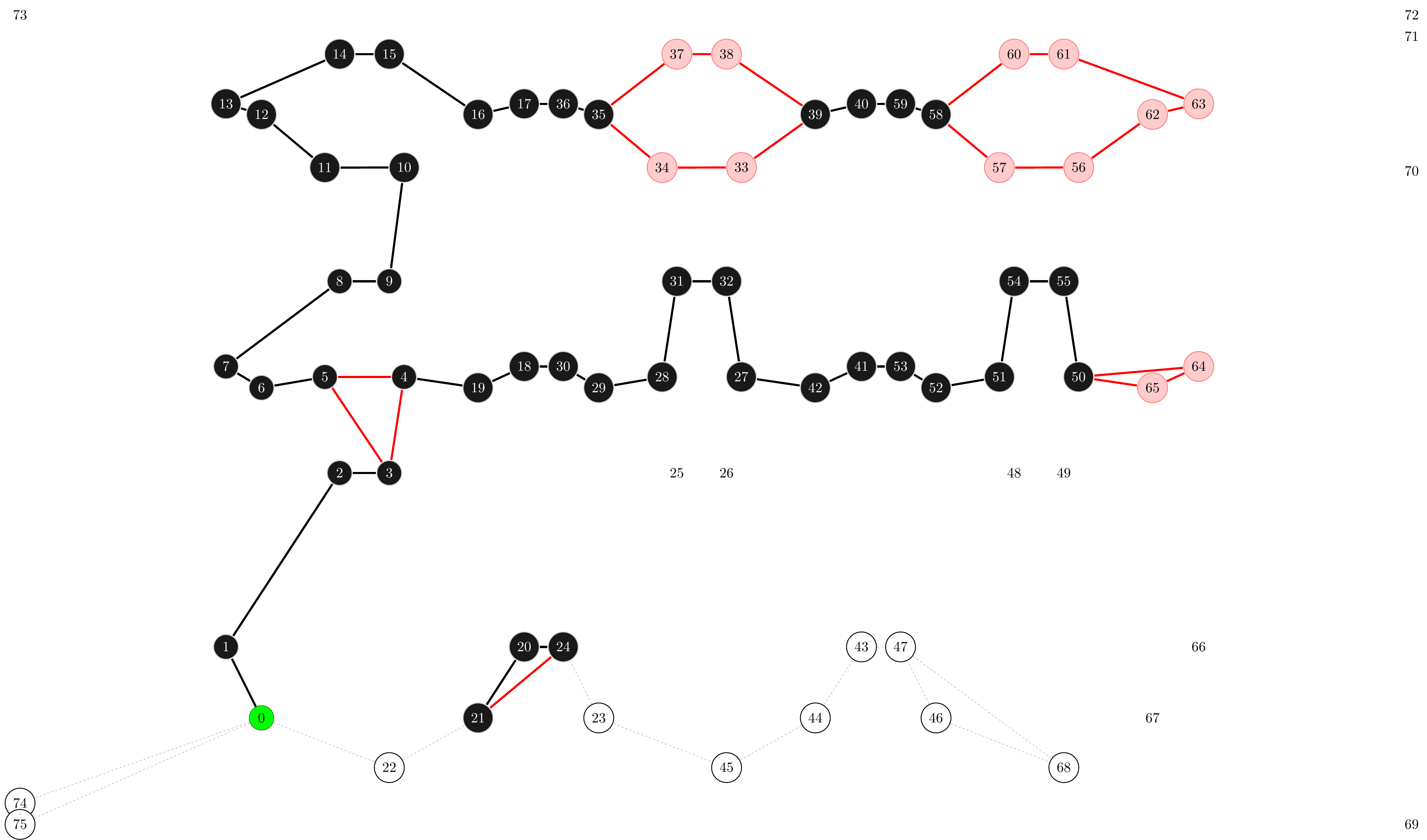}
    \caption{Support graph of pr76-gen1 SEC instance }\label{fig:pr76-no}
  \end{figure}

  \begin{figure}[htb!]
    \centering
    \includegraphics[height=9.0cm]{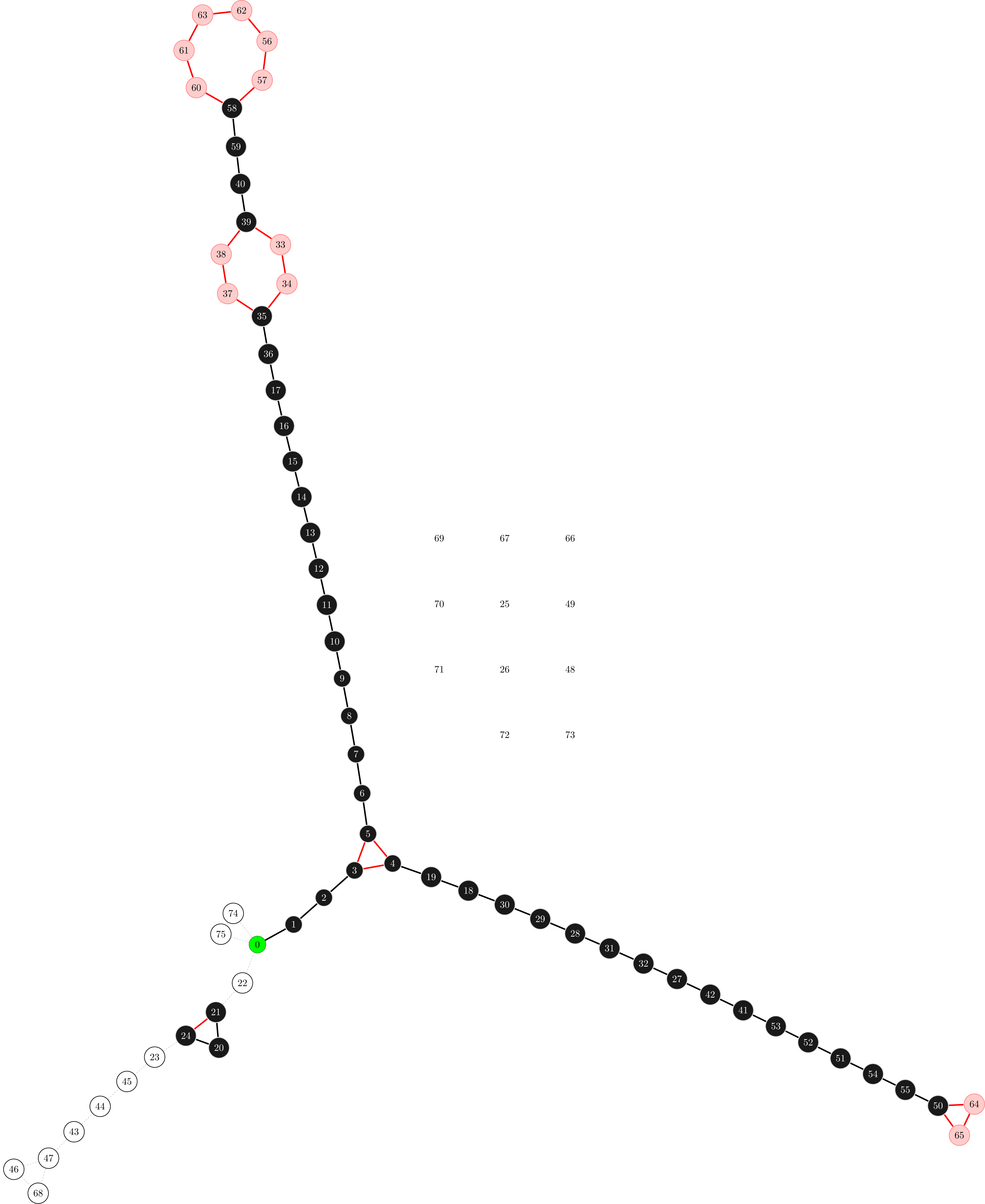}
    \caption{Topological representation of the pr76-gen1 SEC instance}\label{fig:pr76-no-topo}
  \end{figure}

  \clearpage
  \begin{figure}[htb!]
    \centering
    \includegraphics[height=9.0cm]{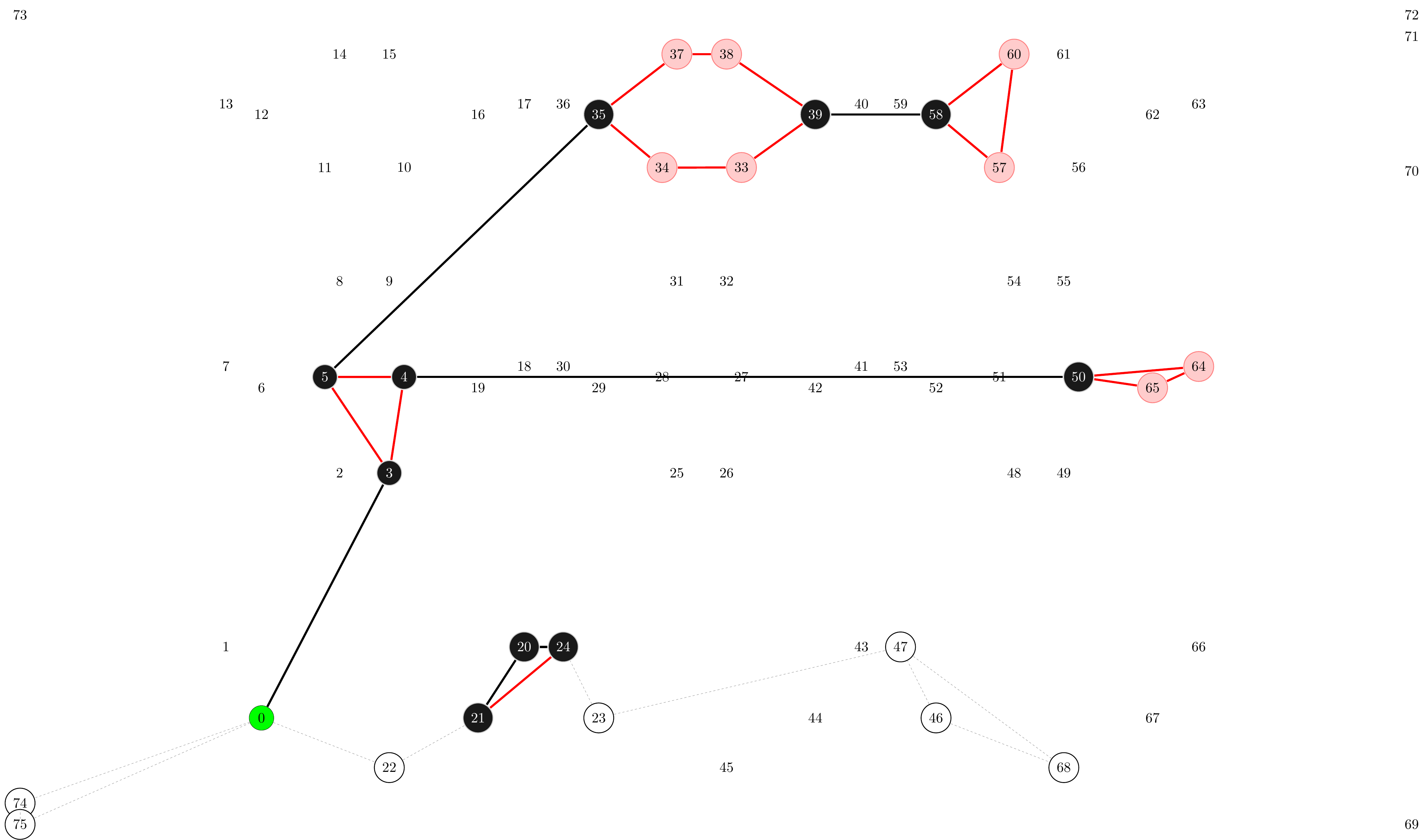}
    \caption{Resulting graph after C1 shrinking strategy}\label{fig:pr76-c1}
  \end{figure}

  \begin{figure}[htb!]
    \centering
    \includegraphics[height=9.0cm]{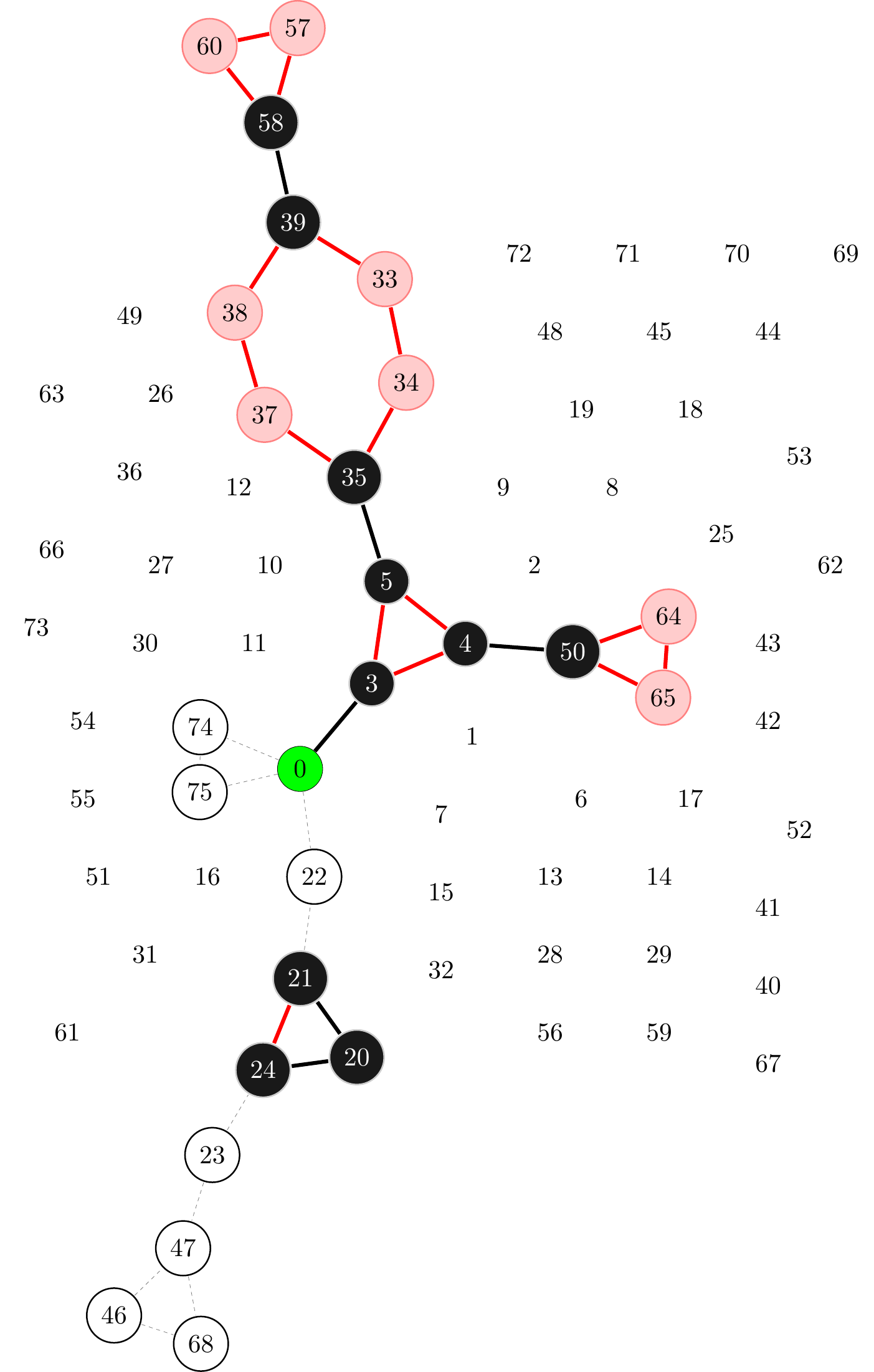}
    \caption{Topological representation of the graph after C1 shrinking strategy}\label{fig:pr76-c1-topo}
  \end{figure}

  \clearpage
  \begin{figure}[htb!]
    \centering
    \includegraphics[height=9.0cm]{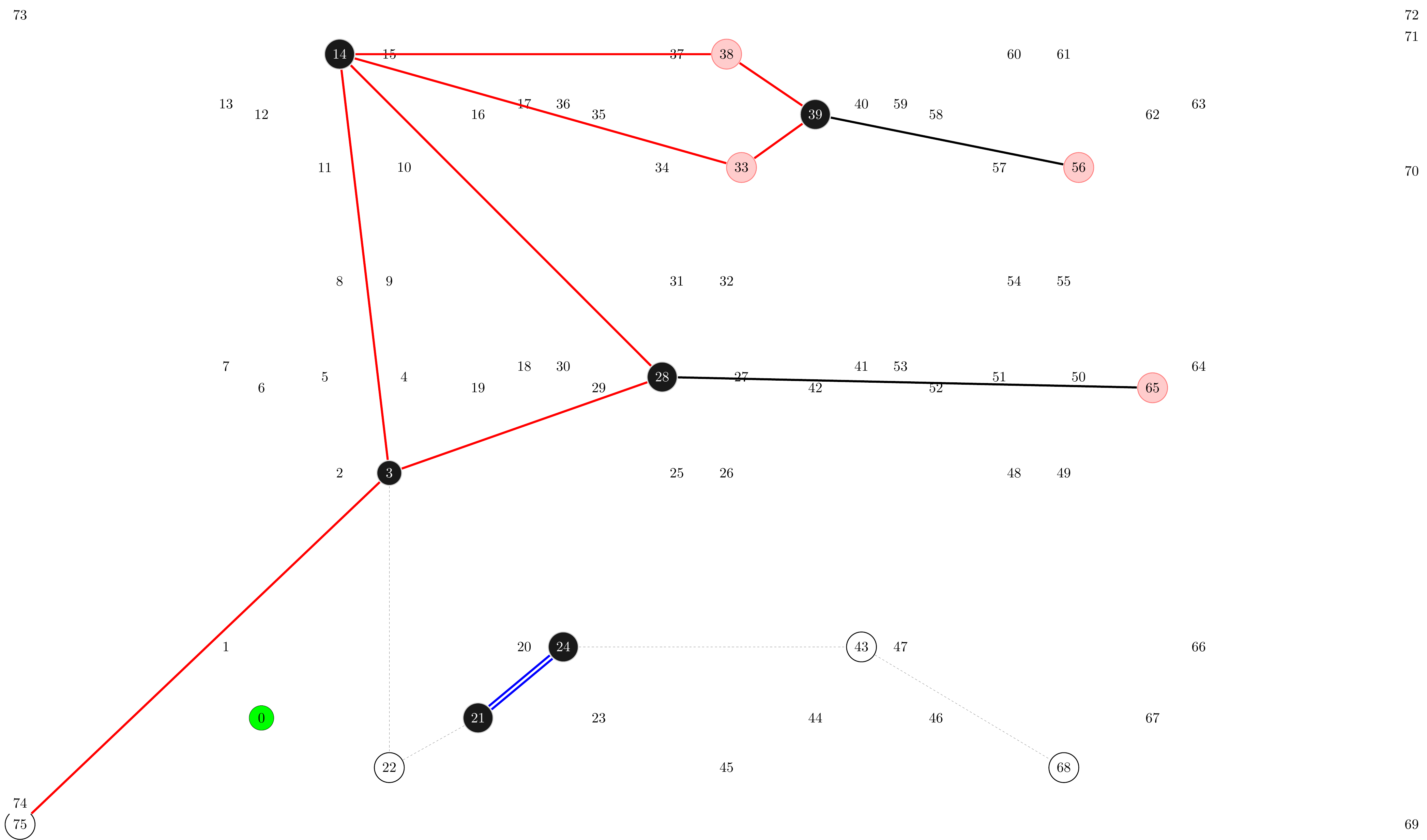}
    \caption{Resulting graph after S1 shrinking strategy}\label{fig:pr76-s1}
  \end{figure}

  \begin{figure}[htb!]
    \centering
    \includegraphics[height=9.0cm]{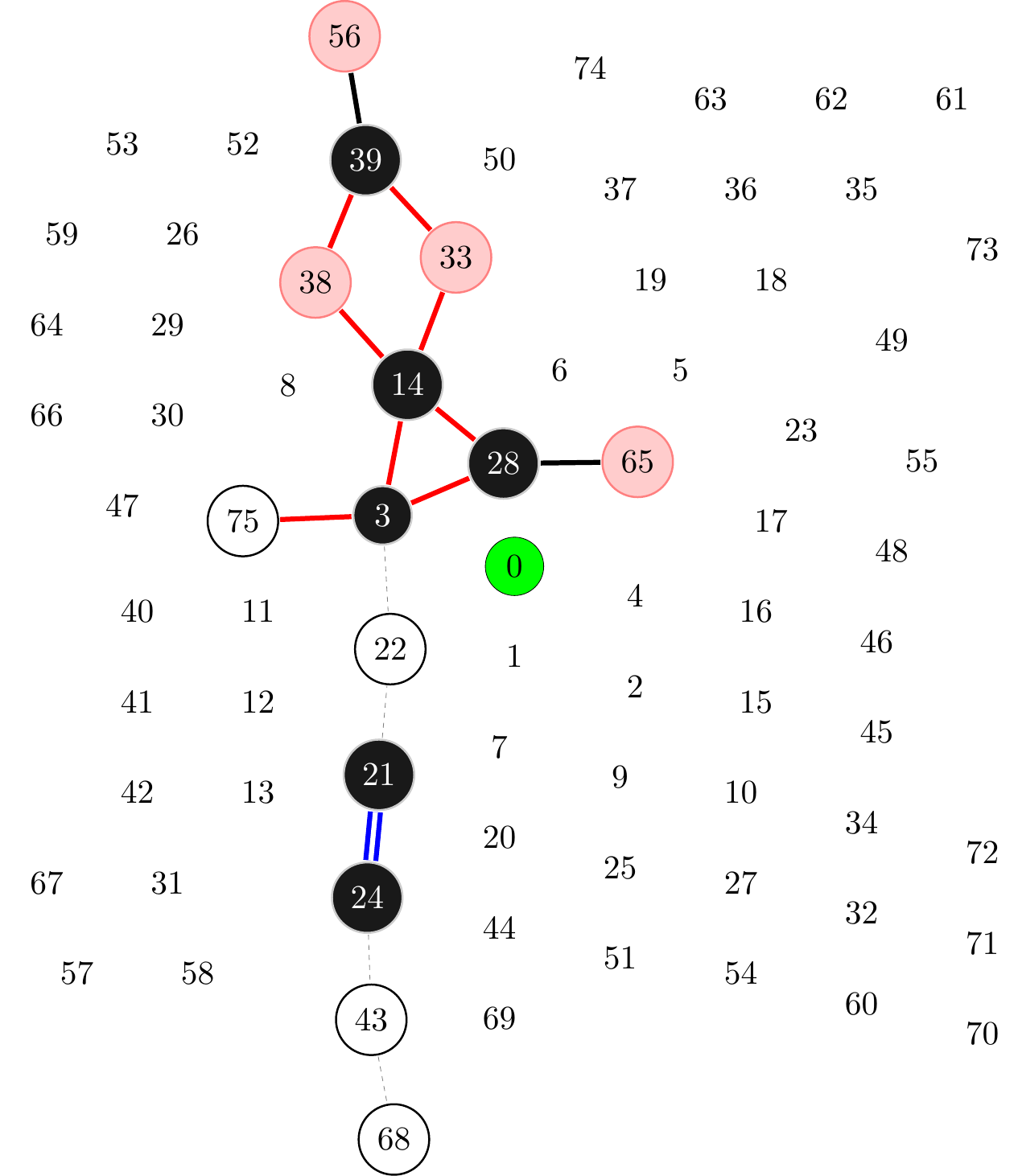}
    \caption{Topological representation of the graph after S1 shrinking strategy}\label{fig:pr76-s1-topo}
  \end{figure}

  \clearpage
  \begin{figure}[htb!]
    \centering
    \includegraphics[height=9.0cm]{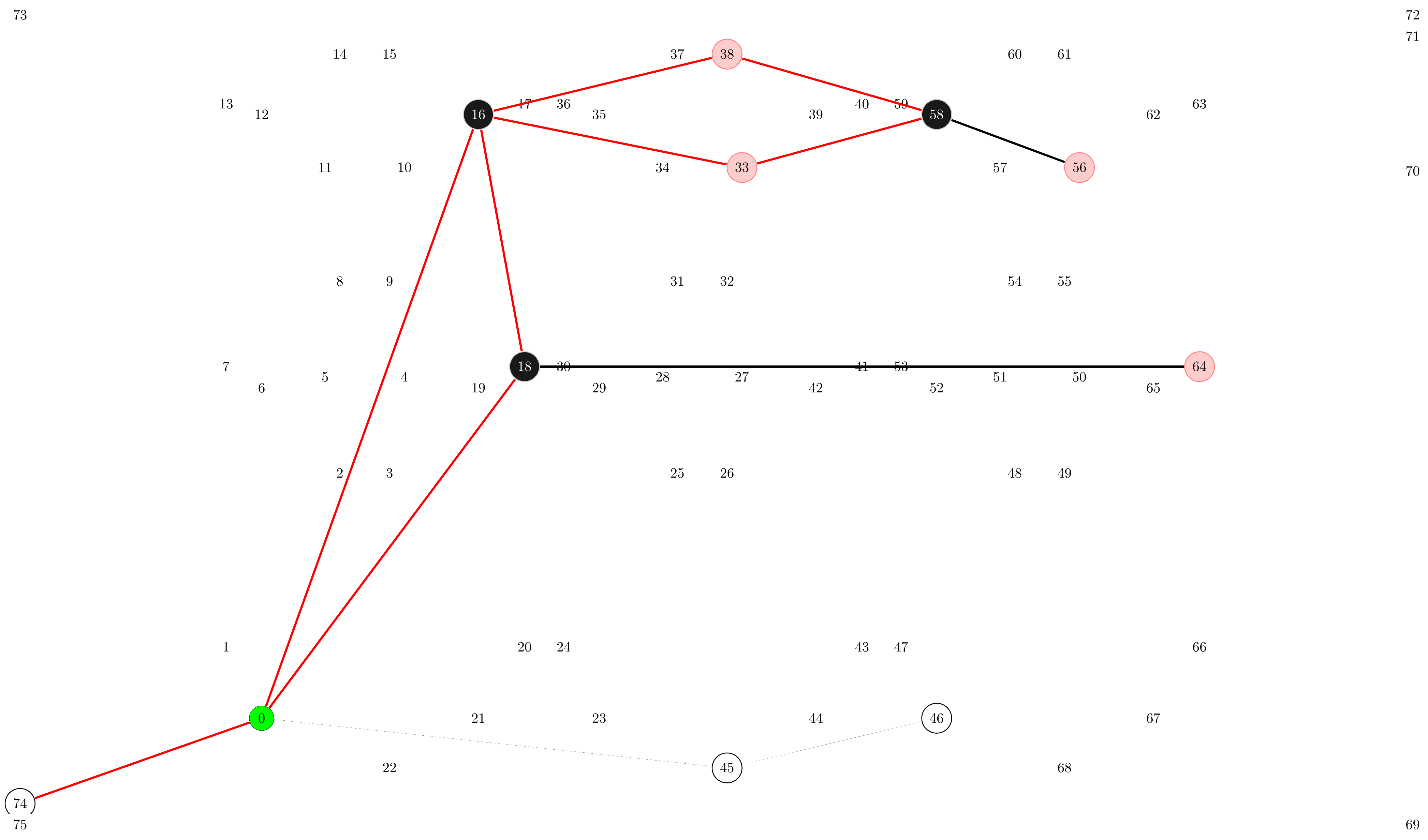}
    \caption{Resulting  graph after S1S2 shrinking strategy}\label{fig:pr76-s1s2}
  \end{figure}

  \begin{figure}[htb!]
    \centering
    \includegraphics[height=9.0cm]{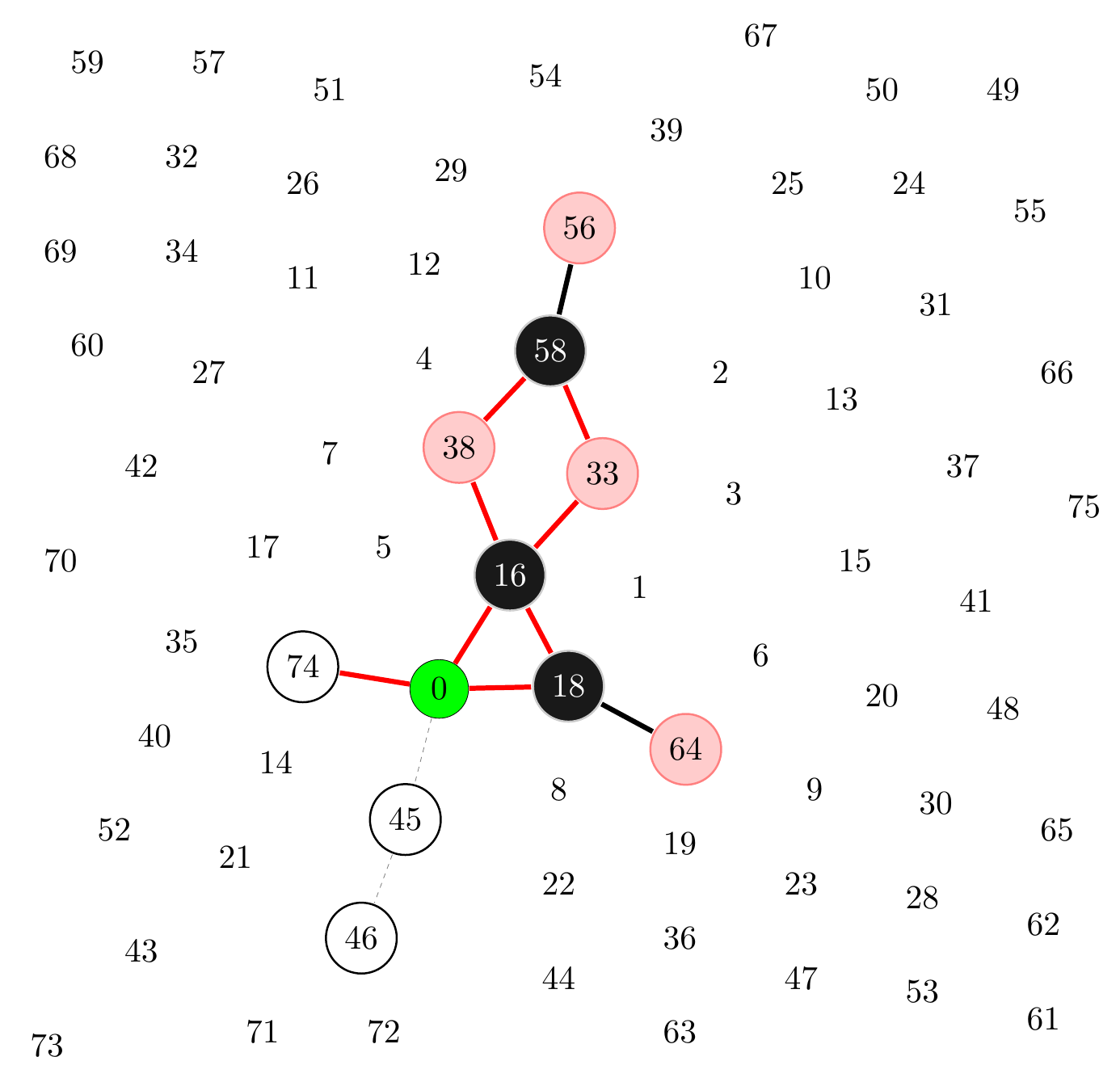}
    \caption{Topological representation of the  graph after S1S2 shrinking strategy}\label{fig:pr76-s1s2-topo}
  \end{figure}

  \end{document}